\documentclass{article}

\usepackage{setspace}
\usepackage[T1]{fontenc} 

\usepackage{graphicx}
\usepackage{amsthm,amsmath,amssymb,enumitem,mathrsfs}
\usepackage{color}
\usepackage[noadjust]{cite}
\usepackage{multirow}
\usepackage{tabu,array,soul}
\usepackage{subcaption}
\usepackage{caption}
\usepackage[ruled,vlined,linesnumbered]{algorithm2e}
\usepackage{siunitx}
\usepackage{textcomp}
\usepackage{tikz}
\usetikzlibrary{decorations.pathreplacing, calc}
\usetikzlibrary{positioning}
\usetikzlibrary {arrows.meta}
\usepackage{rotating}
\usepackage{amsbsy}

\usepackage{hyperref}
\hypersetup{
  colorlinks = true,
  linkcolor = black!30!red,
  citecolor = black!30!green
}

\usepackage[textsize=footnotesize,color=green!40]{todonotes}

\newtheorem{theorem}{Theorem}
\newtheorem{definition}[theorem]{Definition}

\newtheorem{lemma}[theorem]{Lemma}

\newtheorem{claim}[theorem]{Claim}

\newcommand{\overbar}[1]{\mkern 1.5mu\overline{\mkern-1.5mu#1\mkern-1.5mu}\mkern 1.5mu}

\newcommand{\set}[1]{\left\{#1\right\}}
\newcommand{\interval}[1]{\left[#1\right]}
\newcommand{\abs}[1]{\left|#1\right|}
\newcommand{\parentesis}[1]{\left(#1\right)}

\newcommand{\TikzInterval}[2]{\draw #1 -- #2;\draw #1 -- ++(0,0.2);\draw #1 -- ++(0,-0.2);\draw #2 -- ++(0,0.2);\draw #2 -- ++(0,-0.2);}
\newcommand{\TikzIntSep}[1]{\draw #1 -- ++(0,0.2);\draw #1 -- ++(0,-0.2);}
\newcommand{\TikzSimplicial}[1]{\draw #1 node[fill] {};}

\newcommand{\Interval}[1]{\boldsymbol{#1}}

\newcommand{\Root}[1]{R(#1)}

\newcommand{\Implication}[2]{IMP\left[\neg \Interval{#1} \rightarrow \Interval{#2}\right]}

\newcommand{\startInterval}{{\Interval{o}}}
\newcommand{\trueInterval}{{\Interval{\top}}}
\newcommand{\trueVertex}{\top}
\newcommand{\startVertex}{o}

\newcommand{\CGa}[1]{\Interval{a_{#1}}}
\newcommand{\CGy}[1]{\Interval{y_{#1}}}

\newcommand{\CGb}[1]{\Interval{b_{#1}}}
\newcommand{\CGc}[1]{\Interval{c_{#1}}}
\newcommand{\CGap}[1]{\Interval{a'_{#1}}}
\newcommand{\CGbp}[1]{\Interval{b'_{#1}}}
\newcommand{\CGcp}[1]{\Interval{c'_{#1}}}

\newcommand{\CGd}[1]{\Interval{d_{#1}}}
\newcommand{\CGf}[1]{\Interval{f_{#1}}}
\newcommand{\CGCov}[1]{\Interval{cov_{#1}}}

\newcommand{\INS}[2]{\textsc{INS}\left[\Interval{#1} , \Interval{#2}\right]}
\newcommand{\InsertSigma}[2]{\Interval{\sigma(#1,#2)}}

\newcommand{\alphaAND}[2]{\Interval{\alpha\left(#1,#2\right)}}
\newcommand{\betaAND}[2]{\Interval{\beta\left(#1,#2\right)}}
\newcommand{\gammaAND}[2]{\Interval{\gamma\left(#1,#2\right)}}
\newcommand{\deltaAND}[2]{\Interval{\delta\left(#1,#2\right)}}

\newcommand{\CovGadget}[1]{\textsc{COV}[#1]}
\newcommand{\AND}[2]{\textsc{AND}\left[\Interval{#1} , \Interval{#2}\right]}

\newcommand{\VarGadget}{\mathscr{X}}
\newcommand{\StartGadget}{\mathscr{S}}
\newcommand{\ClauseGadget}{\mathscr{C}}
\newcommand{\EndGadget}{\mathscr{E}}

\newcommand{\Tail}[1]{\Interval{e_{#1}}}

\newcommand{\Compatible}[1]{\mathcal{C}_{#1}}

\newcommand{\auxGraph}[2]{H^{#2}_{#1}}

\newcommand{\Svertices}[2]{S^{#2}_{#1}}

\newcommand{\typpe}{\tau}

\newcommand{\typext}{\tau^{ext}}

\newcommand{\typent}{\tau^{int}}

\newcommand{\typecov}{\tau^{cov}}
\newcommand{\typeuncov}{\tau^{uncov}}
\newcommand{\typebag}{\tau^{bag}}

\newcommand{\solution}[2] {sol\left[#1 , #2\right] }

\newcommand{\x}[1]{x_{#1}}

\newcommand{\bound}{4 + 7n + 58m}
\newcommand{\boundWithoutSimplicial}{n + 6m}
\newcommand{\nbSimplicial}{4+ 6n + 52m}
\newcommand{\ntracks}{2+4n+35m}

\title{Algorithms and complexity for geodetic sets on interval and chordal graphs\thanks{This research was supported by the IFCAM project ``Applications of graph homomorphisms'' (MA/IFCAM/18/39), the ERC grant titled PARAPATH, the IDEX-ISITE initiative CAP 20-25 (ANR-16-IDEX-0001), the International Research Center ``Innovation Transportation and Production Systems'' of the I-SITE CAP 20-25, and the ANR project GRALMECO (ANR-21-CE48-0004). This work was partially carried out while Harmender Gahlawat was employed at G-SCOP, Grenoble-INP and Université Clermont Auvergne and Ben-Gurion University of the Negev.\\ This paper contains the full versions of parts of an extended abstract from the proceedings of ISAAC 2020~\cite{ISAAC-version}. The other results of~\cite{ISAAC-version} were published in a journal~\cite{chakraborty2023algorithms}, but are completely disjoint from the results in the present paper.}}

\author{Dibyayan Chakraborty\thanks{School of Computing, University of Leeds, United Kingdom.}  \and Sandip Das\thanks{Indian Statistical Institute, Kolkata, India.} \and Florent Foucaud \thanks{Université Clermont Auvergne, CNRS, Clermont Auvergne INP, Mines Saint-Etienne, LIMOS, 63000 Clermont-Ferrand, France.} \and Harmender Gahlawat \thanks{Department of Mathematics, Indian Institute of Technology Delhi, Hauz Khas, New Delhi 110016, India} \and Dimitri Lajou \thanks{Univ. Bordeaux, Bordeaux INP, CNRS, LaBRI, UMR5800, F-33400 Talence, France}}

\begin{document}

\maketitle

\begin{abstract}
We study the computational complexity of finding the \emph{geodetic number} of a graph on chordal graphs and interval graphs. A set $S$ of vertices of a graph $G$ is a \emph{geodetic set} if every vertex of $G$ lies in a shortest path between some pair of vertices of $S$. The \textsc{Minimum Geodetic Set (MGS)} problem is to find a geodetic set with minimum cardinality of a given graph. We  show that \textsc{Minimum Geodetic Set} is fixed parameter tractable for chordal graphs when parameterized by its \emph{tree-width} (which equals its clique number). This implies a polynomial-time algorithm for $k$-trees, for fixed $k$. Then, we show that \textsc{Minimum Geodetic Set} is NP-hard on interval graphs, thereby answering a question of Ekim et al. (LATIN, 2012), who showed that \textsc{Minimum Geodetic Set} is polynomial-time solvable on proper interval graphs. As interval graphs are very constrained, to prove the latter result, we design a rather sophisticated reduction technique to work around their inherent linear structure.
\end{abstract}

\section{Introduction}

A simple undirected graph $G$ has vertex set $V(G)$ and edge set $E(G)$. For two vertices $u,v\in V(G)$, let $I(u,v)$ denote the set of all vertices in $G$ that lie in some shortest path between $u$ and $v$.  For a subset $S$ of vertices of a graph $G$, let $I(S)=\bigcup_{u,v\in S} I(u,v) $. We say that $T\subseteq V(G)$ is \emph{covered} by $S$ if $T\subseteq I(S)$. A set of vertices $S$ is a \emph{geodetic set} if $V(G)$ is covered by $S$. The \emph{geodetic number} is the minimum integer $k$ such that $G$ has a geodetic set of cardinality $k$. Given a graph $G$, the \textsc{Minimum Geodetic Set (MGS)} problem, introduced in~\cite{harary1993}, is to compute a geodetic set of $G$ with minimum cardinality. In this paper, we study the computational complexity of \textsc{Minimum Geodetic Set} on interval and chordal graphs. \textsc{Minimum Geodetic Set} is a natural graph covering problem that falls in the class of problems dealing with the important geometric notion of \emph{convexity}: see~\cite{BOOKaraujo2025,farber1986,bookGC} for some general discussion on graph convexities. The setting of \textsc{Minimum Geodetic Set} is quite natural, and it can be applied to facility location problems such as the optimal determination of bus routes in a public transport network~\cite{caldam2020,bus}. See also~\cite{ekim2012} for further applications. The aim of this paper is to study \textsc{Minimum Geodetic Set} on chordal graphs. Chordal graphs are the graphs with no induced cycle of order greater than~$3$; equivalently, they are the intersection graphs of subtrees of trees. Their structural properties imply an interesting behaviour for various types of convexity, as pointed out in~\cite{JCMCC96,farber1986}.

The algorithmic complexity of \textsc{Minimum Geodetic Set} has been studied actively. \textsc{Minimum Geodetic Set} is known to be NP-hard on chordal graphs~\cite{JCMCC96}, chordal bipartite graphs~\cite{dourado2008,dourado2010}, subcubic graphs~\cite{bueno2018}, planar graphs~\cite{caldam2020}, partial grids~\cite{ISAAC-version,chakraborty2023algorithms}, co-bipartite graphs~\cite{ekim2012}, line graphs~\cite{ISAAC-version}, and graphs of diameter~2~\cite{ISAAC-version}. 
On the positive side, polynomial-time algorithms to solve \textsc{Minimum Geodetic Set} are known for cographs~\cite{dourado2010}, split graphs~\cite{dourado2010,JCMCC96} and its superclass well-partitioned chordal graphs~\cite{wellpart}, ptolemaic graphs~\cite{farber1986} and more generally distance-hereditary graphs~\cite{dh}, block-cactus graphs~\cite{ekim2012}, outerplanar graphs~\cite{mezzini2018}, solid grid graphs~\cite{ISAAC-version,chakraborty2023algorithms}, and proper interval graphs~\cite{ekim2012}. From the perspective of parameterized complexity, \textsc{Minimum Geodetic Set} is unlikely to be FPT (it is W[1]-hard) for the parameters solution size, feedback vertex set number, and pathwidth, combined~\cite{kellerhals2020}. In fact it is even NP-hard for graphs that have pathwidth~17 and feedback vertex set number~13~\cite{DBLP:conf/iwpec/Tale25}. The problem is FPT for each of the parameters tree-depth, modular-width and feedback edge set number~\cite{kellerhals2020}, and also for treewidth and diameter, combined~\cite{foucaud2023tight}. The approximability of \textsc{Minimum Geodetic Set} has also been studied in~\cite{caldam2020,davot2021approximation}, and fine-grained complexity aspects with respect to solution size and structural parameters have been studied in~\cite{foucaud2023tight,DBLP:conf/stacs/FoucaudGK0IST25}.

\medskip \noindent \textbf{Our results.} To complement the hardness of \textsc{Minimum Geodetic Set} on chordal graphs, in this paper, we design an FPT algorithm for \textsc{Minimum Geodetic Set} on chordal graphs when parameterized by its treewidth, which equals its clique number $\omega$ minus one. 
We use dynamic programming on tree-decompositions to prove the following. 

\begin{theorem}\label{thm:fpt-chordal}\label{thm:fpt-interval}
	\textsc{Minimum Geodetic Set} can be solved in time $2^{2^{O(w)}} n$ for chordal graphs and in time $2^{O(w)}n$ for interval graphs, where $n$ is the order and $w$ the treewidth of the input graph.
\end{theorem}

This result applies to the following setting. A \emph{$k$-tree} is a graph formed by starting with a complete graph on $(k + 1)$ vertices and then repeatedly adding vertices by making each added vertex adjacent to exactly $k$ neighbours forming a $(k+1)$-clique. Allgeier~\cite{allgeier} gave a polynomial-time algorithm to solve \textsc{Minimum Geodetic Set} on maximal outerplanar graphs, which is a subclass of $2$-trees, and thus our algorithm generalizes this result (note that $2$-trees are both chordal and planar), as it shows that \textsc{Minimum Geodetic Set} can be solved in time $2^{2^{O(k)}}n$ for $k$-trees of order $n$. Recall that this is unlikely to hold for \emph{partial} $k$-trees (which are exactly the graphs of treewidth at most $k$) since \textsc{Minimum Geodetic Set} is NP-hard for graphs of treewidth~14~\cite{DBLP:conf/iwpec/Tale25}. Recently, it has been proved that \textsc{Minimum Geodetic Set} does not admit a $2^{2^{o(w)}}\cdot n^{O(1)}$ algorithm (where $w$ is the treewidth) even on bounded diameter graphs~\cite{foucaud2023tight}.

In this paper, we further strengthen the existing NP-hardness result of \textsc{Minimum Geodetic Set} on chordal graphs by proving it to be NP-hard on \emph{interval} graphs. An \emph{interval representation} of a graph $G$ is a collection of intervals on the real line such that two intervals intersect if and only if the corresponding vertices are adjacent in $G$. A graph is an \emph{interval graph} if it has an interval representation. Ekim et al.~\cite{ekim2012} asked if there is a polynomial-time algorithm for \textsc{Minimum Geodetic Set} on interval graphs. We give a negative answer to their question in the following theorem (note that proper interval graphs are the interval graphs with no induced $K_{1,3}$).

\begin{theorem}\label{thm:interval-hard}
	\textsc{Minimum Geodetic Set} is NP-hard for interval graphs (even with no induced $K_{1,5}$).
\end{theorem}

This result is somewhat surprising, as most covering problems can be solved in polynomial time on interval graphs (but other distance-based problems, like \textsc{Metric Dimension}, are NP-complete for interval graphs~\cite{FoucaudMNPV17}). Our reduction (from \textsc{$3$-Sat}) uses a quite involved novel technique, that we hope can be used to prove similar results for other distance-related problems on interval graphs. The main challenge here is to overcome the linear structure of the graph to transmit information across the graph. To this end, we use a sophisticated construction of many parallel \emph{tracks}, \textit{i.e.}, shortest paths with intervals of (mostly) the same length spanning roughly the whole graph, and such that each track is shifted with respect to the previous one. Each track  represents shortest paths that will be used by solution vertices from our variable and clause gadgets. In between the tracks, we are able to build our gadgets.

We remark that \textsc{Minimum Geodetic Set} admits a polynomial-time algorithm on proper interval graphs by a nontrivial dynamic programming scheme~\cite{ekim2012}. Problems known to be NP-complete on interval graphs but polynomial-time solvable on proper interval graphs are quite rare; two examples known to us are \textsc{Equitable Coloring}~\cite{equitable-coloring} and \textsc{Induced Subgraph Isomorphism}~\cite{heg2015}. Theorem~\ref{thm:interval-hard} together with the algorithm from~\cite{ekim2012} adds \textsc{Minimum Geodetic Set} to this list. The state-of-the-art complexity status of \textsc{Minimum Geodetic Set} for various subclasses of chordal graphs is depicted in Figure~\ref{fig:diagram}.


\medskip \noindent \textbf{Structure of the paper.} In Section~ \ref{sec:chordal}, we describe the fixed parameter tractable algorithm for chordal graphs. In Section~\ref{sec:interval-hard}, we prove hardness for interval graphs. We conclude in Section~\ref{sec:conclusion}.

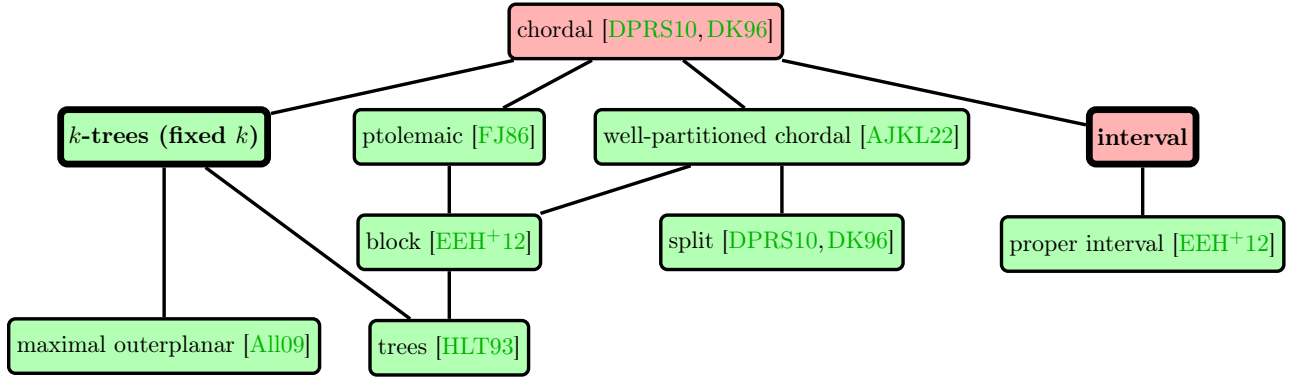
\begin{figure}[!tpb]
\centering
\scalebox{0.9}{\begin{tikzpicture}[node distance=7mm]

\tikzstyle{mybox}=[fill=white,line width=0.5mm,rectangle, minimum height=.8cm,fill=white!70,rounded corners=1mm,draw];
\tikzstyle{myedge}=[line width=0.5mm]
\newcommand{\tworows}[2]{\begin{tabular}{c}{#1}\\{#2}\end{tabular}}

    \node[mybox, fill=red!30] (chordal) [] {chordal \cite{dourado2010,JCMCC96}};
    \node[mybox, fill=green!30] (wp) [below =of chordal,xshift=20mm] {well-partitioned chordal \cite{wellpart}} edge[myedge] (chordal);
    \node[mybox, fill=green!30] (split) [below =of wp] {split \cite{dourado2010,JCMCC96}} edge[myedge] (wp);
    \node[mybox, fill=red!30,line width=1mm] (interval) [right=of wp,xshift=10mm] {\textbf{interval}} edge[myedge] (chordal);
    \node[mybox, fill=green!30] (proper) [below=of interval] {proper interval \cite{ekim2012}} edge[myedge] (interval);
    \node[mybox, fill=green!30] (ptolemaic) [left=of wp] {ptolemaic \cite{farber1986}} edge[myedge] (chordal);
    \node[mybox, fill=green!30,line width=1mm,xshift=-5mm] (ktree) [left=of ptolemaic] {\textbf{$k$-trees (fixed $k$)}} edge[myedge] (chordal);
    \node[mybox, fill=green!30] (block) [below=of ptolemaic]
    {block \cite{ekim2012}} edge[myedge] (ptolemaic) edge[myedge] (wp);
    \node[mybox, fill=green!30] (mop) [below=of ktree,yshift=-15mm]
    {maximal outerplanar \cite{allgeier}} edge[myedge] (ktree);    
    
    \node[mybox, fill=green!30] (tree) [below=of block] {trees \cite {harary1993}} edge[myedge] (block)  edge[myedge] (ktree);
  \end{tikzpicture}}

\caption{Inclusion diagram for subclasses of chordal graphs. If a class $A$ has an upward path to class $B$, then $A$ is included in $B$. For graphs in the green classes, \textsc{Minimum Geodetic Set} is polynomial-time solvable; for graphs in the red classes, it is NP-complete. The results from the two bold boxes are proved in this paper.}
\label{fig:diagram}
\end{figure}

\section{Preliminaries}\label{sec:prelim}
Let $G$ be a graph. Throughout this paper, unless otherwise specified, all graphs are assumed to be connected. This is without loss of generality, as the geodetic number of a disconnected graph is simply the sum of the geodetic numbers of its connected components.  For a subset $S\subseteq V(G)$, let $G[S]$ denote the subgraph of $G$ induced by the vertex set $S$. Further, let $G-S$ denote the subgraph $G[V(G)\setminus S]$. To ease the presentation, for a vertex $v\in V(G)$ and a subgraph $H$ of $G$, we denote $G-\{v\}$ by $G-v$ and $G-V(H)$ by $G-H$. For two vertices $u,v\in V(G)$, let $d(u,v)$ denote the length of a shortest path between $u$ and $v$. For a vertex $v\in V(G)$, let $N(v) = \{u~|~uv\in E(G)\}$ and let $N[v] = N(v)\cup \{v\}$. A subgraph $H$ of $G$ is a \textit{clique} if every $u,v\in V(H)$, $uv \in E(G)$. The maximum size of a clique in $G$ is defined as the \textit{clique number} of $G$, denoted by $\omega(G)$. A vertex set $X\subseteq V(G)$ is a \textit{clique cutset} of $G$ if $G[X]$ is a clique and $G-X$ is disconnected.

A graph is \textit{chordal} if every cycle of length at least $4$ has a chord. A vertex $v\in V(G)$ is \textit{simplicial} if $N[v]$ induces a clique in $G$. It is well known that every non-trivial chordal graph has at least two simplicial vertices. It is easy to observe that if $v$ is a simplicial vertex of $G$, then every geodetic set of $G$ contains $v$ as for any $u,w$ distinct from $v$, $v\notin I(u,w)$.

\begin{definition}[{\bf Treewidth}]\label{D:tw}
 A \emph{tree decomposition} of a graph $G$ is a pair $(T,\beta)$, where $T$ is a tree and $\beta$ is a mapping from $V(T)$  to subsets of $V(G)$ (also called \emph{bags}), i.e., $\beta: V(T) \rightarrow 2^{V(G)}$, satisfying the following properties. 
\begin{enumerate}
    \item For every $uv \in E(G)$, there exists $t\in V(T)$, such that $\{u,v\} \subseteq \beta(t)$.
    \item For every $v \in V(G)$, the subgraph of $T$ induced by the set $T_v = \{ t \in V(T)~|~v \in \beta(t)\}$ is a non-empty tree.
\end{enumerate}

\noindent The \emph{width} of a tree decomposition $(T,\beta)$ is $\max_{t\in V(T)} |\beta(t)| - 1$. The \emph{treewidth} of $G$ is the minimum possible width of a tree decomposition of $G$. The mapping $\beta$ is extended from vertices of $T$ to subgraphs of $T$. In particular, for a subgraph $U$  of $T$,  $\beta(U) = \bigcup_{v\in V(U)} \beta(v)$.
\end{definition}

In order to better facilitate the dynamic programming, it is sometimes beneficial to have a tree decomposition with ``nice'' additional properties. Our algorithm performs dynamic programming on a \emph{nice tree decomposition} of the input chordal graph, as defined in~\cite{niceTWchordal} based on the notion of nice tree-decompositions for general graphs~\cite{niceTW}. To ease the presentation of our algorithm, we assume additional properties that the leaf nodes and the root node are empty bags. Notice that this can be attained by adding $O(|V(G)|)$ many nodes to the tree decomposition using standard techniques.

\begin{definition}\label{D:niceTW}
A \emph{nice tree decomposition} of a chordal graph $G$ is a tree decomposition $(T, \beta)$, of width $\omega(G)-1$, where $T$ is a rooted tree and each internal node of $T$ has one or two children, with the following additional properties.
\begin{enumerate}

\item For each node $v\in V(T)$, either $\beta(v)=\emptyset$ or $G[\beta(v)]$ is a clique.

\item Each node of $T$ belongs to one of the following types: \emph{introduce}, \emph{forget}, \emph{join} or \emph{leaf}.

\item A join node $v$ has two children $v_1$ and $v_2$ such that $\beta(v) = \beta(v_1) = \beta(v_2)$.

\item An introduce node $v$ has one child $v_1$ such that $\beta(v) \setminus \set{x} = \beta(v_1)$, where $x \in \beta(v)$.

\item A forget node $v$ has one child $v_{1}$ such that $\beta(v)  = \beta(v_1) \setminus \set{x}$, where $x \in \beta(v_1)$.

\item A leaf node $v$ is a leaf of $T$ with $\beta(v)=\emptyset$.

\item The tree $T$ is rooted at a leaf node $r$ with $\beta(r) = \emptyset$.
\end{enumerate}
\end{definition}

\section{FPT Algorithm for Chordal Graphs}\label{sec:chordal}
We present an FPT algorithm for chordal graphs parameterized by the clique number (which is also the treewidth plus $1$). We assume a nice tree decomposition $(T, \beta)$ of our chordal graph. If no such decomposition is given, we can compute one in $O(n+m)$ time. Our algorithm traverses the nice tree decomposition bottom-up. At each node $v$ of the tree, we  construct a table of size $O\left(2^{2^{\omega(G)}}\right)$ containing ``partial solutions'' for the graph induced by the vertices in the bags of the subtree rooted at $v$ (denoted by $G_{\leq v}$). Informally, a partial solution for $G_{\leq v}$ is a set of vertices of $G_{\leq v}$ that, when combined with certain other vertices of $G-G_{\leq v}$, covers all vertices of $G_{\leq v}$. We associate a ``type'' to each of these partial solutions which encodes, among other information, the vertices required to convert it into a valid solution for $G_{\leq v}$, and the effect of this partial solution to the rest of the graph (Definition~\ref{def:type}).

To ensure that at least one of these partial solutions can be ``extended'' and will be part of a geodetic set with minimum cardinality of $G$, we characterize the shortest path structure between a pair of vertices $u,w$ where $u\in V(G_{\leq v})$ and $w\in V(G) \setminus V(G_{\leq v})$ (Lemma~\ref{lem:cover-clique}). Observe that $v$ cannot be a leaf or  parent of a leaf in $T$, and therefore, we observe that the vertices in the bag $\beta(v)$ induce a \emph{clique cutset} (clique whose removal disconnects the graph) and all shortest paths between $u,w$ contain vertices from $\beta(v)$. Let $X'\subseteq \beta(v)$ be the vertices lying in some shortest path between $u,w$.
Observe that ``pre-selecting'' the vertices of $X'$ captures the effect of $w$ on $G_{\leq v}$, if $w$ is selected in the solution set. For a given solution set, doing this for all vertices of the set, we obtain a collection of subsets of $\beta(v)$. Hence, by considering all $2^{2^{|\beta(v)|}}$ different collections of subsets of $\beta(v)$, we can capture the possible effects of all the solution vertices in $G-G_{\leq v}$, i.e., ``\emph{exterior vertices}'' on $G_{\leq v}$. For different collections of subsets of $\beta(v)$, we have different ``types'' of partial solutions. 

Once we have all the partial solutions for the children of a node $v$, we show how to extend these to get the partial solutions of $v$. It is possible that a partial solution of a node of some ``type'' is extended to a partial solution of its parent of a different ``type''. Depending on the node under consideration, we define an exhaustive set of rules to ensure that the extended partial solutions are valid (Definitions~\ref{def:introduce}, \ref{def:forget}, \ref{def:join}, \ref{def:root}). We prove the exhaustiveness of these rules in Lemmas~\ref{lem:introduce-minimal}, \ref{lem:forget-minimal}, \ref{lem:join-minimal}, \ref{lem:root-minimal} and we prove the correctness of our algorithm in Lemma~\ref{lem:correct}.

\subsection{Geodetic Sets and Clique Cutsets}
Recall that on general (non-chordal) graphs, \textsc{Minimum Geodetic Set} is W[1]-hard when parameterized by treewidth (even for much more restricted parameters solution size, feedback vertex set number and pathwidth, combined)~\cite{kellerhals2020} and even NP-hard for graphs of feedback vertex set number~13 and pathwidth~17~\cite{DBLP:conf/iwpec/Tale25}. What helps us achieve tractability in the case of chordal graphs is the existence of clique cutsets in chordal graphs and the fact that they interact nicely with the shortest paths. We introduce a few definitions and notations first. 

\begin{definition}
    Consider a graph $G$, and let $y$ be a vertex of $G$ and $X$ be a subset of $V(G)$ such that $G[X]$ is a clique. We say that $y$ is \emph{close} to a nonempty set $A \subseteq X$ with respect to $X$ if $d(y,x) = d_y$ when $x \in A$ and $d(y,x) = d_y +1$ when $x \in X \setminus A$, for some integer $d_y$.
\end{definition}  

Using this definition, we have the following intuitive but crucial lemma.

\begin{lemma}\label{lem:cover-clique}
Let $X$ be a clique cutset of a graph $G$ and $u,v$ be vertices lying in two different connected components of $G-X$. Let $A$ and $B$ be two subsets of $X$ such that $u$ (resp. $v$) is close to $A$ (resp. $B$) with respect to $X$. Then, a vertex $x\in I(u,v) \cap X$ if and only if $x \in A \cap B$ or, $A \cap B = \emptyset$ and $x \in A \cup B$.
\end{lemma}

\begin{proof}

Assume that for all $y \in A$, $d(u,y) = d_u$ and for all $y \in B$, $d(v,y) = d_v$ for some integers $d_u$, $d_v$ and let $x \in X$. We distinguish the following two cases. See Figure~\ref{fig:cutset} for an illustration.

\begin{figure}
\centering
\begin{subfigure}{.5\textwidth}
  \centering
  \includegraphics[width=.9\linewidth]{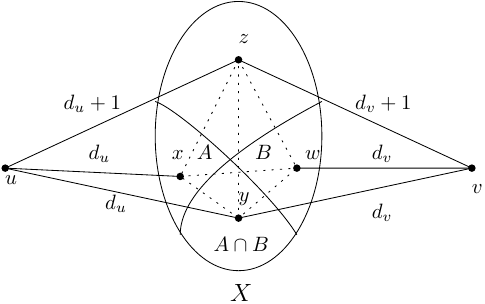}
  \caption{$A\cap B \neq \emptyset$ and $y\in A\cap B$. Here, $d(u,v)= d_u +d_v$.}
  \label{fig:OC1}
\end{subfigure}%
\begin{subfigure}{.5\textwidth}
  \centering
  \includegraphics[width=.9\linewidth]{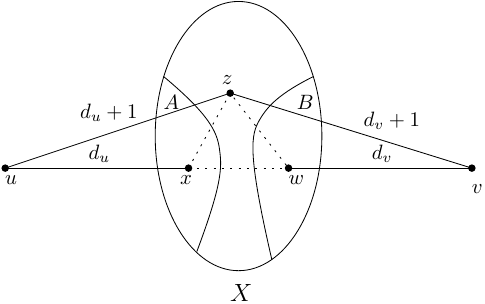}
  \caption{$A\cap B = \emptyset$. Here, $d(u,v)= d_u +d_v+1$.}
  \label{fig:OC2}
\end{subfigure}
\caption{Illustration of Lemma~\ref{lem:cover-clique}. In both (a) and (b), any $u,v$-path passing through $z$ has length at least $d_u+d_v+2$, and hence $z\notin I(u,v)$. In (a), $y\in I(u,v)$ and $x,w\notin I(u,v)$ since any $u,v$-path that passes through any of $x$ or $y$ has length at least $d_u+d_v+1$. In (b), both $x,w\in I(u,v)$.} 
\label{fig:cutset}
\end{figure}

\begin{enumerate}
    \item  $A \cap B\neq\emptyset$. Let $y$ be some vertex in $A\cap B$. First, observe that $d(u,v) \leq d(u,y)+d(v,y) = d_u + d_v$ and $d(u,v) \geq d_u+d_v$ as any shorter path  would contradict the definition of $A$ and $B$. Thus $d(u,v) = d_u+d_v$. If $x \in A \cap B$, then $d(u,x) + d(x,v) = d_u + d_v$,  and hence, $x \in I(u,v)$. Conversely, if $x \notin A \cap B$, then observe that $d(u,x) + d(x,v) \geq d_u + d_v +1$ since every $u$ to $v$ path contains an edge entirely contained in $G[X]$, and hence, $x \notin I(u,v)$.
    
    \item $A \cap B = \emptyset$. First, observe that $d(u,v) \geq d_u + d_v + 1$, as any shorter path would imply the existence of a vertex $y \in X$ such that $d(u,y) + d(y,v) \leq d_u + d_v $, which would contradict $A \cap B = \emptyset$. Since there is some vertex $y\in A$, $d(u,v) \leq d_u + d_v + 1$ as $d(u,y) + d(y,v) \leq d_u + d_v +1$. Hence, $d(u,v) = d_u + d_v + 1$.  Now, if $x\in A \cup B$, then  $d(u,x) + d(x,v) = d_u + d_v +1$, and hence $x \in I(u,v)$. Conversely, if $x \notin A \cup B$, then $d(u,x) + d(x,v) = d_u + d_v +2$, and hence $x \notin I(u,v)$.
\end{enumerate}
This completes the proof.
\end{proof}

\subsection{The Algorithm}\label{SS:algo}

From now on, let $T$ denote a nice tree decomposition of $G$. For a node $v\in T$, let $G_{\leq v}$ be the subgraph of $G$ induced by the vertices present in the nodes of the subtree of $T$ rooted at $v$; let $G_{< v}$ denote $G_{\leq v}-\beta(v)$; $G_{\geq v}$ denote $G-V(G_{ <v})$; and $G_{> v}$ denote $G_{\geq v}-\beta(v)$. For a node $v$, let $\mathcal{T}_v$ be the set of all $5$-tuples $\typpe = (\typent, \typext, \typecov, \typeuncov, \typebag)$  where $\typent,\typext$ are Boolean vectors of size $2^{|\beta(v)|}$ indexed by subsets of $\beta(v)$ and $\typecov,\typeuncov,\typebag$ are subsets of $\beta(v)$, and together they constitute a partition of $\beta(v)$. Since $|\beta(v)|\leq \omega(G)$, the cardinality of $\mathcal{T}_v$ is $2^{2^{O(\omega(G))}}$. We provide an intuition for the elements of the tuple in the following paragraphs.

For a node $v$ and a $5$-tuple $\typpe = (\typent, \typext, \typecov, \typeuncov, \typebag)$  let $\auxGraph{v}{\typpe}$ denote the graph obtained by adding a vertex $S_v$ to $G_{\leq v}$ whenever there is a set $S\subseteq \beta(v)$ with $\typext[S] = 1$, and making $S_v$ adjacent to each $x \in S$. See Figure~\ref{fig:EXT} for an illustration. Let $\Svertices{v}{\typpe} = \set{S_v \colon S\subseteq \beta(v),\typext[S]=1}$ denote the set of newly added vertices. Observe that $G_{\leq v}$ is an induced subgraph of $\auxGraph{v}{\typpe}$ for any $5$-tuple $\typpe \in \mathcal{T}_v$.

\begin{figure}
    \centering
    \includegraphics[scale=1]{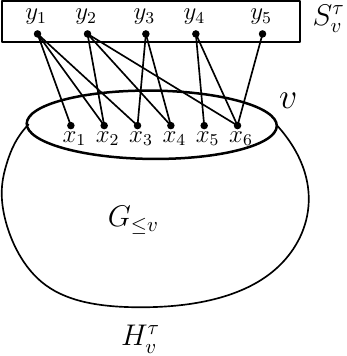}
    \caption{Here $v$ is a node of $T$ such that $\beta(v) = \set{x_1,\ldots,x_6}$. Further, let $\typpe = (\typent, \typext, \typecov, \typeuncov, \typebag)$  be a $5$-tuple such that for $S\subseteq \beta(v)$, $\typext[S] = 1$ iff $S\in \mathcal{S}=\set{\set{x_1,x_2,x_3}, \set{x_2,x_4,x_6}, \set{x_3,x_4}, \set{x_5,x_6}, \set{x_6}}$. We construct $\auxGraph{v}{\typpe}$ by first considering the graph $G_{\leq v}$ and for each $S\in \mathcal{S}$, we add a vertex in $\Svertices{v}{\typpe}$ and connect it to each vertex of $S$.}
    \label{fig:EXT}
\end{figure}

We begin by providing an intuition for our $5$-tuple parameters.   For each $5$-tuple $\typpe$ of $\mathcal{T}_v$, we will store a ``partial solution'' of minimum size, say $Q\subseteq V(G_{\leq v})$, which has the following properties: In graph $\auxGraph{v}{\typpe}$, for each vertex $x\in \typecov$, there exists a pair of vertices $w_1,w_2 \in Q\cup \Svertices{v}{\typpe}$ such that $w_1\in Q$ and $x$ is covered by $w_1$ and $w_2$ (i.e., $x\in I(w_1,w_2)$).
Intuitively, for a set $A\subseteq \beta(v)$ of vertices, the Boolean $\typent[A]$ represents whether there is some vertex $y$ in the partial solution for $G_{\leq v}$ such that $y$ is close to $A$ with respect to $\beta(v)$ (``int'' stands for ``interior'').  The Boolean $\typext[A]$ represents whether we need to add, at a later step of our algorithm, some vertex $y$ which is close to $A$ with respect to $\beta(v)$. Here, $y$ is a vertex of $G_{> v}$ that needs to be added later to the solution, in the upper part of the tree (``ext'' stands for ``exterior''). Moreover, $\typebag$ is the set of vertices of the partial solution that are present in the bag (i.e., $Q\cap \beta(v)$).  Furthermore, $\typecov$ imposes the condition on which vertices of $\beta(v)$ must be covered by a pair of vertices, one of which must be from the partial solution and the other vertex could be either from the partial solution or could be an exterior solution vertex from $G_{>v}$. Finally, $\typeuncov = \beta(v)\setminus (\typebag\cup \typecov)$ are the vertices of $\beta(v)$ that will only be covered by a pair of vertices from $G_{>v}$.

Since each non-empty bag of our tree decomposition induces a clique, the information whether there is a vertex in the partial solution that is close to $A$ with respect to $\beta(v)$ (stored in $\typent[A]$) and whether there will be a vertex that will be selected later in the solution that will be close to $A$ with respect to $\beta(v)$ (stored in $\typext[A]$) shall be very useful to us (using Lemma~\ref{lem:cover-clique}). We formally define the properties of the $5$-tuples and partial solutions that we will be using for our algorithm in the following definition.

\begin{definition}\label{def:type}
Let $v$ be a node of $T$. A $5$-tuple $\typpe = (\typent, \typext, \typecov, \typeuncov, \typebag)$  of $\mathcal{T}_v$ is a ``type associated with $v$'' if there exists a set of vertices $X\subseteq V(\auxGraph{v}{\typpe})$ such that the following hold.
\begin{enumerate}[label=(\roman*)]
    
    \item \label{it:type-1} $\Svertices{v}{\typpe} \subseteq X$ and $\typebag = X \cap \beta(v)$.
        
    \item\label{it:type-2} For each vertex $w\in V(G_{< v}) \cup \typecov$, there exists a pair $w_1,w_2 \in X$ such that $w\in I(w_1,w_2)$ and $w_1 \in X\setminus \Svertices{v}{\typpe}$.

    \item\label{it:type-3} For a subset $A\subseteq \beta(v)$, we have $\typent[A]=1$ if and only if $X \cap V(G_{\leq v})$ contains a vertex which is close to $A$ with respect to $\beta(v)$.
    
\end{enumerate}

Moreover, we shall say that the set $X\setminus \Svertices{v}{\typpe}$ is a ``{\emph certificate}'' for $(v,\typpe)$.
\end{definition}

\medskip
\noindent\textbf{Remark} Note that Definition~\ref{def:type} requires that vertices in $\typecov$ are covered by $X$, but it does not explicitly forbid vertices in $\typeuncov$ from being covered. In our dynamic programming algorithm, if a partial solution $X$ covers a vertex $w \in \typeuncov$, it will also be a valid certificate for a "stronger" type where $w \in \typecov$. Since we compute a minimum-size set, the algorithm will naturally propagate the type that correctly identifies $w$ as covered whenever this information is necessary for extending the solution. Thus, for simplicity, we do not enforce the negative constraint in the definition.

Observe that the only type associated with the root node of $T$ is $\typpe_0 = (\mathbf{0},\mathbf{0},\emptyset,\emptyset,\emptyset)$ where $\mathbf{0}$ denotes the all-$0$ vector. Now, we characterize the types associated with different sorts of nodes of the nice tree decomposition $T$. Lemma~\ref{lem:cover-clique} implies that to compute an optimal (partial) solution for $G_{\leq v}$ for a given node $v$, it is sufficient to ``guess'' for which subsets $A$ of $\beta(v)$, there will exist (in the future solution that will be computed for ancestors of $v$) a vertex $y$ which is close to $A$ with respect to $\beta(v)$. 

Suppose we have a fixed geodetic set $D$ of $G$. In the following lemma, we show that when $D$ is restricted to a particular subgraph $G_{\leq v}$, the set $D\cap G_{\leq v}$ serves as a certificate of some $(v,\typpe)$. 

\begin{lemma}\label{lm:type}
Let $D$ be a geodetic set of $G$, then for each node $v\in T$, there is a $5$-tuple $\typpe=(\typent,\typext,\typecov,\typeuncov,\typebag)$ such that $\typpe$ is a type associated with $v$ and $D\cap V(G_{\leq v})$ is a certificate of $(v,\typpe)$.
\end{lemma}
\begin{proof}
We construct $\typpe$ from $D$ as follows. Define $\typebag = D\cap \beta(v)$. For each vertex $u\in D\cap V(G_{\leq v})$ we find the set $Z_u\subseteq \beta(v)$ such that $u$ is close to $Z_u$ with respect to $\beta(v)$ and put $\typent[Z_u]=1$. For each $u\in D\cap V(G_{>v})$ we find the set $Z_u\subseteq \beta(v)$ such that $u$ is close to $Z_u$ with respect to $\beta(v)$ and put $\typext[Z_u]=1$. Similarly, let us construct $H^{\typpe}_v$, and let $X= S^{\typpe}_v \cup (D\cap G_{\leq v})$. Now, for each vertex $w\in \beta(v)\setminus D$, if $w\in I(w_1,w_2)$ such that $w_1,w_2\in X$ and $w_1\in X\setminus S^{\typpe}_v$, then we put $w$ in $\typecov$. Finally, we put $\typeuncov=\beta(v) \setminus (\typebag\cup \typecov)$. Observe that $D \cap V(G_{\leq v})$ is a certificate of $(v,\typpe)$ and $\typpe$ is a type associated with $v$. 
\end{proof}

For a node $v$, there might be some $5$-tuples in $\mathcal{T}_v$ which are not associated with $v$. In the following lemma, we establish certain restrictions that any type associated with $v$ must follow.

\begin{lemma}\label{lem:valid}
Let $v$ be a node of $T$ and $\typpe=(\typent,\typext,\typecov,\typeuncov,\typebag)$ be a type associated with $v$. Then $\typpe$ must satisfy all of the following conditions.
\begin{enumerate}[label=(\alph*)]
\item Whenever we have a vertex $u \in \typebag$ we have  $\typent[\set{u}] = 1$.

\item $\typent[\emptyset] = \typext[\emptyset] = 0$. 

\item Let $x \in \beta(v)$ and $A,B\subseteq \beta(v)$ such that $\typent[A] = 1$ and $\typext[B] = 1$. Then, if either $x \in A \cap B$ or $A \cap B = \emptyset$ and $x \in A \cup B$, then $x\in \typecov$.
\end{enumerate}
\end{lemma}
\begin{proof}
    The proof of (a) follows from the fact that for a clique cutset $X$ and a vertex $u\in X$, $u$ is close to $\set{u}$ with respect to $X$. The proof of (b) follows from the fact that there can be no vertex close to $\emptyset$ with respect to $\beta(v)$. The proof of (c) follows directly from Lemma~\ref{lem:cover-clique}.
\end{proof}

From now on for a node $v$, we will only consider the $5$-tuples which satisfy the conditions of Lemma~\ref{lem:valid}. 





Let $v$ be an introduce node and $u$ be its child, and let $\typpe,\typpe_1$ be types associated with $v,u$ respectively. Below, we state some \emph{compatibility rules} that $\typpe$ and $\typpe_1$ must follow so that the certificate of $(u,\typpe_1)$ can be extended to a  certificate of $(v,\typpe)$.

\begin{definition}\label{def:introduce}
Let $v$ be an introduce node and $\typpe$ be a type associated with $v$. Let $u$ be the child of $v$ such that $\beta(v) =\beta(u) \cup \{x\}$ and $\typpe_1$ be a type associated with $u$. The pair $(\typpe,\typpe_1)$ is \emph{compatible} if the following holds.

\begin{enumerate}[label=(\alph*)]

\item \label{it:intro-a} $\typebag_1 = \typebag \setminus \set{x} =\typebag \cap \beta(u)$.

\item \label{it:intro-bnew}  If $x\notin \typebag$, then $\typecov_1 = \typecov \setminus \set{x} = \typecov \cap \beta(u)$ and $\typeuncov_1 = \typeuncov \setminus \set{x} = \typeuncov \cap \beta(u)$.

\item \label{it:intro-bnew1} If $x\in \typebag$, then let us define the set $Cov(v)$ as follows: for each $w\in \beta(u)$, we put the vertex $w$ in $Cov(v)$ if and only if all of the following are true: 
\begin{enumerate}
    \item $\exists A\subseteq \beta(u)$ with  $w\in A$ such that $\typext[A] = \typext_1[A]=1$, and 
    \item for each $B\subseteq \beta(u)$ such that $w\in \beta(u)$, $\typent_1[B] = 0$, and 
    \item for each $B\subseteq \beta(u)$ with $w\in \beta(u)$ such that $\typext_1[B] = 1$, $\typebag_1 \subseteq  B\setminus \{w\}$.
\end{enumerate}
Then,  $\typecov_1 = \typecov \setminus Cov(v)$, $\typecov = \typecov_1\cup Cov(v)$, and $\typeuncov_1 = \typeuncov \cup Cov(v)$. 

\item\label{it:intro-c} For all non-empty $A \subsetneq \beta(u)$, $\typext_1[A] = 1$ if and only if  $\typext[A] = 1$ or $\typext[A \cup \set{x}] = 1$. 
    
\item\label{it:intro-d} $\typext_1[\emptyset] = 0$. Moreover, given that $\beta(u) \neq \emptyset$, $\typext_1[\beta(u)] = 1$ if and only if  $\typext[\beta(v)] = 1$, $\typext[\beta(u)] = 1$, $x \in \typebag$ or $\typext[\set{x}] = 1$.

\item\label{it:intro-e} If $x\in \typecov$ then there exist non-empty sets $A\subseteq \beta(u)$, $B\subseteq \beta(v)\setminus A$ such that $\typent_1[A]=1$ and $\typext[B\cup \set{x}]=1$.

\item\label{it:intro-f} $\typent[\set{x}] = 1$ if and only if $x \in \typebag$.
    
\item \label{it:intro-g} For all non-empty $A \subseteq \beta(u)$, $\typent [A  \cup \set{x}] = 0$.    

\item \label{it:intro-h} For all $A \subseteq \beta(u)$, $\typent[A] = 1$ if and only if  $\typent_1[A] = 1$.

\end{enumerate}

\end{definition}

\begin{lemma}\label{lem:introduce-minimal}
Let $v$ be an introduce node, $\typpe$ be a type associated with $v$, $u$ be the child of $v$ and $D$ be a minimal certificate of $(v,\typpe)$. Then there exists a type $\typpe_1$ associated with $u$ such that $(\typpe,\typpe_1)$ is a compatible pair.
\end{lemma}

\begin{proof}
Let $\beta(v)=\beta(u) \cup \set{x}$. Define $\typebag_1 = D \cap \beta(u)$,  and $\typent_1[A] = 1$ if and only if there exists $y \in D \setminus \set{x}$ such that $y$ is close to $A$ with respect to $\beta(u)$. Finally, define $\typecov_1$ and $\typeuncov_1$ according to Definitions~\ref{def:introduce}\ref{it:intro-bnew} and \ref{def:introduce}\ref{it:intro-bnew1}, and define $\typext_1$ according to Definitions~\ref{def:introduce}\ref{it:intro-c} and \ref{def:introduce}\ref{it:intro-d}. To complete our proof, we will establish that the set $D'=D\setminus \set{x}$ satisfies all of the conditions in Definition~\ref{def:type} for $(u,\typpe_1)$ and the pair $(\typpe,\typpe_1)$ satisfies all conditions in Definition~\ref{def:introduce}. Observe that all of the conditions are trivially satisfied if $\beta(u) =\emptyset$, thus we assume that $\beta(u) \neq \emptyset$ for the rest of the proof. 

To begin with, we establish that $\typpe_1$ is a type associated with $u$ satisfying all of the conditions in Definition~\ref{def:type}. Since $\typpe$ is a type associated with $v$ and $D$ is a (minimal) certificate, we have that there exists a set of vertices $X\subseteq V(H^{\typpe}_v)$ such that $X=D \cup S^{\typpe}_v$ and (i) $S^{\typpe}_v \subseteq X$  and $\typebag=X\cap \beta(v)$, (ii) for each vertex $w\in (V(G_{\leq v})\setminus \beta(v)) \cup \typecov$, $\exists w_1,w_2\in X$ such that $w\in I(w_1,w_2)$ and $w_1\in X\setminus S^\typpe_v$, and (iii) for each $A\subseteq \beta(v)$, we have $\typent[A]=1$ if and only if $X\cap V(G_{\leq v})$ contains a vertex which is close to $A$ with respect to $\beta(v)$. Now, let $X'= (X\cap V(G_{\leq u}) \cup S^{\typpe_1}_u = D'\cup S^{\typpe_1}_u$ . By definition of $\typebag_1$ and $X'$, we have that $\typebag_1 = D\cap \beta(u) = (D' \cup \set{x}) \cap \beta(u) = D' \cap \beta(u)$ (since $x\notin \beta(u)$), and $S^{\typpe_1}_v \subseteq X'$. Since Definition~\ref{def:type}\ref{it:type-1} and \ref{def:type}\ref{it:type-3} are trivially satisfied, to complete our proof we will show that for each vertex $w\in V(G_{< v}) \cup \typecov$, $\exists w_1,w_2\in X$ such that $w\in I(w_1,w_2)$ and $w_1\in X\setminus S^\typpe_v$. Targeting a contradiction, let there be some vertex $w'\in V(G_{< u})\cup \typecov_1$ such that there is no $w'_1,w'_2 \in X'$ such that $w'\in I(w'_1,w'_2)$ and $w'_1\in X'\setminus S^{\typpe_1}_v$. Since, $\typecov_1\subseteq \typecov$ and $V(G_{\leq u})\setminus \beta(u) = V(G_{\leq v})\setminus \beta(v)$, we have that $w'\in (V(G_{\leq v})\setminus \beta(v)) \cup \typecov$. Therefore, $\exists w_1,w_2\in X$ such that $w'\in I(w_1,w_2)$ and $w_1\in X\setminus S^\typpe_v \subseteq (X'\setminus S^{\typpe_1}_u)\cup \set{x}$. We distinguish two cases depending on whether $x\in \typebag$ or not.

\begin{enumerate}
    \item $x\notin \typebag$: Then,  $w_1 \neq x$ and $w_2\neq x$. In this case $w_1\in X'\setminus S^{\typpe_1}_u$. Here, if $w_2\in X'\setminus S^{\typpe_1}_u$, then we reach a contradiction since $w'\in I(w_1,w_2)$, $w_1,w_2\in X'$ and $w_1\in X'\setminus S^{\typpe_1}_v$. If $w_2\notin X'\setminus S^{\typpe_1}_u$, then $w_2\in S^{\typpe}_v$. Thus, $\exists A\subseteq \beta(v)$ such that $\typext[A]=1$. Now, as per Definition~\ref{def:introduce}\ref{it:intro-c}-\ref{it:intro-d}, if $A\subsetneq \beta(v)$, then $\typext_1[A]=1$ and else, $\typext[\beta(u)]$=1. In both cases, there is a vertex $w'_2\in S^{\typpe_1}_u$ whose impact on the bag is same as $w_2$, leading us to a contradiction.

    \item $x\in \typebag$. If $w_1\neq x$ and $w_2\neq x$, then our proof is identical to the above case. Therefore, without loss of generality, let us assume that $w_1= x$. Since $x\in \typebag$, due to Definition~\ref{def:introduce}\ref{it:intro-d}, we have that $\typext_1[\beta(u)]=1$, i.e., $\exists z\in S^{\typpe_1}_{u}$ such that $N_{H^{\typpe_1}_{u}}(z)=\beta(u)$. Therefore, if $w_2\notin S^{\typpe}_v$ (i.e, $w_2\in D\cap (V(G_{\leq v})\setminus \{x\}) = D\cap V(G_{\leq u})$), then $w'\in I(w_2,z)$, a contradiction. 
    
    Thus, let us assume that $w_2\in S^{\typpe}_v$. Therefore, $\exists A\subseteq \beta(v)\setminus \{x\} =\beta(u)$ with $w'\in A$ such that $\typext[A] = \typext_1[A]=1$ (Definition~\ref{def:introduce}\ref{it:intro-h}). We will establish that $w'$ satisfies all three conditions of Definition~\ref{def:introduce}\ref{it:intro-bnew1}. Moreover, if $\exists B \subseteq \beta(u)$ with $w'\in B$ such that $\typent_1[B]=1$, then $w'$ is covered due to Lemma~\ref{lem:valid}(c). Therefore, we have that $\forall B\subseteq \beta(u)$ with $w'\in \beta(u)$, $\typent_1[B]=0$. Next, if $\exists B\subseteq \beta(u)$ with $w'\in B$ such that $\typext_1[B]=1$ and $\typebag_1$ contains a vertex $z\in \beta(u) \setminus (B \setminus \set{w'})$, then observe that $w'\in I(z,B_v)$. Therefore, we have that $\forall B\subseteq \beta(u)$ with $w'\in \beta(u)$ such that $\typext_1[B] =1$, $\typebag_1\subseteq B \setminus \set{w'}$. Observe that $w'$ satisfies all three conditions of Definition~\ref{def:introduce}\ref{it:intro-bnew1}, implying $w'\in Cov(v)$. Since $\typecov_1 = \typecov \setminus Cov(v)$ (recall that $x\in \typebag$ and hence Definition~\ref{def:introduce}\ref{it:intro-bnew1} is applied) and $w'\in Cov(v)$, we have that $w'\notin \typecov_1(u)$. This is a contradiction to our assumption that $w'\in \typecov_1$. This completes our proof.
\end{enumerate}

Finally, we establish that the pair $(\typpe,\typpe_1)$ satisfies all conditions in Definition~\ref{def:introduce}. Observe that  the conditions in Definition~\ref{def:introduce}\ref{it:intro-a}, \ref{def:introduce}\ref{it:intro-bnew}, \ref{def:introduce}\ref{it:intro-bnew1}, \ref{def:introduce}\ref{it:intro-c},\ref{def:introduce}\ref{it:intro-d} follow directly from our construction of $\typpe_1$. Next, we establish other conditions of Definition~\ref{def:introduce}. Since $\beta(u)$ is a clique cutset that separates $x$ from $V(G_{<u})$, we have that $\typent[\set{x}]=1$ if and only if $x\in \typebag$ and for all non-empty $A\subseteq \beta(u)$, $\typent[A\cup \{x\}]=0$ (Definition~\ref{def:introduce}\ref{it:intro-f} and \ref{def:introduce}\ref{it:intro-g}). Definition~\ref{def:introduce}\ref{it:intro-h} follows from the fact that $D'=D\setminus \{x\}$ and for each $A\subseteq \beta(u)$, $\typent [A]$ and $\typent_1[A]$ depend on $D'$. Finally, to complete our proof, we establish Definition~\ref{def:introduce}\ref{it:intro-e}. First, observe that no two vertices $y,z\in \beta(v)\setminus \{x\}$ can cover $x$ since $G[\beta(v)]$ is a clique.  Since $\beta(u) = \beta(v)\setminus \set{x}$ separates $\{x\}$ from $V(G_{<v})$, we have that for each vertex  $y\in V(G_{<v})$, $d(y,\beta(v))+1 \leq d(y,x)$. Therefore, in $H^{\typpe}_v$, for any two vertices $y,z \in V(G_{\leq v})\setminus \{x\}$,  it is not possible that $x\in I(y,z)$. Thus, if $\exists y,z\in H^{\typpe}_v$ such that $x\in I(y,z)$, then exactly one of $\set{y,z}$ must be from $S^{\typpe}_v$ (recall that one of $\set{y,z}$ must be from $V(G_{\leq v})\setminus \{x\}$). Without loss of generality, let us assume that $y\in S^{\typpe}_v$ and $z\in V(G_{\leq v})\setminus \{x\}$. Since $z\in \beta(v)\cap \typebag$, we have that $\typent_1[\{z\}] = 1$ (Lemma~\ref{lem:valid}(a)). Let $y$ correspond to a vertex of $S^{\typpe}_v$ that is close to $B\subseteq \beta(v)$, implying $\typext[B]=1$. Since $x\in I(y,z)$ and $\beta(v)$ induces a clique, we have that $z\notin B$ ($yz\notin E(H^{\typpe}_v)$) and that $x\in B$. Therefore, $\typext[B\cup \{x\}]=1$. This completes our proof. 
\end{proof}

For an introduce node $v$ with child $u$, let $\Compatible{v}$ denote the set of compatible pairs $(\typpe,\typpe_1)$ where $\typpe,\typpe_1$ are types associated with $v$ and $u$, respectively. 

Let $v$ be a forget node and $u$ be its child. Let $\typpe,\typpe_1$ be types associated with $v,u$ respectively. Below, we state some rules that $\typpe$ and $\typpe_1$ must follow so that the certificate of $(u,\typpe_1)$ can be extended to a  certificate of $(v,\typpe)$.

\begin{definition}\label{def:forget}
Let $v$ be a forget node and $\typpe$ be a type associated with $v$. Let $u$ be the child of $v$ such that $\beta(v) = \beta(u)\setminus \set{x}$ and $\typpe_1$ be a type associated with $u$. The pair $(\typpe,\typpe_1)$ is \emph{compatible} if the following holds.

\begin{enumerate}[label=(\alph*)]

    \item\label{it:forget-a} $\typebag = \typebag_1 \setminus \set{x}$,
    
    \item\label{it:forget-b} For all $A \subseteq \beta(v)$, $\typext_1[A] = 1$ if and only if  $\typext[A] = 1$,
    
    \item\label{it:forget-c} For all $A \subseteq \beta(v)$, $\typext_1[A  \cup \set{x}] = 0$.
    
    \item\label{it:forget-d} For all non-empty $A \subsetneq \beta(v)$, $\typent[A] = 1$ if and only if  $\typent_1[A] = 1$ or $\typent_1[A \cup \set{x}] = 1$.
    
    \item\label{it:forget-e} $\typent[\emptyset] = 0$. Moreover, given that $\beta(v) \neq \emptyset$, $\typent[\beta(v)] = 1$ if and only if  $\typent_1[\beta(u)] = 1$, $\typent_1[\beta(v)] = 1$, $\typent_1[\set{x}] = 1$, {or $x\in \typebag_1$}, 
    
    
    \item\label{it:forget-f} If $x\notin \typebag_1$, then $\typecov_1 = \typecov \cup \set{x}$ and $\typeuncov_1 = \typeuncov$, and if $x\in \typebag_1$, then $\typecov_1 = \typecov$ and $\typeuncov_1 = \typeuncov$. 
\end{enumerate}
\end{definition}

\begin{lemma}\label{lem:forget-minimal}
Let $v$ be a forget node, $\typpe$ be a type associated with $v$, $u$ be the child of $v$, and $D$ be a minimal certificate of $(v,\typpe)$. Then there exists a type $\typpe_1$ associated with $u$ such that $(\typpe,\typpe_1)$ is a compatible pair.
\end{lemma}

\begin{proof}
Let $\beta(v) = \beta(u) \setminus \set{x}$. Define $\typebag_1 = D \cap \beta(u)$. Moreover, define $\typeuncov_1$ and $\typecov_1$ as per Definition~\ref{def:forget}\ref{it:forget-f}, and $\typext_1$  according to Definition~\ref{def:forget}\ref{it:forget-b} and \ref{def:forget}\ref{it:forget-c}. Finally, $\typent_1[A] = 1$ if and only if there exists $y \in D$ such that $y$ is close to $A$ with respect to $\beta(u)$. First, we establish that $\typpe_1$ and $D$ satisfy all of the conditions in Definition~\ref{def:type}, and thus $\typpe_1$ is a type associated with $u$. Consider the graphs $H^\typpe_v$ and $H^{\typpe_1}_u$. Observe that due to the definition of $\typext_1$, we have that $S^{\typpe}_v = S^{\typpe_1}_u$. We claim that $D$ is a certificate of $(u,\typpe_1)$ as well. Clearly, $S^{\typpe_1}_v \subseteq D$ and $\typebag_1 = D\cap \beta(u)$, satisfying Definition~\ref{def:type}\ref{it:type-1}. To see that Definition~\ref{def:type}\ref{it:type-2} is satisfied, observe that $ V(G_{< u})\cup \typecov_1 \subseteq V(G_{<v}) \cup \typecov$. Finally, Definition~\ref{def:type}\ref{it:type-3} is trivially satisfied since we define $\typent_1$ to satisfy this definition. Moreover, Definition~\ref{def:forget}\ref{it:forget-f} enforces that $\typpe_1$ is compatible with $\typpe$ only if $x\in \typebag_1 \cup \typecov_1$, i.e., $x\notin \typeuncov_1$.

Next, we establish  that $(\typpe,\typpe_1)$ is a compatible pair. Definition~\ref{def:forget}\ref{it:forget-a} is trivially satisfied, and we define $\typpe_1$ according to Definitions~\ref{def:forget}\ref{it:forget-b}, \ref{def:forget}\ref{it:forget-c}, and \ref{def:forget}\ref{it:forget-f}. Therefore,  to complete our proof, we only need to establish that $(\typpe,\typpe_1)$ also satisfy Definition~\ref{def:forget}\ref{it:forget-d} and \ref{def:forget}\ref{it:forget-e}. To see that Definition~\ref{def:forget}\ref{it:forget-d} is satisfied, let $A\subsetneq \beta(v)$ be an arbitrary subset or $\beta(v)$. In one direction, if $\typent[A] = 1$, then there exists a vertex $w\in V(G_{<v})\cap D$ such that $w$ is close to $A$ with respect to $\beta(v)$. Since $A\neq \beta(v)$, clearly $w\neq x$ as $G[x\cup \beta(v)]$ is a clique. Because $\beta(u) = \beta(v) \cup \{x\}$, we have that either $w$ is close to $A$ or $A\cup \set{x}$ in $\beta(u)$, implying either $\typent_1[A]=1$ or $\typent_1[A\cup \set{x}] = 1$. In the other direction, let $\typent_1[A]=1$ or $\typent_1[A\cup \set{x}] = 1$. Then, there is a vertex $w\in D$ such that $w$ is close to either $A$ or $A\cup \set{x}$   with respect to $\beta(u)$, and hence clearly, $w$ is close to $A$ with respect to $\beta(v) = \beta(u)\setminus \{x\}$, implying $\typent[A]=1$.  Finally, we prove that Definition~\ref{def:forget}\ref{it:forget-e} is satisfied. In one direction, let  $\typent[\beta(v)]  =1$, implying that there is a vertex  $w\in D\cap G_{<v}$ that is close to $\beta(v)$ (with respect to $\beta(v)$). If $w=x$, then $x\in \typebag_1$, and the claim is vacuously true. Thus, we assume that $w\neq x$. Since each vertex $y\in \beta(u)\setminus \{x\}$ is at the same distance from $w$, say $t$, we have that $d(w,x) \in \{ t-1,t,t+1\}$. If $d(w,x)=t-1$, then $\typent_1[\set{x}]=1$, if $d(w,x)=t$, then $\typent_1[\beta(u)]=1$, and if $d(w,x)=t+1$, then $\typent_1[\beta(v)]=1$. In the other direction, let $\typent_1[\beta(u)] = 1$, $\typent_1[\beta(v)] = 1$, $\typent_1[\set{x}] = 1$, or $x\in \typebag_1$. If $\typent_1[\beta(u)] = 1$, $\typent_1[\beta(v)] = 1$ or $\typent_1[\set{x}]=1$, then using similar arguments as above, we have that $\typent[\beta(v)]=1$. If $x\in \typebag_1$, then $x\in D\cap V(G_{<v})$ such that $x$ is close to $\beta(v)$ with respect to $\beta(v)$, implying $\typent[\beta(v)]=1$. This completes our proof.
\end{proof}

For a node $v$ with child $u$, let $\Compatible{v}$ denote the set of compatible pairs $(\typpe,\typpe_1)$ where $\typpe,\typpe_1$ are types associated with $v$ and $u$, respectively. 

Let $v$ be a join node, and let $u_1,u_2$ be its children. Let $\typpe,\typpe_1,\typpe_2$ be types associated with $v,u_1,u_2$ respectively. Below, we state some rules that $\typpe,\typpe_1,\typpe_2$ must follow so that the certificates for $(u_1,\typpe_1)$ and $(u_2,\typpe_2)$ can be combined and extended to a  certificate of $(v,\typpe)$. Since we assume connectedness of the input graph, we can safely assume that a join node will never join two empty nodes.

\begin{definition}\label{def:join}
Let $v$ be a join node and $\typpe$ be a type associated with $v$. Let $u_1,u_2$ be the children of $v$ and $\typpe_1,\typpe_2$ be types associated with $u_1,u_2$ respectively. The triplet $(\typpe,\typpe_1,\typpe_2)$ is compatible if all of the following holds.

\begin{enumerate}[label=(\alph*)]
    \item\label{it:join-a} $\typebag = \typebag_1 = \typebag_2$,
    
    \item\label{it:join-b} For $i,j$ such that $\set{i,j}= \set{1,2}$ and $A \subseteq \beta(v)$, $\typext_i[A] = 1$ if and only if $\typent_j[A] = 1$ or $\typext[A] = 1$.
    
    \item \label{it:join-c} Let $Cov(u_1,u_2)\subseteq \beta(v)$ be the set of all vertices $x$ such that for each $x\in Cov(u_1,u_2)$, there exists $A,B \subseteq \beta(v)$ with $\typent_1[A] =1$, $\typent_2[B] =1$ and either $ x \in A \cap B$ or $A \cap B = \emptyset$ and $x \in A \cup B$. Then, $\typecov = \typecov_1 \cup \typecov_2 \cup Cov(u_1,u_2)$. {Moreover, $\typeuncov = (\typeuncov_1 \cap \typeuncov_2) \setminus Cov(u_1,u_2)$.}
    
    \item \label{it:join-d} For $i,j$ such that $\set{i,j}= \set{1,2}$ and $A \subseteq \beta(v)$, if $\typent_i[A] = 1$ then $\typent[A] = 1$ and $\typext_j[A] = 1$.
    
    \item \label{it:join-e} For $i,j$ such that $\set{i,j}= \set{1,2}$ and $A \subseteq \beta(v)$, if $\typent_i[A] = 0$ and $\typent_j[A] = 0$ then $\typent[A] = 0$.
\end{enumerate}
\end{definition}

\begin{lemma}\label{lem:join-minimal}
Let $v$ be a join node, $\typpe$ be a type associated with $v$, and $D$ be a minimal certificate of $(v,\typpe)$. Let $u_1,u_2$ be the children of $v$. Then there are types $\typpe_1,\typpe_2$ associated with $u_1,u_2$ respectively such that $(\typpe,\typpe_1,\typpe_2)$ is a compatible triplet.
\end{lemma}

\begin{proof}
For each $i\in \set{1,2}$ define $\typebag_i = \typebag$ and $\typent_i[A] = 1$ if and only if there exists $y \in D \cap G_{\leq u_i}$ such that $y$ is close to $A$ with respect to $\beta(u_i)$. Define $\typext_i$ according to Definition~\ref{def:join}\ref{it:join-b}. Define  $\typecov_i$ in the following manner: Let $D_i = D\cap V(G_{\leq u_i})$. Recall the definition of the graph $H^{\typpe_i}_{u_i}$ and $S^{\typpe_i}_{u_i}$ (from beginning of the Section~\ref{SS:algo}). Then, let $\typecov_i = \set{w~|~w\in I(x,y)\cap \beta(v), x\in D_i~ \& ~(y\in D_i\cup S^{\typpe_i}_{u_i} )}$. 
Observe that for each $i\in \set{1,2}$ the sets $D_i$ satisfies all of the conditions in Definition~\ref{def:type} for $(u_i,\typpe_i)$.

We establish that the triplet $(\typpe,\typpe_1,\typpe_2)$ is compatible. Observe that Definition~\ref{def:join}\ref{it:join-a} and \ref{def:join}\ref{it:join-b} are trivially satisfied by construction. Consider $A\subseteq \beta(v)$ and let $\set{i,j} = \set{1,2}$. If $\typent_i[A]=1$, then there is a vertex $w\in D\cap V(G_{<u_i})$ such that $w$ is close to $A$ with respect to $\beta(u_i)=\beta(v)= \beta(u_j)$, implying $\typent[A] =1$, and since $w\in V(G_{>u_j})$, $\typext_j[A]=1$. Thus, Definition~\ref{def:join}\ref{it:join-d} is satisfied. The proof of Definition~\ref{def:join}\ref{it:join-e} follows from the fact that $\beta(u_i)=\beta(u_j) = \beta(v)$ and $D\cap V(G_{<v}) = D\cap (V(G_{<u_i}) \cup V(G_{<u_j}))$. Finally, to see the Definition~\ref{def:join}\ref{it:join-c}, observe that, due to Lemma~\ref{lem:cover-clique}, each vertex of $Cov(u_1,u_2)$ is covered by vertices in $D\cap (V(G_{<u_i}) \cup V(G_{<u_j}))$, and hence $\typecov = \typecov_1 \cup \typecov_2 \cup Cov(u_1,u_2)$ and $\typeuncov = (\typeuncov_1 \cap \typeuncov_2) \setminus Cov(u_1,u_2)$. Therefore, the triplet $(\typpe,\typpe_1,\typpe_2)$ is compatible.
\end{proof}

For a join node $v$ with children $u_1,u_2$, let $\Compatible{v}$ denote the set of compatible triplets $(\typpe,\typpe_1,\typpe_2)$ where $\typpe,\typpe_1,\typpe_2$ are types associated with $v,u_1,u_2$, respectively. Finally, we consider the root $r$. Recall that the only type associated with $r$ is $\typpe_0 = (\mathbf{0},\mathbf{0},\emptyset,\emptyset,\emptyset)$. 

\begin{definition}\label{def:root}
Consider the root $r$, its child $u$ with $\beta(u)=\set{x}$ and a type $\typpe$ associated with $u$. Then $(\typpe_0,\typpe)$ is a compatible pair if $x \in \typecov_1 \cup \typebag_1$ and $\typext[A] = 0$ for all $A \subseteq \beta(u)$.
\end{definition}

The proof for the following lemma is analogous to that of Lemmas~\ref{lem:introduce-minimal}--\ref{lem:join-minimal}.

\begin{lemma}\label{lem:root-minimal}
For any minimal geodetic set of $G$, there is a compatible pair $(\typpe_0,\typpe_1)$ where $\typpe_1$ is a type associated with $u$.
\end{lemma}

Let $\Compatible{r}$ denote the set of compatible pairs $(\typpe_0,\typpe)$ where $\typpe$ is a type associated with the only child $u$ of the root $r$.

Now we are ready to describe our algorithm. We process the nodes of $T$ bottom-up. Let $v$ be the current node under consideration. If $v$ is a leaf node, then define $$\solution{v}{(\mathbf{0},\mathbf{0},\emptyset,\emptyset,\emptyset)}=\emptyset$$ 

\noindent Let $v$ be an introduce node having $u$ as child. Then for each type $\typpe\in \mathcal{T}_v$ define $$\solution{v}{\typpe} = \displaystyle\min\limits_{(\typpe,\typpe_1)\in \Compatible{v}} \solution{u}{\typpe_1}\cup \typebag$$

\noindent Let $v$ be a forget node having $u$ as child. Then for each type $\typpe\in \mathcal{T}_v$ define $$\solution{v}{\typpe} = \displaystyle\min\limits_{(\typpe,\typpe_1)\in \Compatible{v}} \solution{u}{\typpe_1}$$

\noindent Let $v$ be a join node having $u_1,u_2$ as children. Then for each type $\typpe\in \mathcal{T}_v$ define $$\solution{v}{\typpe} = \displaystyle\min\limits_{(\typpe,\typpe_1,\typpe_2)\in \Compatible{v}} \solution{u_1}{\typpe_1} \cup \solution{u_2}{\typpe_2} $$

\noindent Finally, for the root $r$ let $u$ be its child and $\typpe_0 = (\mathbf{0},\mathbf{0},\emptyset,\emptyset,\emptyset) $. Define $$\solution{r}{\typpe_0} = \displaystyle\min\limits_{(\typpe_0,\typpe_1)\in \Compatible{r}} \solution{u}{\typpe_1}$$

We show in the following lemma that $\solution{r}{(\mathbf{0},\mathbf{0},\emptyset,\emptyset,\emptyset)}$ corresponds to a geodetic set of $G$ with minimum cardinality. Recall the definitions of $\auxGraph{v}{\typpe}$ and $\Svertices{v}{\typpe}$.

\begin{lemma}\label{lem:correct}
For each node $v$ and type $\typpe$ associated with $v$, $\solution{v}{\typpe}$ is a certificate of $(v,\typpe)$ with minimum cardinality.
\end{lemma} 
\begin{proof}
The statement of the lemma is trivially true when $v$ is a leaf node. By induction we assume the lemma to be true for all nodes of the subtree rooted at $v$. 

\begin{enumerate}
    
\item Assume that $v$ is an introduce node. Let $u$ be the child of $v$ and $\beta(v) = \beta(u) \cup \set{x}$. First, we show that $\solution{v}{\typpe}$ is a certificate of $(v,\typpe)$. Let $\typpe_1$ be a type associated with $u$ such that $\solution{v}{\typpe} = \solution{u}{\typpe_1} \cup \typebag$ and consider the set $D = \solution{u}{\typpe_1} \cup \typebag \cup \Svertices{v}{\typpe}$. By Definition~\ref{def:introduce}\ref{it:intro-a} we have that $\typebag_1 = \typebag \setminus \set{x}$. 
\sloppy Hence, $D\cap \beta(v) = \solution{v}{\typpe} \cap \beta(v) = (\solution{u}{\typpe_1} \cup \typebag) \cap \beta(v) = (\solution{u}{\typpe_1} \cap \beta(v)) \cup (\typebag \cap \beta(v)) = (\solution{u}{\typpe_1} \cap (\beta(u)\cup \set{x})) \cup \typebag = (\solution{u}{\typpe_1} \cap \beta(u)) \cup \typebag = \typebag_1 \cup \typebag = \typebag$. Hence, $D$ satisfies Definition~\ref{def:type}\ref{it:type-1}. 

Now, consider any vertex $w\in (V(G_{\leq v}) \setminus \beta(v)) \cup \typecov$ which is distinct from $x$. Then $w\in (V(G_{\leq u}) \setminus \beta(u)) \cup (\typecov \cap \beta(u))$. Definition~\ref{def:introduce}\ref{it:intro-bnew} and Definition~\ref{def:introduce}\ref{it:intro-bnew1} dictate $\typecov_1$. If $x\notin \typebag$, then Definition~\ref{def:introduce}\ref{it:intro-bnew} ensures that $\typecov_1= \typecov \cap \beta(u)$ and hence $w\in (V(G_{\leq u}) \setminus \beta(u)) \cup \typecov_1$. Now, due to our induction hypothesis, we have that there exists $w_1,w_2\in \solution{u}{\typpe_1} \cup \Svertices{u}{\typpe_1}$ such that $w\in I(w_1,w_2)$ and $w_1\in \solution{u}{\typpe_1}$. If $w_2\in \solution{u}{\typpe_1}$ then $w_1,w_2\in \solution{v}{\typpe}$ and therefore $w_1,w_2\in D$. Hence, Definition~\ref{def:type}\ref{it:type-2} is satisfied in this case. Else, $w_2\in \Svertices{u}{\typpe_1}$; then, there is a set $A\subseteq \beta(u)$ such that $w_2=A$ and $\typext_1[A] = 1$. If $A\neq \beta(u)$, then, due to Definition~\ref{def:introduce}\ref{it:intro-c}, we know that there is a set $B\supseteq A$ such that $\typext[B]=1$. Hence, there is a vertex $b\in \Svertices{v}{\typpe}$ such that $b=B$ and is adjacent to all vertices of $A$. Observe that $w\in I(w_1,b)$. If $A=\beta(u)$, then again, due to Definition~\ref{def:introduce}\ref{it:intro-d} and similar arguments as above, we have a vertex $b'\in \Svertices{v}{\typpe}$ such that $w\in I(w_1,b')$. Hence, Definition~\ref{def:type}\ref{it:type-2} is satisfied. Now, consider the vertex $x$.  Since $x\notin \typebag$, due to Definition~\ref{def:introduce}\ref{it:intro-e}, we have sets $A\subseteq \beta(u)$, $B\subseteq \beta(v)\setminus A$ such that $\typent_1[A]=1$ and $\typext[B\cup \set{x}]=1$. Hence, $\solution{u}{\typpe_1}$ contains a vertex, say $a$, which is close to $A$ with respect to $\beta(u)$. There also exists a vertex, say $b\in \Svertices{v}{\typpe}$, such that $b=B$. By Lemma~\ref{lem:cover-clique}, $x\in I(a,b)$. Therefore, Definition~\ref{def:type}\ref{it:type-2} is satisfied in this case. Now, consider the case when $x\in \typebag$. Recall the definition of $Cov(v)$ from Definition~\ref{def:introduce}\ref{it:intro-bnew1}, and now, $\typebag=\typebag_1\cup Cov(v)$. We can cover the vertices from $V(G_{<v})\cup \typecov_1$ as discussed above. To complete the proof of the fact that Definition~\ref{def:type}\ref{it:type-2} is satisfied in this case, we only need to show that for each vertex $w\in Cov(v)$, there exist  vertices $S^{\typpe}_v \cup D$ such that $w\in I(w_1,w_2)$. Since $w\in Cov(v)$, we have that $\exists A\subseteq \beta(v)\setminus \{x\}$ such that $\typext[A]=1$ and since $x\in \typebag$, $\typent[\{x\}]=1$. Finally, as $w\in A\cup \{x\}$ and $A\cap \{x\} =\emptyset$, $w$ is covered due to Lemma~\ref{lem:valid}(c). Therefore, Definition~\ref{def:type}\ref{it:type-2} is satisfied.

Consider $A\subseteq \beta(v)$ such that $\typent[A]=1$. If $A=\set{x}$, then, due to Definition~\ref{def:introduce}\ref{it:intro-f}, we have $x\in \typebag$ and therefore $x\in D$. Due to Definition~\ref{def:introduce}\ref{it:intro-g}, we have $A\subseteq \beta(u)$.  By Definition~\ref{def:introduce}\ref{it:intro-h}, $\typent_1[A]=1$ and therefore, $\solution{u}{\typpe_1}$ contains a vertex $w$ such that $w\in V(G_{\leq u})$ and $w$ is close to $A$ with respect to $\beta(u)$. Since $v$ is an introduce node with $\beta(v)=\beta(u)\cup \set{x}$, $w$ must be close to $A$ with respect to $\beta(v)$. Hence, Definition~\ref{def:type}\ref{it:type-3} is satisfied. Hence, $\solution{v}{\typpe}$ is a certificate of $(v,\typpe)$. Now, Lemma~\ref{lem:introduce-minimal} implies that $\solution{v}{\typpe}$ is minimum.

\item Now, assume that $v$ is a forget node. Let $u$ be the child of $v$ and $\beta(v) = \beta(u) \setminus \set{x}$. First, we show that $\solution{v}{\typpe}$ is a certificate of $(v,\typpe)$. Let $\typpe_1$ be a type associated with $u$ such that $\solution{v}{\typpe} = \solution{u}{\typpe_1}$ and consider the set $D = \solution{u}{\typpe_1} \cup \Svertices{v}{\typpe}$. By Definition~\ref{def:forget}\ref{it:forget-a}, we have that $\typebag = \typebag_1 \setminus \set{x}$. 
\sloppy Hence, $D\cap \beta(v) = \solution{v}{\typpe} \cap \beta(v) = \solution{u}{\typpe_1} \cap (\beta(u) \setminus \set{x}) = (\solution{u}{\typpe_1} \cap \beta(u)) \setminus \set{x} = \typebag_1 \setminus \set{x} = \typebag$. Hence $D$ satisfies Definition~\ref{def:type}\ref{it:type-1}. 

Now, consider any vertex $w\in (V(G_{\leq v}) \setminus \beta(v)) \cup \typecov$. Observe that $V(G_{\leq v}) = V(G_{\leq u})$, $\beta(v) \subset \beta(u)$. Moreover, due to Definitions~\ref{def:forget}\ref{it:forget-b} and~\ref{def:forget}\ref{it:forget-c}, we have that $\typext=\typext_1$. Therefore, for any $w\in (V(G_{\leq u}) \setminus \beta(u)) \cup \typecov_1$, there exist $w_1,w_2 \in \solution{u}{\typpe_1}=\solution{v}{\typpe}$ such that $w$ is covered by $w_1,w_2$ in $\auxGraph{v}{\typpe}$. Hence, Definition~\ref{def:type}\ref{it:type-2} is satisfied.

By Definition~\ref{def:forget}\ref{it:forget-d}, for $A\subsetneq \beta(v)$, $\typent[A] = 1$ if and only if $\typent_1[A] = 1$ or $\typent_1[A \cup \set{x}] = 1$. Indeed, if some $y \in G_{\leq v}$ is close to $A$ or $A \cup \set{x}$  with respect to $\beta(u)$, then it is close to $A$ with respect to $\beta(v)$. Conversely, if there exists some $y \in G_{\leq v}$ close to $A$ with respect to $\beta(v)$, then $A$ is $B \cap \beta(v)$  where $B$ is the set to which $y$ is close to with respect to $\beta(u)$. The only possibilities for $B$ are $A$ and $A \cup \set{x}$. By Definition~\ref{def:forget}\ref{it:forget-e}, $\typent[\beta(v)] = 1$ if and only if $\typent_1[\beta(u)] = 1$, $\typent_1[\beta(v)] = 1$, $x\in \typebag_1$, or $\typent_1[\set{x}] = 1$. Indeed, if some $y \in G_{\leq v}$ is close to $\beta(u)$, $\beta(v)$ or $\set{x}$  with respect to $\beta(u)$, then it is close to $\beta(v)$ with respect to $\beta(v)$. Conversely, if there exists some $y \in G_{\leq v}$ close to $\beta(v)$ with respect to $\beta(v)$, then $\beta(v)$ is included in $A$ or $\beta(u) \setminus A$, where $A$ is the set to which $y$ is close to with respect to $\beta(u)$. The only possibilities for $A$ are $\beta(u)$, $\beta(v)$ or $\set{x}$. Hence Definition~\ref{def:type}\ref{it:type-3} is satisfied. Hence, $\solution{v}{\typpe}$ is a certificate of $(v,\typpe)$. Finally, Lemma~\ref{lem:forget-minimal} implies that $\solution{v}{\typpe}$ is minimum.

\item Assume $v$ to be a join node. Let $u_1,u_2$ be the children of $v$. Let $\typpe_1,\typpe_2$ be types associated with $u_1,u_2$ such that $\solution{v}{\typpe} = \solution{u_1}{\typpe_1} \cup \solution{u_2}{\typpe_2} $. Consider the set $D=\solution{v}{\typpe} \cup \Svertices{v}{\typpe}$. Due to Definition~\ref{def:join}\ref{it:join-a} we have $\typebag = \typebag_1=\typebag_2$. This implies $\solution{v}{\typebag} \cap \beta(v) = (\solution{u}{\typpe_1} \cap \typebag_1) \cup (\solution{u_2}{\typpe_2} \cap \typebag_2) = \typebag$.

Consider $y \in (V(G_{\leq v}) \setminus \beta(v)) \cup \typecov$. If $y \in (V(G_{\leq u_1}) \setminus \beta(u_1)) \cup \typecov_1$, then $y$ is covered by a pair of vertices $y_1$ and $y_2$ in $\solution{u}{\typpe_1} \cup \Svertices{u_1}{\typpe_1}$. If $y_1,y_2\in \solution{u_1}{\typpe_1}$ then $y_1,y_2\in D$ and we are done. Otherwise, assume without loss of generality that $y_2 \in \Svertices{u_1}{\typpe_1} \setminus \Svertices{v}{\typpe}$. There must be a set $A\subseteq \beta(u_1)$ such that $\typext_1[A]=1$ and $y_2=A$. By Definition~\ref{def:join}\ref{it:join-b}, either $\typext[A]=1$ or $\typent_2[A]=1$. If the first case is true, then there exists a vertex $y'_2\in \Svertices{v}{\typpe}$ such that $y$ is covered by $y_1$ and $y'_2$ is in $G_{\leq v}$. If the second case is true, then there exists a vertex $y'_2\in \solution{u_2}{\typpe_2}$ such that $y'_2$ is close to $A$ with respect to $\beta(u_2)$. Due to Lemma~\ref{lem:cover-clique}, we have that $y$ is covered by $y_1$ and $y'_2$. The case where $y \in (V(G_{\leq u_2}) \setminus \beta(u_2)) \cup \typecov_2$ is symmetrical.
If $y \in Cov(u_1,u_2)$, by its definition in Definition~\ref{def:join}\ref{it:join-c} and Lemma~\ref{lem:cover-clique}, $y$ is covered by vertices in $\solution{v}{\typpe}$. Hence, Definition~\ref{def:type}\ref{it:type-2} is satisfied.

By Definitions~\ref{def:join}\ref{it:join-d} and~\ref{def:join}\ref{it:join-e}, we have that for any $A\subseteq \beta(v)$, $\typent[A] = 1$ if and only if $\typent_1[A] = 1$ or $\typent_2[A] = 1$. Therefore by induction Definition~\ref{def:type}\ref{it:type-3} is satisfied. The minimality follows from Lemma~\ref{lem:join-minimal}.

\item When $v$ is the root node, the statement follows easily from Definition~\ref{def:root} and Lemma~\ref{lem:root-minimal}.
\end{enumerate}
This completes the proof.
\end{proof}

\subsection{The case of interval graphs}

When the input graph $G$ is an interval graph, the nice tree decomposition of $G$ does not contain any join nodes  because we can choose as "root" the vertex associated with the leftmost interval. Moreover, the linear structure of interval graphs helps us to reduce the time complexity of the dynamic programming algorithm proposed in the previous section. Essentially, we establish that the number of different types associated with a node $v$ is at most $O(2^{\omega (G)})$. We shall use the following lemma.

\begin{lemma}
Let $X$ be a clique cutset of an interval graph $G$. There exists a collection $\mathcal A$ of subsets of $X$ of size $O(\abs{X})$ such that for each vertex $v \in V(G)$, if $v$ is close to $A$ with respect to $X$, then $A \in \mathcal A$.
\end{lemma}

\begin{proof}
If $v \in X$, then $A = \set{v}$. Without loss of generality, assume now that $\min(v) < \min(X)$ (where $\min(v)$ denotes the left endpoint of the interval associated to $v$, and $\min(X)$, the leftmost left endpoint of an interval of $X$). If $u \in X$ such that $d(v,u) = d$, then for every $w \in X$ such that $\min(w) \leq \min(u)$, $d(v,w) \leq d$. Indeed, take a shortest path from $v$ to $u$ and let $z$ be the neighbour of $u$ in this path. Then, $z$ is also a neighbour of $w$. This implies that $v$ is close to a set $A$ with respect to $X$ which belongs to one of the following sets: $\bigcup\limits_{u \in X} \set{\set{w \in X| \min(w) \leq \min(u)}}$. Hence, $$\mathcal A = \bigcup\limits_{u \in X} \set{\set{w \in X| \min(w) \leq \min(u)}, \set{w \in X| \max(w) \geq \max(u)}, \set{u}}$$
Observe that $\abs{\mathcal{A}}$ is $O(|X|)$. 
\end{proof}

The above lemma implies that for an interval graph, the set of $5$-tuples for a node $v$ can be chosen as a subset of $\set{0,1}^{\mathcal A} \times \set{0,1}^{\mathcal A} \times 2^{\beta(v)} \times 2^{\beta(v)} \times 2^{\beta(v)}$. Hence, there are $2^{O(\omega)}$ types for an interval graph. This proves the statement of Theorem~\ref{thm:fpt-chordal} regarding interval graphs.

\section{Hardness for interval graphs}\label{sec:interval-hard}

We now prove Theorem~\ref{thm:interval-hard}.
Let $F$ be an instance of \textsc{3-Sat} with variables $x_1, \dots, x_n$ and clauses $C_1,\dots,C_m$. We construct a set $D$ of intervals in polynomial time such that the geodetic number of the intersection graph of $D$ (denoted as $\mathcal{I}(D)$) is at most $\bound$ if and only if $F$ is a positive instance of \textsc{3-Sat}. 

The key intuition that explains why the problem is hard on interval graphs is that considering two solution vertices $x$ and $y$, the structure of the covered set $I(x,y)$ can be very complicated. Indeed, it can be that many vertices lying ``in between'' $x$ and $y$ in the interval representation, are not covered. This allows us to construct gadgets, by controlling which of these vertices get covered, and which do not. Moreover, we can easily force some vertices to be part of the solution by representing them by an interval of length~$0$ (then, they are simplicial vertices), which is very useful to design our reduction. Nevertheless, implementing this idea turns out to be far from trivial, and to this end we need the crucial idea of \emph{tracks}, which are shortest paths spanning a large part of the construction. Each track starts at a key interval called its \emph{root} (representing a literal, for example) and serves as a shortest path from the root to the rightmost end of the construction. In a way, each track ``carries the effect of the root'' being chosen in a solution to the rest of the graph. The tracks are shifted in a way that no shortcut can be used going from one track to another. We are then able to locally modify the tracks and place our gadgets so that the track of, say, a literal, enables the interval of that literal to cover an interval of a specific clause gadget (while the other tracks are of no use for this purpose).

\subsection{Overview of the reduction}

There are four main stages of our reduction. Figure~\ref{fig:overview}(a) shows a roadmap of the reduction. We initialise it by constructing a set of intervals which we call the \emph{start gadget} (denoted as $\StartGadget$). 

After creating the start gadget, we create the variable gadgets, which are placed consecutively, after the start gadget. For each variable $x_i$, with $1\leq i\leq n$, we create the variable gadget $\VarGadget_i$. Each variable gadget is composed of several \emph{implication gadgets}. An implication gadget $\Implication{\Interval{p}}{\Interval{q}}$ ensures that if $\Interval{p}$ is not chosen in a geodetic set of our constructed intervals, then $\Interval{q}$ must be chosen. These are used to encode the behaviour of the variables of the 3-SAT instance: there will be two possible solutions, corresponding to both truth values of $x_i$. 


After creating all the variable gadgets, we create the clause gadgets, also placed consecutively, after the variable gadgets. For each clause $C_j$ with $1\leq j\leq m$, we construct the clause gadget $\ClauseGadget_j$. Each clause gadget is composed of a \emph{covering gadget}, several implication gadgets, and several \emph{AND gadgets}. The covering gadget of a clause $C_i$ is denoted by $\CovGadget{i}$. For two intervals $\Interval{p}$ and $\Interval{q}$, the corresponding AND gadget is denoted by $\AND{\Interval{p}}{\Interval{q}}$. Together, these gadgets will ensure that all intervals of the clause gadget corresponding to the clause $C_i$ are covered by six intervals if and only if one of the intervals corresponding to the literals of $C_i$ is chosen in a geodetic set. This encodes the behaviour of the clauses of the 3-SAT instance.   

After creating all the clause gadgets, we conclude our construction by creating the \emph{end gadget} $\EndGadget$, placed after all clause gadgets. Figure~\ref{fig:overview}(b) shows the arrangement of the gadgets in the reduction.

\begin{figure}
\centering
\begin{tabular}{c}
      
     \begin{tikzpicture}[
  node distance=1cm
]

    \node (const) at (0.5,0) [draw,thick,minimum width=2cm,minimum height=1cm] {Reduction};

    \node  [draw,thick,minimum width=2cm,minimum height=1cm] (variable) [below=of const,xshift=-3
    cm] {\begin{tabular}{c}
         Variable \\
         Gadget \\ ($\VarGadget_i$)
    \end{tabular}};

    \node  [draw,thick,minimum width=2cm,minimum height=1cm] (clause) [below=of const] {\begin{tabular}{c}
         Clause \\
         Gadget \\ ($\ClauseGadget_j$)
    \end{tabular}};

    \node  [draw,thick,minimum width=2cm,minimum height=1cm] (imply) [below=of variable] {\begin{tabular}{c}
         Implication \\
         Gadget
    \end{tabular}};

    \node  [draw,thick,minimum width=2cm,minimum height=1cm] (start) [left=of imply,xshift=0.25cm] {\begin{tabular}{c}
         Start \\
         Gadget
    \end{tabular}};

    \node  [draw,thick,minimum width=2cm,minimum height=1cm] (cover) [below=of clause] {\begin{tabular}{c}
         Covering \\
         Gadget
    \end{tabular}};

    \node  [draw,thick,minimum width=2cm,minimum height=1cm] (and) [right=of clause] {\begin{tabular}{c}
         AND \\
         Gadget
    \end{tabular}};
    
    \node  [draw,thick,minimum width=2cm,minimum height=1cm] (insert) [below=of and, yshift=-0.25cm] {\begin{tabular}{c}
         Insert \\
         Gadget
    \end{tabular}};

    \node  [draw,thick,minimum width=2cm,minimum height=1cm] (end) [right=of const,xshift=0.08cm] {\begin{tabular}{c}
         End \\
         Gadget
    \end{tabular}};

    \draw[->,thick,densely dotted] (insert.north) -- (and.south);
    \draw[->,thick,densely dotted] (cover.north) -- (clause.south);
    \draw[->,thick,densely dotted] (imply.north) -- (variable.south);
    \draw[->,thick,densely dotted] (imply.east) -- ++ (0.25,0) |- (clause.west);

    \draw[->,thick,densely dotted] (and.west) -- (clause.east);

    \draw[->,thick,densely dotted] (clause.north) -- (const.south);
    \draw[->,thick,densely dotted] (variable.north) -- ++ (0,0.5) -- ++ (2.5,0) -- ++ (0,0.5) ;

    \draw[->,thick,densely dotted] (start.north) |- (const.west);
    \draw[->,thick,densely dotted] (end.west) -- (const.east);

    \end{tikzpicture} \\ 
    (a) \\ \\
    \begin{tikzpicture}[join=bevel,inner sep=0.5mm,line width=1.2pt, scale=0.4]

\foreach \x/\y in {3/4.8, 3/4.6, 7/4.4, 7/4.2, 11/4.0, 11/3.8, 12/3.6, 12.5/3.4, 13/3.2, 17/3.0, 17/2.8, 21/2.6, 21/2.4, 26/2.2, 26.5/2.0, 27/1.8, 31/1.6, 31/1.6}
{
\draw (\x,\y) -- (32.2,\y);
}

\node (pos) at (0,0) {};
\draw [fill=white] ($(pos)+(0,0)$) rectangle ($(pos)+(3,5)$);
\node () at  ($(pos)+(1.5,2.5)$) {$\StartGadget{}$};

\node (pos) at (4,0) {};
\draw [fill=white] ($(pos)+(0,0)$) rectangle ($(pos)+(3,5)$);
\node () at  ($(pos)+(1.5,2.5)$) {$\VarGadget_1$};

\node (pos) at (8,0) {};
\draw [fill=white] ($(pos)+(0,0)$) rectangle ($(pos)+(3,5)$);
\node () at  ($(pos)+(1.5,2.5)$) {$\VarGadget_2$};

\node () at  (12.5,2.5) {$\cdots$};

\node (pos) at (14,0) {};
\draw [fill=white] ($(pos)+(0,0)$) rectangle ($(pos)+(3,5)$);
\node () at  ($(pos)+(1.5,2.5)$) {$\VarGadget_n$};

\node (pos) at (18,0) {};
\draw [fill=white] ($(pos)+(0,0)$) rectangle ($(pos)+(3,5)$);
\node () at  ($(pos)+(1.5,2.5)$) {$\ClauseGadget_1$};

\node (pos) at (22,0) {};
\draw [fill=white] ($(pos)+(0,0)$) rectangle ($(pos)+(3,5)$);
\node () at  ($(pos)+(1.5,2.5)$) {$\ClauseGadget_2$};

\node () at  (26.5,1.4) {$\cdots$};

\node (pos) at (28,0) {};
\draw [fill=white] ($(pos)+(0,0)$) rectangle ($(pos)+(3,5)$);
\node () at  ($(pos)+(1.5,2.5)$) {$\ClauseGadget_m$};

\node (pos) at (32,0) {};
\draw [fill=white] ($(pos)+(0,0)$) rectangle ($(pos)+(3,5)$);
\node () at  ($(pos)+(1.5,2.5)$) {$\EndGadget$};

\end{tikzpicture} \\ 
     (b)
\end{tabular}

\caption{Overview of the reduction. (a) Roadmap of the reduction procedure for proving Theorem~\ref{thm:interval-hard}. (b) An illustration of the arrangements of the gadgets are shown. The box with label $\StartGadget$ represents the start gadget, the box labelled $\EndGadget$ represents the end gadget. A box labelled $\ClauseGadget_i$ represents the gadget for clause $C_i$ and a box labelled $\VarGadget_j$, the gadget for variable $x_j$. Lines between such gadgets represent the tracks.}\label{fig:overview}
\end{figure}
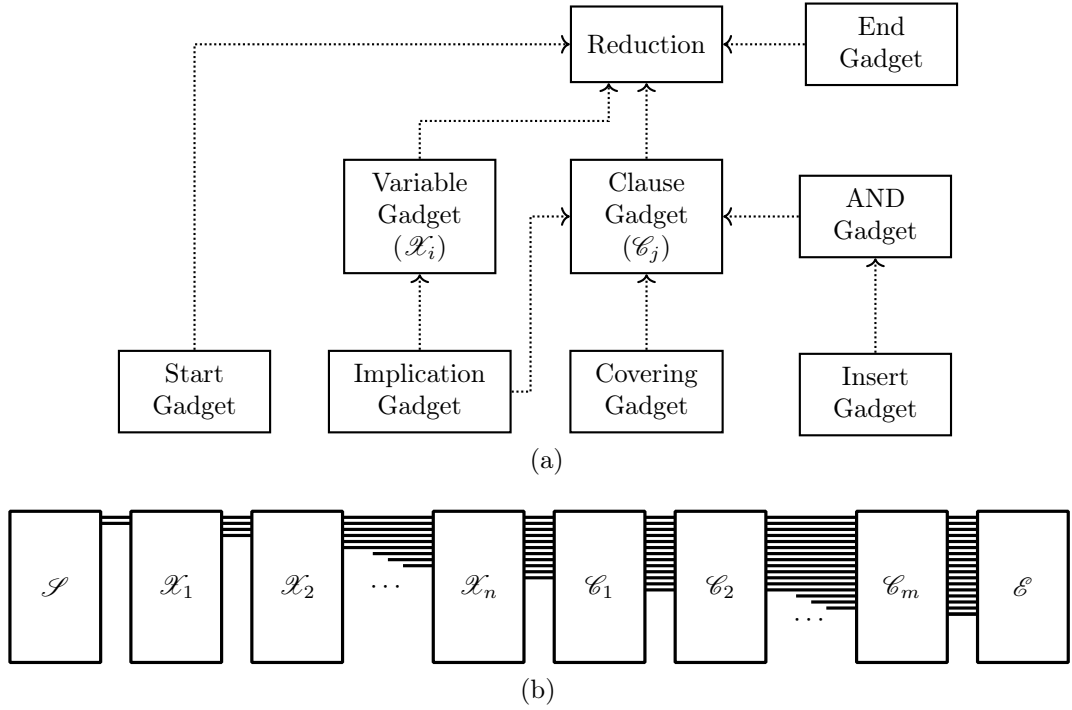

\medskip \noindent \textbf{Organisation of this section:} In Subsection~\ref{sec:notations-rev}, we introduce some notations to use them in the description of the reduction. In Subsection~\ref{sec:start-rev}, we describe the construction of the start gadget. In Subsection~\ref{sec:implication-rev}, we describe a generic construction of the implication gadget $\Implication{p}{q}$. In Subsection~\ref{sec:cover-gadget-rev}, we describe a generic construction of the cover gadget $\CovGadget{i}$. In Subsection~\ref{sec:insert-gadget-rev}, we describe a generic construction of the insert gadget $\INS{p}{q}$.  In Subsection~\ref{sec:and-gadget-rev}, we describe a generic construction of the AND-gadget $\AND{p}{q}$. Then, in Subsection~\ref{sec:variable-rev}, we describe the construction of the variable gadget. In Subsection~\ref{sec:clause-gadget-rev}, we describe the construction of the clause gadget. In Subsection~\ref{sec:end-gadget-rev}, we describe the construction of the end gadget.

\subsection{Notations}\label{sec:notations-rev}

We shall use the following notations. Let $S$ be a set of intervals. For a vertex $v\in V(\mathcal{I}(S))$, the corresponding interval will be denoted by $\Interval{v}$. The notations $min(\Interval{v}), max(\Interval{v})$ shall denote the left boundary and right boundary of $\Interval{v}$, respectively. The \emph{rightmost neighbour} of an interval $\Interval{v}$ is the interval intersecting $\Interval{v}$ that has the maximum right boundary. For a set $S$ of intervals, let $min(S)=\min\{ min(\Interval{v}) \colon \Interval{v}\in S \}$, $max(S)=\max\{ max(\Interval{v}) \colon \Interval{v}\in S \}$. For two intervals $\Interval{u},\Interval{v}$ we have $\Interval{u} < \Interval{v}$ if $max(\Interval{u}) < min(\Interval{v})$.

Let $S$ be a set of intervals and $\Interval{u},\Interval{v}\in S$. A shortest path between $\Interval{u},\Interval{v}$ is a shortest path between $u,v$ in $\mathcal{I}(S)$. The set $I(\Interval{u},\Interval{v})$ is the set of intervals that belong to some shortest path between $\Interval{u},\Interval{v}$. The geodetic set of $S$ is analogously defined. For a set of intervals $S$, the phrase ``$S$ is covered by $S'$'' means that $S'$ is a geodetic set of $S$.

A \emph{point interval} is an interval of the form $[a,a]$. A \emph{unit interval} is an interval of the form $[a,a+1]$. A set of intervals is \emph{good} if no two intervals contain each other. A set $T=\{\Interval{u}_1, \Interval{u_2}, \ldots, \Interval{u_t}\}$ of intervals is a \emph{track} if $\max(\Interval{u_i}) = \min\left(\Interval{u_{i+1}}\right)$ for all $1\leq i<t$. Observe that if $T$ is a track, then $\mathcal{I}(T)$ is a path. In our construction, each track $T$ will be associated with a set of intervals called its \emph{roots}, denoted by $\Root{T}$. Sometimes we shall use the sentence ``root $\Interval{v}$ of a track $T$'' to say $\Interval{v} \in \Root{T}$. 

\begin{definition}
Let $T$ and $T'$ be two tracks such that $T\cup T'$ is a good set of intervals. Then, $T < T'$ if $max(T) < max(T')$. 
\end{definition}

Let $\mathcal{T}$ be a set of tracks and $T\in \mathcal{T}$. The phrase ``\emph{the track just preceding $T$}'' shall refer to the track $T'$ such that $T'<T$ and there is no $T''$ such that $T'<T''<T$. The phrases ``\emph{the track just following $T$}'', ``\emph{maximal track of $\mathcal{T}$}'' and ``\emph{minimal track of $\mathcal{T}$}'' are analogously defined.

\subsection{Initiation and construction of start gadget $\StartGadget$}\label{sec:start-rev} 

Let $F$ be an instance of \textsc{3-Sat} with variables $x_1, \dots, x_n$ and clauses $C_1,\dots,C_m$. Let $\epsilon=\frac{1}{(n+m)^4}$. The start gadget $\StartGadget$ consists of four intervals which are defined as follows: the \emph{start interval} $\startInterval = \interval{1,1}$, $\Interval{u_\startVertex} = \interval{1,2}$, the \emph{true interval} $\trueInterval = \interval{1+\epsilon,1+\epsilon}$ and $\Interval{ u_{\trueVertex} } = \interval{1+\epsilon,2+\epsilon}$. Let $T_1=\{\Interval{u_\startVertex} \}$ and $T_2=\{\Interval{u_\trueVertex}\}$. Observe that $T_1,T_2$ are tracks and $T_1 < T_2$. See Figure~\ref{fig:start-gadget}.

\begin{figure}
\centering
\begin{tikzpicture}[join=bevel,inner sep=0.5mm,line width=1.2pt, scale=0.4]
\node[draw=none] (s1) at (0,1) {};
\node[left=0mm of s1] {$\startInterval$};
\TikzSimplicial{(s1)}
\node[draw=none] (s1) at (1,1) {};
\node[right=0mm of s1] {$\trueInterval$};
\TikzSimplicial{(s1)}
\node[draw=none] (q) at (5,0) {$\Interval{u_\startInterval}$};
\TikzInterval{(0,0)}{(4,0)}
\node[draw=none] (q) at (6,-1) {$\Interval{u_\trueInterval}$};
\TikzInterval{(1,-1)}{(5,-1)}
\end{tikzpicture}
\caption{The start gadget $\StartGadget$. For drawing purposes, the proportions of the intervals were changed. Nevertheless, the obtained interval graph is unchanged.}\label{fig:start-gadget}
\end{figure}

We initialize two more sets, the set $\mathcal{T}=\{T_1,T_2\}$, and the set $D = \StartGadget$. In what follows, $\mathcal{T}$ will contain all constructed tracks and $D$ will contain all constructed intervals. As we proceed with the construction, we shall insert more intervals in $T_1,T_2$ while maintaining that both of them are tracks. We shall also add more tracks in $\mathcal{T}$. Let  $\Root{T_1}=\{\startInterval\}$ and $\Root{T_2}=\{\trueInterval\}$. Recall that for a track $T$, $\Root{T}$ denotes the set of root intervals of $T$.

\subsection{Implication gadget of a root $\Interval{p}$}\label{sec:implication-rev}

In order to construct the variable gadgets and the clause gadgets, we need to define the implication gadget.
Below we describe a generic procedure to construct implication gadgets of a root $\Interval{p}$ which is different from $\startInterval$ of $\StartGadget$. Let $T_{\Interval{p}}\in \mathcal{T}$ be the track such that $\Interval{p} \in \Root{T_{\Interval{p}}}$. Since $\Interval{p} \neq \startInterval$, $T_{\Interval{p}}$ is not the minimal element in $\mathcal{T}$. Below we describe the three steps for constructing $\Implication{p}{q}$. See Figure~\ref{fig:implication} for a graphical representation of the intervals.

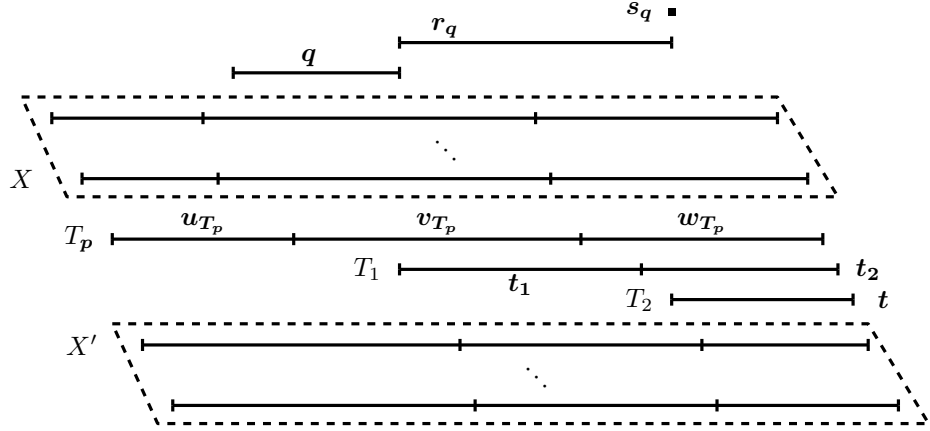
\begin{figure}
    \centering
    \begin{tikzpicture}[join=bevel,inner sep=0.5mm,line width=1.2pt, scale=0.4]
\node[draw=none] (s1) at (18.5,11.5) {};
\node[left=1mm of s1] {$\Interval{s_q}$};
\TikzSimplicial{(s1)}
\node[draw=none] (q) at (6.5,10) {$\Interval{q}$};
\TikzInterval{(4,9.5)}{(9.5,9.5)}
\node[draw=none] (r) at (11,11) {$\Interval{r_q}$};
\TikzInterval{(9.5,10.5)}{(18.5,10.5)}
\node (pos) at (-3,8.7) {};
\draw[dashed] ($(pos)+(0,0)$) -- ($(pos)+(25,0)$) -- ($(pos)+(27,-3.3)$) -- ($(pos)+(1.5,-3.3)$) -- ($(pos)+(0,0)$); 
\node[draw=none] (t) at ($(pos)+(1,-0.7)$) {};
\TikzInterval{(t.center)}{($(t.center)+(24,0)$)}
\TikzIntSep{($(t.center)+(5,0)$)}
\TikzIntSep{($(t.center)+(16,0)$)}
\node[draw=none] (t) at ($(pos)+(2,-2.7)$) {};
\TikzInterval{(t.center)}{($(t.center)+(24,0)$)}
\TikzIntSep{($(t.center)+(4.5,0)$)}
\TikzIntSep{($(t.center)+(15.5,0)$)}
\node[rotate=135] at ($(pos)+(14,-1.7)$) {$\cdots$};
\node () at ($(pos)+(0,-2.7)$) {$X$};
\node[draw=none] (tl) at (0,4) {};
\node[left=1mm of tl] {$T_{\Interval{p}}$};
\TikzInterval{(tl.center)}{(23.5,4)}
\TikzIntSep{(6,4)}
\TikzIntSep{(15.5,4)}
\node () at (3,4.5) {$\Interval{u_{T_{\Interval{p}}}}$};
\node () at (10.75,4.5) {$\Interval{v_{T_{\Interval{p}}}}$};
\node () at (19.5,4.5) {$\Interval{w_{T_{\Interval{p}}}}$};

\node[draw=none] (tnew1) at (18.5,2) {};
\node[left=1mm of tnew1] {$T_2$};
\node at (25.5,2) {$\Interval{t}$};
\TikzInterval{(tnew1.center)}{(24.5,2)}

\node[draw=none] (tnew2) at (9.5,3) {};
\node[left=1mm of tnew2] {$T_1$};
\node at (25,3) {$\Interval{t_2}$};
\node at (13.5,2.5) {$\Interval{t_1}$};
\TikzInterval{(tnew2.center)}{(24,3)}
\TikzIntSep{(17.5,3)}

\node (pos) at (0,1.2) {};
\draw[dashed] ($(pos)+(0,0)$) -- ($(pos)+(25,0)$) -- ($(pos)+(27,-3.3)$) -- ($(pos)+(1.5,-3.3)$) -- ($(pos)+(0,0)$); 
\node[draw=none] (t) at ($(pos)+(1,-0.7)$) {};
\TikzInterval{(t.center)}{($(t.center)+(24,0)$)}
\TikzIntSep{($(t.center)+(10.5,0)$)}
\TikzIntSep{($(t.center)+(18.5,0)$)}
\node[draw=none] (t) at ($(pos)+(2,-2.7)$) {};
\TikzInterval{(t.center)}{($(t.center)+(24,0)$)}
\TikzIntSep{($(t.center)+(10,0)$)}
\TikzIntSep{($(t.center)+(18,0)$)}
\node[rotate=135] at ($(pos)+(14,-1.7)$) {$\cdots$};
\node () at ($(pos)+(-1,-0.7)$) {$X'$};

\end{tikzpicture}
    \caption{The implication gadget $\Implication{p}{q}$. For drawing purposes, the proportions of the intervals were changed. Nevertheless, the obtained interval graph is unchanged.}
    \label{fig:implication}
\end{figure}

\begin{enumerate}

\item \sloppy \textbf{Extension of existing tracks}: For each track $T\in \mathcal{T}$, introduce three new intervals $\Interval{u_T}=\interval{max(T),max(T)+1}$, $\Interval{v_T}=\interval{max(T)+1,max(T)+2}$ and $\Interval{w_T}=\interval{max(T)+2,max(T)+3}$. Let $T_{new}=\{ \Interval{u_T}, \Interval{v_T}, \Interval{w_T} \}$. Observe that, for two tracks $T,T'\in \mathcal{T}$ with $T < T'$, we have $(T \cup T_{new}) < (T'\cup T'_{new})$. 

\item \textbf{Creation of new intervals}: Let $X$ and $X'$ be the tracks that precede and follow $T_{\Interval{p}}$ in $\mathcal{T}$, respectively. Note that $X$ always exists since $\Interval{p}\neq \Interval{o}$. When $X'$ exists, let $\theta = max(\Interval{u_{X'}})$ and $\theta'=max(\Interval{v_{X'}})$. Otherwise, $\theta=max(\Interval{u_{T_{\Interval{p}}}}) + \epsilon$ and $\theta'=max(\Interval{v_{T_{\Interval{p}}}}) + \epsilon$. 

Define $\Interval{q}=\interval{\frac{max(\Interval{u_X}) + max(\Interval{u_{T_{\Interval{p}}}})}{2}, \frac{max(\Interval{u_{T_{\Interval{p}}}}) + \theta }{2}}$, 
$\Interval{r_q}=\left[ max(\Interval{q}), \frac{max(\Interval{v_{T_{\Interval{p}}}}) + \theta'}{2} \right]$ and $\Interval{s_q}=[max(\Interval{r_q}),max(\Interval{r_q})]$.

\item \textbf{Creation of new tracks}: In this step, we shall create two new tracks.
We define three more intervals as follows: $\Interval{t}=\interval{max(\Interval{s_q}),max(\Interval{s_q})+1}$, $\Interval{t_1}=\interval{max(\Interval{q}),\frac{max(\Interval{v_{T_{\Interval{p}}}}) + min(\Interval{s_q})}{2}}$ and $\Interval{t_2}=\interval{max(\Interval{t_1}),max(\Interval{t_1})+1}$.
Now, let  $T_1=\{ \Interval{t_1},\Interval{t_2} \}$, $\Root{T_1}=\{ \Interval{q}\}$, $T_{2}=\{\Interval{t}\}$ and $\Root{T_{2}}=\{\Interval{r_q},\Interval{s_q}\}$. 
\end{enumerate} 

To complete the construction of the implication gadget, we define $\Implication{p}{q} = \{ \Interval{q}, \Interval{r_q}, \Interval{s_q} \} \cup T_1 \cup T_2 \cup \bigcup\limits_{T\in \mathcal{T}} \{\Interval{u_T}, \Interval{v_T}, \Interval{w_T} \} $. Let $D=D \cup \Implication{p}{q}$. For each $T\in \mathcal{T}$, let $T=T\cup T_{new}$ and $\mathcal{T} = \mathcal{T} \cup \{T_1,T_2\}$. Observe that the intersection graph of $D$ does not contain $K_{1,5}$ as induced subgraph.


\subsection{Construction of covering gadgets}\label{sec:cover-gadget-rev}

Below we describe the three steps for constructing the covering gadget of the clause $C_i$. See Figure~\ref{fig:covering}.

\begin{figure}
    \centering
    \begin{tikzpicture}[join=bevel,inner sep=0.5mm,line width=1.2pt, scale=0.35]
\node[draw=none] () at (3,1) {$\CGc{i}$};
\node[draw=none] () at (3,2) {$\CGb{i}$};
\node[draw=none] () at (3,3) {$\CGa{i}$};
\TikzInterval{(4,1)}{(11,1)}
\TikzInterval{(4,2)}{(10,2)}
\TikzInterval{(4,3)}{(9,3)}
\node[draw=none] (s1) at (4,4) {};
\node[left=1mm of s1] {$\CGd{i}$};
\TikzSimplicial{(s1)}
\node[draw=none] () at (26,1.5) {$\CGCov{i}$};
\TikzInterval{(24,1)}{(30,1)}
\node[draw=none] (s1) at (30,2) {};
\node[left=1mm of s1] {$\CGf{i}$};
\TikzSimplicial{(s1)}
\node (pos) at (0,0) {};
\draw[dashed] ($(pos)+(0,0)$) -- ($(pos)+(35,0)$) -- ($(pos)+(37,-3.3)$) -- ($(pos)+(1.5,-3.3)$) -- ($(pos)+(0,0)$); 
\node[draw=none] (t) at ($(pos)+(1,-0.7)$) {};
\TikzInterval{(t.center)}{($(t.center)+(34,0)$)}
\TikzIntSep{($(t.center)+(12,0)$)}
\TikzIntSep{($(t.center)+(18,0)$)}
\node[draw=none] (t) at ($(pos)+(2,-2.7)$) {};
\TikzInterval{(t.center)}{($(t.center)+(34,0)$)}
\TikzIntSep{($(t.center)+(12,0)$)}
\TikzIntSep{($(t.center)+(18,0)$)}
\node[rotate=135] at ($(pos)+(16,-1.7)$) {$\cdots$};

\node[draw=none] (t) at (3,-4.5) {};
\node[draw=none] () at (3,-4.5) {$T_{\CGd{i}}$};
\TikzInterval{(4,-4.5)}{($(t.center)+(34,0)$)}
\TikzIntSep{($(t.center)+(12,0)$)}
\TikzIntSep{($(t.center)+(18,0)$)}

\node[draw=none] (t) at (3.5,-5.5) {};
\node[draw=none] () at (8,-5.5) {$T_{\CGa{i}}$};
\TikzInterval{(9,-5.5)}{($(t.center)+(34,0)$)}
\TikzIntSep{($(t.center)+(12,0)$)}
\TikzIntSep{($(t.center)+(22,0)$)}

\node[draw=none] (t) at (4,-6.5) {};
\node[draw=none] () at (9,-6.5) {$T_{\CGb{i}}$};
\TikzInterval{(10,-6.5)}{($(t.center)+(34,0)$)}
\TikzIntSep{($(t.center)+(12,0)$)}
\TikzIntSep{($(t.center)+(22,0)$)}

\node[draw=none] (t) at (4.5,-7.5) {};
\node[draw=none] () at (10,-7.5) {$T_{\CGc{i}}$};
\TikzInterval{(11,-7.5)}{($(t.center)+(34,0)$)}
\TikzIntSep{($(t.center)+(12,0)$)}
\TikzIntSep{($(t.center)+(22,0)$)}

\node[draw=none] (t) at (5,-8.5) {};
\node[draw=none] () at (29,-8.5) {$T_{\CGf{i}}$};
\TikzInterval{(30,-8.5)}{($(t.center)+(34,0)$)}




\end{tikzpicture}
    \caption{The covering gadget \CovGadget{i}. For drawing purposes, the proportions of the intervals were changed. Nevertheless, the obtained interval graph is unchanged.}
    \label{fig:covering}
\end{figure}
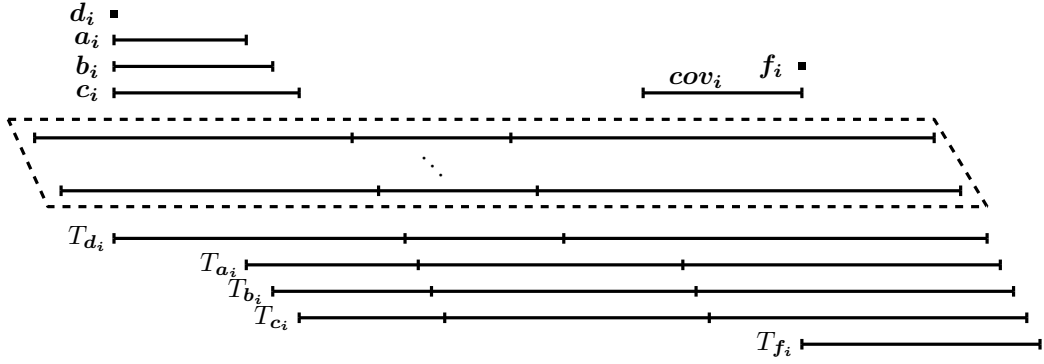

\begin{enumerate}
\item \sloppy \textbf{Extension of existing tracks:} For each track $T\in \mathcal{T}$, introduce three new intervals $\Interval{u^i_T}=\interval{max(T),max(T)+1}$, $\Interval{v^i_T}=\interval{max(T)+1,max(T)+2}$ and $\Interval{w^i_T}=\interval{max(T)+2,max(T)+3}$. Let $T_{new}=\{ \Interval{u^i_T}, \Interval{v^i_T}, \Interval{w^i_T} \}$. Observe that, for two tracks $T,T'\in \mathcal{T}$ with $T < T'$, we have $(T \cup T_{new}) < (T'\cup T'_{new})$. 

\item \sloppy \textbf{Creation of new intervals:} Let $T$ be the maximal track in $\mathcal{T}$. Let $\theta = min(\Interval{u^i_T})+\epsilon$. We define 
$\CGa{i} = \interval{\theta, \theta+\epsilon}$, 
$\CGb{i} = \interval{\theta,\theta+2\epsilon}$, 
$\CGc{i} = \interval{\theta, \theta+3\epsilon}$ and
$\CGd{i} = \interval{\theta,\theta}$.
Also define $\CGCov{i} = \interval{max(\Interval{v^i_T}) + 4\epsilon, max(\Interval{v^i_T}) + 7\epsilon}$,  and $\CGf{i} = \interval{ max(\CGCov{i}), max(\CGCov{i})} $. 
 
\item \sloppy \textbf{Creation of new tracks:} Now we create five more tracks as follows. 
Let 
$T_{\CGa{i}} = \set{ \interval{max(\CGa{i})+k , max(\CGa{i})+k+1} |~k \in \set{0,1,2}}$, 
$T_{\CGb{i}} = \set{ \interval{max(\CGb{i})+k , max(\CGb{i})+k+1} |~k \in \set{0,1,2}}$ and
$T_{\CGc{i}} = \set{ \interval{max(\CGc{i})+k , max(\CGc{i})+k+1} |~k \in \set{0,1,2}}$. 
Observe that $T_{\CGa{i}},T_{\CGb{i}},T_{\CGc{i}}$ are tracks and define $\Root{T_{\CGa{i}}} = \set{\CGa{i}}, \Root{T_{\CGb{i}}} = \set{\CGb{i}}$ and $\Root{T_{\CGc{i}}} = \set{\CGc{i}}$. 

Also, define $T_{\CGd{i}} = \set{ \interval{max(\CGd{i})+k , max(\CGd{i})+k+1} |k \in \set{0,1,2}}$ where $\Root{T_{\CGd{i}}} = \set{\CGd{i}}$ and $T_{\CGf{i}} = \set {  \interval{max(\CGf{i}), max(\CGf{i})+1 } } $ where $\Root{T_{\CGf{i}}} = \set{\CGCov{i},\CGf{i}}$.
\end{enumerate}

\sloppy To complete the construction of the covering gadget of $C_i$, we define 
$\CovGadget{i} =  \set{ \CGa{i} , \CGb{i}, \CGc{i}, \CGCov{i}, \CGd{i}, \CGf{i} } \cup \{ \Interval{u^i_T}, \Interval{v^i_T}, \Interval{w^i_T} \}_{T\in \mathcal{T}} \cup \left(\displaystyle\bigcup\limits_{\Interval{y}\in \{\Interval{a},\Interval{b},\Interval{c},\Interval{d},\Interval{f}\}} T_{\CGy{i}}\right)$. For each $T\in \mathcal{T}$, let $T=T\cup T_{new}$ and $\mathcal{T} = \mathcal{T} \cup \set{T_{\CGa{i}},T_{\CGb{i}},T_{\CGc{i}},T_{\CGd{i}},T_{\CGf{i}}}$. We set $D = D \cup \CovGadget{i}$. Observe that the intersection graph of $D$ does not contain $K_{1,5}$ as induced subgraph.

\subsection{Construction of the insert gadget}\label{sec:insert-gadget-rev}

Let $T_{\Interval{p}}$ and $T_{\Interval{q}}$ be two tracks of $\mathcal{T}$ with roots $\Interval{p}$ and $\Interval{q}$, respectively. Without loss of generality, assume that $T_{\Interval{p}} < T_{\Interval{q}}$. Below we describe the three steps for constructing the insert gadget $\INS{p}{q}$. See Figure~\ref{fig:insert}.

\begin{figure}
    \centering
    \begin{tikzpicture}[join=bevel,inner sep=0.5mm,line width=1.2pt, scale=0.4]
\node[draw=none] (s1) at (9.5,10) {};
\node[left=1mm of s1] {$\InsertSigma{p}{q}$};
\TikzSimplicial{(s1)}
\node (pos) at (-3,8.7) {};
\draw[dashed] ($(pos)+(0,0)$) -- ($(pos)+(25,0)$) -- ($(pos)+(27,-3.3)$) -- ($(pos)+(1.5,-3.3)$) -- ($(pos)+(0,0)$); 
\node[draw=none] (t) at ($(pos)+(1,-0.7)$) {};
\TikzInterval{(t.center)}{($(t.center)+(24,0)$)}
\node[draw=none] (t) at ($(pos)+(2,-2.7)$) {};
\TikzInterval{(t.center)}{($(t.center)+(24,0)$)}
\node[rotate=135] at ($(pos)+(14,-1.7)$) {$\cdots$};
\node () at ($(pos)+(0,-2.7)$) {$T_{\Interval{p}}$};

\node[draw=none] (tnew2) at (9.5,4.5) {};
\node[left=1mm of tnew2] {$T_m$};
\TikzInterval{(tnew2.center)}{(24,4.5)}

\node (pos) at (0,3.2) {};
\draw[dashed] ($(pos)+(0,0)$) -- ($(pos)+(25,0)$) -- ($(pos)+(27,-3.3)$) -- ($(pos)+(1.5,-3.3)$) -- ($(pos)+(0,0)$); 
\node[draw=none] (t) at ($(pos)+(1,-0.7)$) {};
\TikzInterval{(t.center)}{($(t.center)+(24,0)$)}
\node[draw=none] (t) at ($(pos)+(2,-2.7)$) {};
\TikzInterval{(t.center)}{($(t.center)+(24,0)$)}
\node[draw=none] (t) at ($(pos)+(1.5,-1.7)$) {};
\TikzInterval{(t.center)}{($(t.center)+(24,0)$)}
\node () at ($(pos)+(-.5,-1.9)$) {$T_{\Interval{q}}$};
\node at ($(pos)+(14,-1.2)$) {$\cdots$};
\node at ($(pos)+(14,-2.2)$) {$\cdots$};
\node () at ($(pos)+(-1,-0.7)$) {$X$};

\end{tikzpicture}
    \caption{The insert gadget $\INS{p}{q}$. For drawing purposes, the proportions of the intervals were changed. Nevertheless, the obtained interval graph is unchanged.}
    \label{fig:insert}
\end{figure}
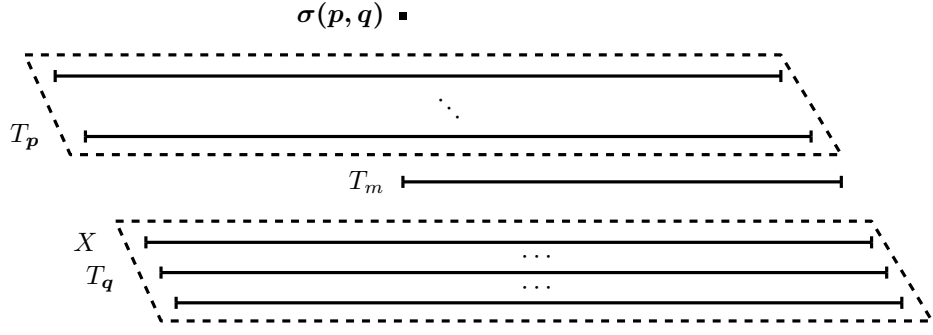

\begin{enumerate}
\item \sloppy \textbf{Extension of existing tracks:} For each track $T\in \mathcal{T}$, we introduce one new interval $\Interval{u_T}=\interval{max(T),max(T)+1}$. Let $T_{new}=\set{ \Interval{u_T} } $. Observe that, for two tracks $T,T'\in \mathcal{T}$ with $T < T'$, we have $(T \cup T_{new}) < (T'\cup T'_{new})$.

\item \sloppy \textbf{Creation of a new interval:} Let $X$ be the track that just follows $T_{\Interval{p}}$ in $\mathcal{T}$. Observe that $X$ always exists. Let $\InsertSigma{p}{q} = \interval{\frac{max(\Interval{u}_{T_{\Interval{p}}}) + max(\Interval{u}_{X}) }{2} , \frac{max(\Interval{u}_{T_{\Interval{p}}}) + max(\Interval{u}_{X}) }{2}}$.

\item \sloppy \textbf{Creation of a new track:} Let $T_m = \set{  \interval{max(\InsertSigma{p}{q}), max(\InsertSigma{p}{q})+1}}$ and $\Root{T_m} = \set{\InsertSigma{p}{q}}$. 
\end{enumerate}

To complete the construction, we define $\INS{p}{q} = \set{ \InsertSigma{p}{q}}\ \cup\ \set{ \Interval{u_T}, \Interval{v_T} }_{T\in \mathcal{T}}\ \cup\ T_m$. We set $D=D\ \cup\ \INS{p}{q}$. For each $T\in \mathcal{T}$, let $T=T\cup T_{new}$ and $\mathcal{T} = \mathcal{T} \cup \{ T_m\}$. Observe that $T_{\Interval{p}} < T_m < T_{\Interval{q}}$ in $\mathcal{T}$. Moreover, observe that the intersection graph of $D$ does not contain $K_{1,5}$ as induced subgraph.

\subsection{Construction of AND gadgets}\label{sec:and-gadget-rev}
Let $T_{\Interval{p}}$ and $T_{\Interval{q}}$ be two tracks of $\mathcal{T}$ with roots $\Interval{p}$ and $\Interval{q}$, respectively. Without loss of generality, assume that $T_{\Interval{p}} < T_{\Interval{q}}$. Below we describe the four steps of constructing $\AND{\Interval{p}}{\Interval{q}}$. See Figure~\ref{fig:and}.

\begin{figure}
    \centering
    \begin{tikzpicture}[join=bevel,inner sep=0.5mm,line width=1.2pt, scale=0.4]
\node[draw=none] (s1) at (9,11.5) {};
\node[left=1mm of s1] {$\betaAND{p}{q}$};
\TikzSimplicial{(s1)}
\node[draw=none] (q) at (7,10) {$\alphaAND{p}{q}$};
\TikzInterval{(4,9.5)}{(10,9.5)}
\node[draw=none] (q) at (22,10) {$\deltaAND{p}{q}$};
\TikzInterval{(19,9.5)}{(25,9.5)}
\node[draw=none] (r) at (14.5,11) {$\gammaAND{p}{q}$};
\TikzInterval{(10,10.5)}{(19,10.5)}
\node (pos) at (-3,8.7) {};
\draw[dashed] ($(pos)+(0,0)$) -- ($(pos)+(35,0)$) -- ($(pos)+(37,-3.3)$) -- ($(pos)+(1.5,-3.3)$) -- ($(pos)+(0,0)$); 
\node[draw=none] (t) at ($(pos)+(1,-0.7)$) {};
\TikzInterval{(t.center)}{($(t.center)+(34,0)$)}
\TikzIntSep{($(t.center)+(5,0)$)}
\node[draw=none] (t) at ($(pos)+(2,-2.7)$) {};
\TikzInterval{(t.center)}{($(t.center)+(34,0)$)}
\TikzIntSep{($(t.center)+(4.5,0)$)}
\node[rotate=135] at ($(pos)+(19,-1.7)$) {$\cdots$};
 \node () at ($(pos)+(0,-2.7)$) {$Y_1$};
\node[draw=none] (tl) at (0,4) {};
\node[left=1mm of tl] {$T_{\Interval{p}}$};
\TikzInterval{(tl.center)}{(33.5,4)}
\TikzIntSep{($(tl.center)+(6,0)$)}

\node[draw=none] (tl) at (9,3) {};
\node[left=1mm of tl] {$T_{1}$};
\TikzInterval{(tl.center)}{(34,3)}

\node[draw=none] (tl) at (0.5,2) {};
\node[left=1mm of tl] {$T_{m}$};
\TikzInterval{(tl.center)}{(34.5,2)}
\TikzIntSep{($(tl.center)+(13,0)$)}



\node (pos) at (0,1.2) {};
\draw[dashed] ($(pos)+(0,0)$) -- ($(pos)+(35,0)$) -- ($(pos)+(37,-3.3)$) -- ($(pos)+(1.5,-3.3)$) -- ($(pos)+(0,0)$); 
\node[draw=none] (t) at ($(pos)+(1,-0.7)$) {};
\TikzInterval{(t.center)}{($(t.center)+(34,0)$)}
\TikzIntSep{($(t.center)+(13.5,0)$)}
\node[draw=none] (t) at ($(pos)+(2,-2.7)$) {};
\TikzInterval{(t.center)}{($(t.center)+(34,0)$)}
\TikzIntSep{($(t.center)+(13,0)$)}
\node[rotate=135] at ($(pos)+(19,-1.7)$) {$\cdots$};
 \node () at ($(pos)+(0,-2.7)$) {$Y_2$};

\node[draw=none] (tl) at (17,-3) {};
\node[left=1mm of tl] {$T_{2}$};
\TikzInterval{(tl.center)}{(36.5,-3)}

\node[draw=none] (tl) at (3,-4) {};
\node[left=1mm of tl] {$T_{\Interval{q}}$};
\TikzInterval{(tl.center)}{(37,-4)}
\TikzIntSep{($(tl.center)+(19,0)$)}

\node[draw=none] (tl) at (24,-5) {};
\node[left=1mm of tl] {$T_{3}$};
\TikzInterval{(tl.center)}{(37.5,-5)}

\node (pos) at (3,-5.8) {};
\draw[dashed] ($(pos)+(0,0)$) -- ($(pos)+(35,0)$) -- ($(pos)+(37,-3.3)$) -- ($(pos)+(1.5,-3.3)$) -- ($(pos)+(0,0)$); 
\node[draw=none] (t) at ($(pos)+(1,-0.7)$) {};
\TikzInterval{(t.center)}{($(t.center)+(34,0)$)}
\TikzIntSep{($(t.center)+(26.5,0)$)}
\node[draw=none] (t) at ($(pos)+(2,-2.7)$) {};
\TikzInterval{(t.center)}{($(t.center)+(34,0)$)}
\TikzIntSep{($(t.center)+(26,0)$)}
\node[rotate=135] at ($(pos)+(19,-1.7)$) {$\cdots$};
 \node () at ($(pos)+(-1,-0.7)$) {$Y'_2$};

\end{tikzpicture}
    \caption{The AND gadget $\AND{p}{q}$. For drawing purposes, the proportions of the intervals were changed. Nevertheless, the obtained interval graph is unchanged.}
    \label{fig:and}
\end{figure}
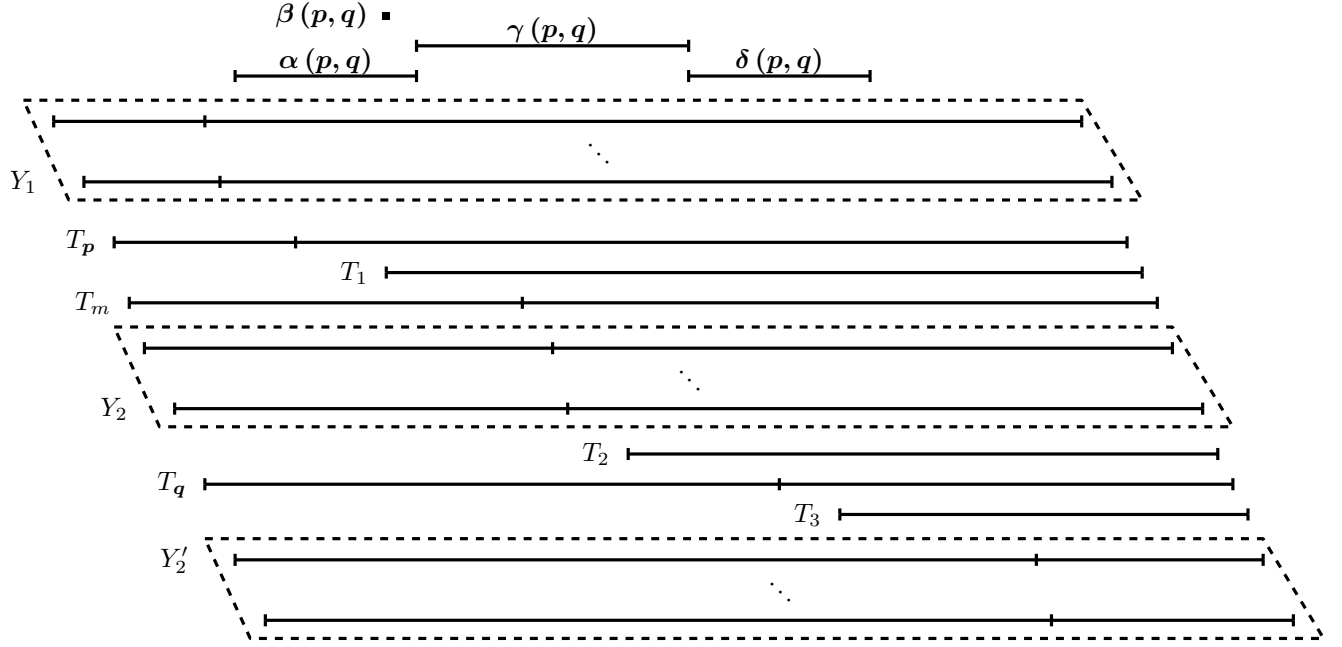

\begin{enumerate}

\item \sloppy \textbf{Creation of insert gadget:} Create an insert gadget $\INS{\Interval{p}}{\Interval{q}}$. Recall that the interval named $\InsertSigma{p}{q}$ exists and it is the root of track, say, $T_m$. 

\item \sloppy \textbf{Extension of existing tracks:} For each track $T\in \mathcal{T}$, introduce two new intervals $\Interval{u}_T=[max(T),max(T)+1]$ and $\Interval{v}_T=[max(T)+1,max(T)+2]$. Let $T_{new}=\{ \Interval{u}_T, \Interval{v}_T \}$. Observe that, for two tracks $T,T'\in \mathcal{T}$ with $T < T'$, we have $(T \cup T_{new}) < (T'\cup T'_{new})$.

\item \sloppy \textbf{Creation of new intervals:} Let $Y_1$  be the track just preceding $T_{\Interval{p}}$. Recall that $T_m$ is the track just following $T_{\Interval{p}}$.
Define $\alphaAND{p}{q} = \interval{ \frac{ max(\Interval{u_{Y_1}}) + max(\Interval{u_{T_{\Interval{p}}}}) }{2} , \frac{ max(\Interval{u_{T_{\Interval{p}}}}) + max(\Interval{u_{T_m}}) }{2} } $ 
and $\betaAND{p}{q} = \interval{ \frac{max(\Interval{u_{T_{\Interval{p}}}}) + max(\alphaAND{p}{q})}{2} , \frac{max(\Interval{u}_{T_{\Interval{p}}}) +  max(\alphaAND{p}{q}) }{2} }$. 

Let $Y_2$ be the track just preceding $T_{\Interval{q}}$ in $\mathcal{T}$. Observe that either $Y_2 = T_m$ or $T_m < Y_2 < T_{\Interval{q}}$. Now define $\gammaAND{p}{q} = \interval{ max(\alphaAND{p}{q}) , \frac{ max(\Interval{u_{Y_2}}) + max(\Interval{u_{T_{\Interval{q}}}}) }{2}} $. Let $Y'_2$ be the track just following $T_{\Interval{q}}$ in $\mathcal{T}$. We define $h = max(\Interval{u_{Y'_2}})$ if $Y'_2$ exists and $h = max(\Interval{u_{T_{\Interval{q}}}})+\epsilon$ otherwise. Now we define $\deltaAND{p}{q} = \interval{max(\gammaAND{p}{q}), \frac{ max(\Interval{u_{T_{\Interval{q}}}}) + h }{2}}$.


\item \sloppy \textbf{Creation of new tracks:} 
We create $T_1 = \set{ \interval{max(\betaAND{p}{q}), max(\betaAND{p}{q} + 1} } $ where $\Root{T_1}=\set{\alphaAND{p}{q},\betaAND{p}{q}}$, $T_2 =  \set{\interval{max(\gammaAND{p}{q}), max(\gammaAND{p}{q}) + 1} }$ where $\Root{T_2} = \set{\gammaAND{p}{q}}$ and $T_3 = \set{\interval{\deltaAND{p}{q},\deltaAND{p}{q}+1}}$ where $\Root{T_3} = \set{\deltaAND{p}{q}}$. 
\end{enumerate}

To complete the construction, define $\AND{p}{q} = \set{ \alphaAND{p}{q} , \betaAND{p}{q}, \gammaAND{p}{q}, \deltaAND{p}{q}} \cup \set{ \Interval{u_T}, \Interval{v_T} }_{T\in \mathcal{T}} \cup \set{T_j}_{1\leq j\leq 3}$. For each $T\in \mathcal{T}$, let $T=T\cup T_{new}$. Let $D = D \cup \AND{p}{q}$ and $\mathcal{T}=\mathcal{T} \cup \{T_1, T_2,T_3\}$. Observe that the intersection graph of $D$ does not contain $K_{1,5}$ as induced subgraph.

\subsection{Construction of variable gadgets} \label{sec:variable-rev}

We construct the variable gadgets sequentially and connect each of them to the previous one ($\VarGadget_1$ is connected to the start gadget $\StartGadget$). Assuming that we have placed $\StartGadget,\VarGadget_1,\ldots,\VarGadget_{i-1}$, we construct $\VarGadget_i$ as follows. For variable $x_i$, the gadget $\VarGadget_i$ consists of two implication gadgets. Let $D$ and $\mathcal{T}$ be the set of intervals and tracks created so far. First, we construct $\Implication{\trueVertex}{x_i}$. Observe that the sets $D$ and $\mathcal{T}$ have been updated after the last operation. There is an interval $\Interval{x_i}$ in $D$ and there is a track $T\in \mathcal{T}$ whose root is $\Interval{x_i}$. Now we construct $\Implication{x_i}{\overbar{x_i}}$. Observe that after constructing all the variable gadgets, for each literal $\ell$, there is an interval named $\Interval{\ell}$ in $D$. See Figure~\ref{fig:var-gadget} for an illustration of $\VarGadget_i$ created corresponding to the variable $x_i$. 

\begin{figure}[t]
    \centering
    \begin{tikzpicture}
        \node (imp1) at (0.5,0) [draw,thick,minimum width=2cm,minimum height=4cm] {$\Implication{\trueVertex}{\Interval{x_i}}$};
        
        \node (imp2) at (3.3,0) [draw,thick,minimum width=2cm,minimum height=4cm] {$\Implication{\Interval{x_i}}{\Interval{\overline{x_i}}}$};

        \node at (0.5,2.2) {$\Interval{x_i}$, $\Interval{r_{x_i}}$, $\Interval{s_{x_i}}$ };

        \node at (3.2,2.2) {$\Interval{\overline{x_i}}$, $\Interval{r_{\overline{x_i}}}$, $\Interval{s_{\overline{x_i}}}$};
        
    \end{tikzpicture}
    \caption{The variable gadget $\VarGadget_i$ created corresponding to the variable $x_i$. The first implication gadget contains three special intervals named $\Interval{x_i}, \Interval{r_{x_i}}, \Interval{s_{x_i}}$. The second implication gadget contains three special intervals named $\Interval{\overline{x_i}}, \Interval{r_{\overline{x_i}}}, \Interval{s_{\overline{x_i}}}$.}
    \label{fig:var-gadget}
\end{figure}
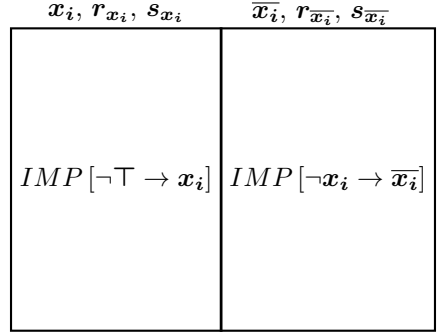

\subsection{Construction of clause gadgets}\label{sec:clause-gadget-rev}



We shall complete our construction of clause gadget $\ClauseGadget_i$ corresponding to the clause $C_i = (\ell_i^1, \ell_i^2, \ell_i^3)$. First, we create the covering gadget $\CovGadget{i}$ and update $D,\mathcal{T}$ as described in Section~\ref{sec:cover-gadget-rev}. Recall from the construction of $\CovGadget{i}$ that the intervals named $ \CGa{i}$, $\CGb{i}$, $\CGc{i}$ exist. Also recall from the construction of variable gadgets (described in Section~\ref{sec:variable-rev}) that the intervals $\Interval{\ell_i^1}$, $\Interval{\ell_i^2}$, $\Interval{\ell_i^3}$ and $\Interval{\overbar{\ell_i^1}}$, $\Interval{\overbar{\ell_i^2}}$, $\Interval{\overbar{\ell_i^3}}$ exist. Now we create, in this order, $\Implication{ \CGa{i}}{\CGap{i}}$, 
$\AND{\CGa{i}}{\ell_i^1}$, $\AND{\CGap{i}}{\Interval{\overbar{\ell_i^1}}}$,
$\Implication{ \CGb{i}}{\CGbp{i}}$, 
$\AND{\CGb{i}}{\ell_i^2}$, 
$\AND{\CGbp{i}}{\Interval{\overbar{\ell_i^2}}}$,
$\Implication{ \CGc{i}}{\CGcp{i}}$, 
$\AND{\CGc{i}}{\ell_i^3}$, 
$\AND{\CGcp{i}}{\Interval{\overbar{\ell_i^3}}}$ where $\CGap{i}$, $\CGbp{i}$ and $\CGcp{i}$ are three new intervals constructed in the corresponding implication gadgets. This completes the construction of $\ClauseGadget_i$. See Figure~\ref{fig:clause-gadget} for an illustration of the clause gadget $\ClauseGadget_i$.

\begin{figure}
    \centering
    \scalebox{0.65}{
    \begin{tikzpicture}
        
        \node (cov) at (0.5,0) [draw,thick,minimum width=2cm,minimum height=4cm] {$\CovGadget{i}$};
        \node (imp1) at (2.95,0) [draw,thick,minimum width=2cm,minimum height=4cm] {$\Implication{\Interval{a_i}}{\Interval{a'_i}}$};
        \node (and1) at (5.55,0) [draw,thick,minimum width=2cm,minimum height=4cm] {$\AND{\Interval{a_i}}{\Interval{\ell_i^1}}$};
        \node (and2) at (7.87,0) [draw,thick,minimum width=2cm,minimum height=4cm] {$\AND{\Interval{a'_i}}{\Interval{\overline{\ell_i^1}}}$};
        \node (imp2) at (10.49,0) [draw,thick,minimum width=2cm,minimum height=4cm] {$\Implication{\Interval{b_i}}{\Interval{b'_i}}$};
        \node at (13,0) [draw,thick,minimum width=2cm,minimum height=4cm] {$\AND{\Interval{b_i}}{\Interval{\ell_i^2}}$};
        \node at (15.3,0) [draw,thick,minimum width=2cm,minimum height=4cm] {$\AND{\Interval{b'_i}}{\Interval{\overline{\ell_i^2}}}$};
        \node (imp3) at (17.9,0) [draw,thick,minimum width=2cm,minimum height=4cm] {$\Implication{\Interval{c_i}}{\Interval{c'_i}}$};
        \node at (20.49,0) [draw,thick,minimum width=2cm,minimum height=4cm] {$\AND{\Interval{c_i}}{\Interval{\ell_i^3}}$};
        \node at (22.8,0) [draw,thick,minimum width=2cm,minimum height=4cm] {$\AND{\Interval{c'_i}}{\Interval{\overline{\ell_i^3}}}$};

        \node at (0.5,4) {\Large \begin{tabular}{c}
          $\CGd{i}$ \\
          $\CGa{i}$ \\
          $\CGb{i}$ \\
          $\CGc{i}$ \\
          $\CGCov{i}$\\
          $\CGf{i}$
     \end{tabular}};

        \node at (2.95,3) {\Large \begin{tabular}{c}
          $\Interval{\CGap{i}}$ \\
          $\Interval{r_{\CGap{i}}}$ \\
          $\Interval{s_{\CGap{i}}}$
     \end{tabular}};
     
              \node at (5.55,4) {\Large \begin{tabular}{c}
          $\InsertSigma{\Interval{\CGa{i}}}{\Interval{\ell^i_1}}$ \\
          $\alphaAND{\Interval{\CGa{i}}}{\Interval{\ell^i_1}}$ \\
          $\betaAND{\Interval{\CGa{i}}}{\Interval{\ell^i_1}}$ \\
          $\gammaAND{\Interval{\CGa{i}}}{\Interval{\ell^i_1}}$ \\
          $\deltaAND{\Interval{\CGa{i}}}{\Interval{\ell^i_1}}$
        \end{tabular}};

                \node at (7.87,-4.5) {\Large \begin{tabular}{c}
          $\InsertSigma{\Interval{\CGap{i}}}{\Interval{\overline{\ell^i_1}}}$ \\
          $\alphaAND{\Interval{\CGap{i}}}{\Interval{\overline{\ell^i_1}}}$ \\
          $\betaAND{\Interval{\CGap{i}}}{\Interval{\overline{\ell^i_1}}}$ \\
          $\gammaAND{\Interval{\CGap{i}}}{\Interval{\overline{\ell^i_1}}}$ \\
          $\deltaAND{\Interval{\CGap{i}}}{\Interval{\overline{\ell^i_1}}}$
        \end{tabular}};
        
        \node at (10.49,3) {\Large \begin{tabular}{c}
          $\Interval{\CGbp{i}}$ \\
          $\Interval{r_{\CGbp{i}}}$ \\
          $\Interval{s_{\CGbp{i}}}$
     \end{tabular}};

        \node at (13,4) {\Large \begin{tabular}{c}
          $\InsertSigma{\Interval{\CGb{i}}}{\Interval{{\ell^i_2}}}$ \\
          $\alphaAND{\Interval{\CGb{i}}}{\Interval{{\ell^i_2}}}$ \\
          $\betaAND{\Interval{\CGb{i}}}{\Interval{{\ell^i_2}}}$ \\
          $\gammaAND{\Interval{\CGb{i}}}{\Interval{{\ell^i_2}}}$ \\
          $\deltaAND{\Interval{\CGb{i}}}{\Interval{{\ell^i_2}}}$
        \end{tabular}};
        
        \node at (15.3,-4.5) {\Large \begin{tabular}{c}
          $\InsertSigma{\Interval{\CGbp{i}}}{\overline{\Interval{\ell^i_2}}}$ \\
          $\alphaAND{\Interval{\CGbp{i}}}{\Interval{\overline{\ell^i_2}}}$ \\
          $\betaAND{\Interval{\CGbp{i}}}{\Interval{\overline{\ell^i_2}}}$ \\
          $\gammaAND{\Interval{\CGbp{i}}}{\Interval{\overline{\ell^i_2}}}$ \\
          $\deltaAND{\Interval{\CGbp{i}}}{\Interval{\overline{\ell^i_2}}}$
        \end{tabular}};
        
        \node at (17.9,3) {\Large \begin{tabular}{c}
              $\Interval{\CGcp{i}}$ \\
              $\Interval{r_{\CGcp{i}}}$ \\
              $\Interval{s_{\CGcp{i}}}$
        \end{tabular}};
        
        \node at (20.49,4) {\Large \begin{tabular}{c}
          $\InsertSigma{\Interval{\CGc{i}}}{\Interval{\ell^i_3}}$ \\
          $\alphaAND{\Interval{\CGc{i}}}{\Interval{\ell^i_3}}$ \\
          $\betaAND{\Interval{\CGc{i}}}{\Interval{\ell^i_3}}$ \\
          $\gammaAND{\Interval{\CGc{i}}}{\Interval{\ell^i_3}}$ \\
          $\deltaAND{\Interval{\CGc{i}}}{\Interval{\ell^i_3}}$
        \end{tabular}};

                \node at (22.8,-4.5) {\Large \begin{tabular}{c}
          $\InsertSigma{\Interval{\CGcp{i}}}{\Interval{\overline{\ell^i_3}}}$ \\
          $\alphaAND{\Interval{\CGcp{i}}}{\Interval{\overline{\ell^i_3}}}$ \\
          $\betaAND{\Interval{\CGcp{i}}}{\Interval{\overline{\ell^i_3}}}$ \\
          $\gammaAND{\Interval{\CGcp{i}}}{\Interval{\overline{\ell^i_3}}}$ \\
          $\deltaAND{\Interval{\CGcp{i}}}{\Interval{\overline{\ell^i_3}}}$
        \end{tabular}};
        
    \end{tikzpicture}}
    \caption{Illustration of the clause gadget $\ClauseGadget_i$ with literals $\ell_i^1,\ell_i^2,\ell_i^3$. The names of the special intervals created at each gadget are highlighted above or below the box illustrating the respective gadgets.}
    \label{fig:clause-gadget}
\end{figure}

\subsection{Construction of end gadget}\label{sec:end-gadget-rev}

\begin{sidewaysfigure}[!h]
    \centering
    \scalebox{0.5}{
    \begin{tikzpicture}

     \node at (-18.7,0) [draw,thick,minimum width=2cm,minimum height=4cm] {$\StartGadget$};

      \node at (-18.7,2.2) { \Large $\startInterval, \trueInterval$ };

    \node at (-16.3,0) [draw,thick,minimum width=2cm,minimum height=4cm] {$\Implication{\trueVertex}{\Interval{x_1}}$};

     \node at (-13.4,0) [draw,thick,minimum width=2cm,minimum height=4cm] {$\Implication{\Interval{x_1}}{\Interval{\overline{x_1}}}$};

     \node at (-16.3,3) {\Large \begin{tabular}{c}
          $\Interval{x_1}$  \\
          $\Interval{r_{x_1}}$ \\
          $\Interval{s_{x_1}}$
     \end{tabular}};

     \node at (-13.4,3) {\Large \begin{tabular}{c}
          $\Interval{\overline{x_1}}$ \\
          $\Interval{r_{\overline{x_1}}}$ \\
          $\Interval{s_{\overline{x_1}}}$
     \end{tabular}};
        
    \node at (-10.5,0) [draw,thick,minimum width=2cm,minimum height=4cm] {$\Implication{\trueVertex}{\Interval{x_2}}$};

     \node at (-7.65,0) [draw,thick,minimum width=2cm,minimum height=4cm] {$\Implication{\Interval{x_2}}{\Interval{\overline{x_2}}}$};

     \node at (-10.5,3) {\Large \begin{tabular}{c}
          $\Interval{x_2}$  \\
          $\Interval{r_{x_2}}$ \\
          $\Interval{s_{x_2}}$
     \end{tabular}};

     \node at (-7.65,3) {\Large \begin{tabular}{c}
          $\Interval{\overline{x_2}}$ \\
          $\Interval{r_{\overline{x_2}}}$ \\
          $\Interval{s_{\overline{x_2}}}$
     \end{tabular}};

    \node at (-4.8,0) [draw,thick,minimum width=2cm,minimum height=4cm] {$\Implication{\trueVertex}{\Interval{x_3}}$};

        \node at (-1.95,0) [draw,thick,minimum width=2cm,minimum height=4cm] {$\Implication{\Interval{x_3}}{\Interval{\overline{x_3}}}$};

            \node at (-4.8,3) {\Large \begin{tabular}{c}
          $\Interval{x_3}$  \\
          $\Interval{r_{x_3}}$ \\
          $\Interval{s_{x_3}}$
     \end{tabular}};

     \node at (-1.95,3) {\Large \begin{tabular}{c}
          $\Interval{\overline{x_3}}$ \\
          $\Interval{r_{\overline{x_3}}}$ \\
          $\Interval{s_{\overline{x_3}}}$
     \end{tabular}};
        
        \node (cov) at (0.5,0) [draw,thick,minimum width=2cm,minimum height=4cm] {$\CovGadget{1}$};

        \node (imp1) at (2.95,0) [draw,thick,minimum width=2cm,minimum height=4cm] {$\Implication{\Interval{a_1}}{\Interval{a'_1}}$};

        \node (and1) at (5.55,0) [draw,thick,minimum width=2cm,minimum height=4cm] {$\AND{\Interval{a_1}}{\Interval{x_1}}$};

        \node (and2) at (7.87,0) [draw,thick,minimum width=2cm,minimum height=4cm] {$\AND{\Interval{a'_1}}{\Interval{\overline{x^1}}}$};

        \node (imp2) at (10.49,0) [draw,thick,minimum width=2cm,minimum height=4cm] {$\Implication{\Interval{b_1}}{\Interval{b'_1}}$};

        \node at (13,0) [draw,thick,minimum width=2cm,minimum height=4cm] {$\AND{\Interval{b_1}}{\Interval{\overline{x_2}}}$};

        \node at (15.3,0) [draw,thick,minimum width=2cm,minimum height=4cm] {$\AND{\Interval{b'_1}}{\Interval{x_2}}$};

        \node (imp3) at (17.9,0) [draw,thick,minimum width=2cm,minimum height=4cm] {$\Implication{\Interval{c_1}}{\Interval{c'_1}}$};

        \node at (20.49,0) [draw,thick,minimum width=2cm,minimum height=4cm] {$\AND{\Interval{c_1}}{\Interval{x_3}}$};

        \node at (0.5,4) {\Large \begin{tabular}{c}
          $\CGd{1}$ \\
          $\CGa{1}$ \\
          $\CGb{1}$ \\
          $\CGc{1}$ \\
          $\CGCov{1}$\\
          $\CGf{1}$
     \end{tabular}};

        \node at (2.95,3) {\Large \begin{tabular}{c}
          $\Interval{\CGap{1}}$ \\
          $\Interval{r_{\CGap{1}}}$ \\
          $\Interval{s_{\CGap{1}}}$
     \end{tabular}};
     
              \node at (5.55,4) {\Large \begin{tabular}{c}
          $\InsertSigma{\Interval{\CGa{1}}}{\Interval{x_1}}$ \\
          $\alphaAND{\Interval{\CGa{1}}}{\Interval{x_1}}$ \\
          $\betaAND{\Interval{\CGa{1}}}{\Interval{x_1}}$ \\
          $\gammaAND{\Interval{\CGa{1}}}{\Interval{x_1}}$ \\
          $\deltaAND{\Interval{\CGa{1}}}{\Interval{x_1}}$
        \end{tabular}};

                \node at (7.87,-4) {\Large \begin{tabular}{c}
          $\InsertSigma{\Interval{\CGap{1}}}{\Interval{\overline{x_1}}}$ \\
          $\alphaAND{\Interval{\CGap{1}}}{\Interval{\overline{x_1}}}$ \\
          $\betaAND{\Interval{\CGap{1}}}{\Interval{\overline{x_1}}}$ \\
          $\gammaAND{\Interval{\CGap{1}}}{\Interval{\overline{x_1}}}$ \\
          $\deltaAND{\Interval{\CGap{1}}}{\Interval{\overline{x_1}}}$
        \end{tabular}};
        
        \node at (10.49,3) {\Large \begin{tabular}{c}
          $\Interval{\CGbp{1}}$ \\
          $\Interval{r_{\CGbp{1}}}$ \\
          $\Interval{s_{\CGbp{1}}}$
     \end{tabular}};

        \node at (13,4) {\Large \begin{tabular}{c}
          $\InsertSigma{\Interval{\CGb{1}}}{\Interval{\overline{x_2}}}$ \\
          $\alphaAND{\Interval{\CGb{1}}}{\Interval{\overline{x_2}}}$ \\
          $\betaAND{\Interval{\CGb{1}}}{\Interval{\overline{x_2}}}$ \\
          $\gammaAND{\Interval{\CGb{1}}}{\Interval{\overline{x_2}}}$ \\
          $\deltaAND{\Interval{\CGb{1}}}{\Interval{\overline{x_2}}}$
        \end{tabular}};
        
        \node at (15.3,-4) {\Large \begin{tabular}{c}
          $\InsertSigma{\Interval{\CGbp{1}}}{\Interval{x_2}}$ \\
          $\alphaAND{\Interval{\CGbp{1}}}{\Interval{x_2}}$ \\
          $\betaAND{\Interval{\CGbp{1}}}{\Interval{x_2}}$ \\
          $\gammaAND{\Interval{\CGbp{1}}}{\Interval{x_2}}$ \\
          $\deltaAND{\Interval{\CGbp{1}}}{\Interval{x_2}}$
        \end{tabular}};
        
        \node at (17.9,3) {\Large \begin{tabular}{c}
              $\Interval{\CGcp{1}}$ \\
              $\Interval{r_{\CGcp{1}}}$ \\
              $\Interval{s_{\CGcp{1}}}$
        \end{tabular}};
        
        \node at (20.49,4) {\Large \begin{tabular}{c}
          $\InsertSigma{\Interval{\CGc{1}}}{\Interval{x_3}}$ \\
          $\alphaAND{\Interval{\CGc{1}}}{\Interval{x_3}}$ \\
          $\betaAND{\Interval{\CGc{1}}}{\Interval{x_3}}$ \\
          $\gammaAND{\Interval{\CGc{1}}}{\Interval{x_3}}$ \\
          $\deltaAND{\Interval{\CGc{1}}}{\Interval{x_3}}$
        \end{tabular}};

                \node at (22.8,-4) {\Large \begin{tabular}{c}
          $\InsertSigma{\Interval{\CGcp{1}}}{\Interval{\overline{x_3}}}$ \\
          $\alphaAND{\Interval{\CGcp{1}}}{\Interval{\overline{x_3}}}$ \\
          $\betaAND{\Interval{\CGcp{1}}}{\Interval{\overline{x_3}}}$ \\
          $\gammaAND{\Interval{\CGcp{1}}}{\Interval{\overline{x_3}}}$ \\
          $\deltaAND{\Interval{\CGcp{1}}}{\Interval{\overline{x_3}}}$
        \end{tabular}};
        
        \node at (22.8,0) [draw,thick,minimum width=2cm,minimum height=4cm] {$\AND{\Interval{c'_1}}{\Interval{\overline{x_3}}}$};

        \node at (25,0) [draw,thick,minimum width=2cm,minimum height=4cm] {$\EndGadget$};

        \node at (25,3) {\Large $\set{\Tail{T} \colon {T\in \mathcal{T}} }$};
    \end{tikzpicture}}
    \caption{Illustration of the gadgets for the 3-SAT instance $(x_1 \lor \overline{x_2} \lor x_3)$. The names of the special intervals created at each gadget are highlighted above or below the box illustrating the respective gadgets.}
    \label{fig:complete-gadget}
\end{sidewaysfigure}

For each $T\in \mathcal{T}$, we introduce two new intervals $\Interval{u}_T=[max(T), max(T) + 1], \Tail{T} = [max(\Interval{u}_T),max(\Interval{u}_T)]$ and define $T=T\cup \{ \Interval{u}_T \}, D=D\cup \{ \Interval{u}_T, \Interval{e}_T \}_{T\in \mathcal{T}}$. For each $T\in \mathcal{T}$, let $\Tail{T}$ be the \emph{tail} of $T$. The end gadget $\EndGadget$ consists of all the new intervals created above.  See Figure~\ref{fig:complete-gadget} for an illustration of all the gadgets created corresponding to the 3-SAT instance $F=(x_1 \lor \overline{x_2} \lor x_3)$. Observe that the intersection graph of $D$ does not contain $K_{1,5}$ as induced subgraph.

\subsection{Proofs}

\begin{figure}[!h]
    \centering
    \scalebox{0.7}{
    \begin{tikzpicture}
        
    \node (thm3) at (-0.5,0) [draw,thick,minimum width=2cm,minimum height=1cm] {Theorem~\ref{thm:interval-hard}};

    \node [below=of thm3, draw,thick,minimum width=2cm,minimum height=1cm] (lem37) {Lemma~\ref{lem:uper-bound-interval}};

    \node [left=of lem37, xshift=-4cm, draw,thick,minimum width=2cm,minimum height=1cm] (lem44) {Lemma~\ref{lem:sat-opt}};

    \node [right=of lem37, xshift=4cm, draw,thick,minimum width=2cm,minimum height=1cm] (lem51) {Lemma~\ref{lem:opt-sat}};

    \node [below=of lem51, xshift=-2cm, draw,thick,minimum width=2cm,minimum height=1cm] (lem50) {Lemma~\ref{lem:clause-select-1}};

    \node [below=of lem51, xshift=2cm, draw,thick,minimum width=2cm,minimum height=1cm] (lem49) {Lemma~\ref{lem:clause-select}};

    \node [below=of lem50, draw,thick,minimum width=2cm,minimum height=1cm] (lem47) {Lemma~\ref{lem:implication-select}};
    
    \node [below=of lem49, draw,thick,minimum width=2cm,minimum height=1cm] (lem48) {Lemma~\ref{lem:AND-select}};

    \node [below=of lem51, yshift=-4cm, draw,thick,minimum width=2cm,minimum height=1cm] (lem45) {Lemma~\ref{lem:good-geodetic-set}};

    \node [below=of lem45, draw,thick,minimum width=2cm,minimum height=1cm] (lem39) {Lemma~\ref{lem:semi-good}};

        \node [below=of lem44, xshift =-5cm, yshift=-6cm, draw,thick,minimum width=2cm,minimum height=1cm] (lem40) {Lemma~\ref{lem:track-covered-rev}};

        \node [below=of lem44, xshift =-2cm, yshift=-6cm, draw,thick,minimum width=2cm,minimum height=1cm] (lem41) {Lemma~\ref{lem:imply-cover-rev}};

        \node [below=of lem44, xshift =2cm, yshift=-6cm, draw,thick,minimum width=2cm,minimum height=1cm] (lem42) {Lemma~\ref{lem:cover-gadget-1-rev}};

        \node [below=of lem44, xshift =5cm, yshift=-6cm, draw,thick,minimum width=2cm,minimum height=1cm] (lem43) {Lemma~\ref{lem:AND-cover-1-rev}};

        \node [below=of lem44, yshift=-8cm, draw,thick,minimum width=2cm,minimum height=1cm] (lem38) {Lemma~\ref{lem:shortest}};

        \node [below=of lem37, yshift=-10cm, draw,thick,minimum width=2cm,minimum height=1cm] (construct) {Construction of $D$};      

        \node [below=of lem44, rounded corners, fill=gray, opacity=0.1, yshift=-5.5cm, draw,thick,minimum width=13cm,minimum height=2cm] (repre) {};
        
        \draw[->,thick,densely dotted] (construct.north) -- (lem37.south);
        \draw[->,thick,densely dotted] (construct.west) -| (lem40.south);
        \draw[->,thick, dotted] (lem38.south) ++ (0,-1.5) -- (lem38.south);
        \draw[->,thick, dotted] (lem39.west) ++ (-3.5,0) -- (lem39.west);
        \draw[->,thick,densely dotted] (lem41.south) ++ (0,-3.5) -- (lem41.south);
        \draw[->,thick,densely dotted] (lem42.south) ++ (0,-3.5) -- (lem42.south);
        \draw[->,thick,densely dotted] (lem43.south) ++ (0,-3.5) -- (lem43.south);         
        \draw[->,thick,densely dotted] (construct.east) -- ++ (1,0) |- (lem45.west);

        \draw[->,thick,densely dashed] (lem38.west) -- ++ (-3.5,0) -- ++ (0,1.5); \draw[->,thick,densely dashed] (lem38.east) -- ++ (3.5,0) -- ++ (0,1.5);\draw[->,thick,densely dashed] (lem38.east) -- ++ (1.5,0) -- ++ (0,1.5);\draw[->,thick,densely dashed] (lem38.west) -- ++ (-1.5,0) -- ++ (0,1.5);
        
        \draw[->,thick,densely dotted] (repre.north) -- (lem44.south);
        \draw[->,thick,densely dashed] (repre.north) ++ (0.75,0)  -- ++(0,1.25) -- ++ (12.25,0);

        \draw[->,thick,densely dotted] (lem39.north) -- (lem45.south);
        \draw[->,thick,densely dotted] (lem45.north) -- ++ (0,0.5) -| (lem47.south); 
        \draw[->,thick,densely dotted] (lem45.north) -- ++ (0,0.5) -| (lem48.south);
        \draw[->,thick,densely dotted] (lem47.north) -- (lem50.south);
        \draw[->,thick,densely dotted] (lem48.north) -- (lem49.south);
        \draw[->,thick,densely dotted] (lem48.west) -- ++ (-0.5,0) -- ++ (0,1.75) -- ++ (-1.5,0);
        \draw[->,thick,densely dotted] (lem49.west) -- (lem50.east);
        \draw[->,thick,densely dotted] (lem50.north) |- (lem51.west);
        \draw[->,thick,densely dotted] (lem49.north) |- (lem51.east);

        \draw[->,thick,densely dotted] (lem37.north) -- (thm3.south);
        
        \draw[->,thick,densely dotted] (lem44.north) |- (thm3.west);

        \draw[->,thick,densely dotted] (lem51.north) |- (thm3.east);

    \end{tikzpicture}}
    \caption{Roadmap for proof of Theorem~\ref{thm:interval-hard}.}
    \label{fig:roadmap-interval}
\end{figure}
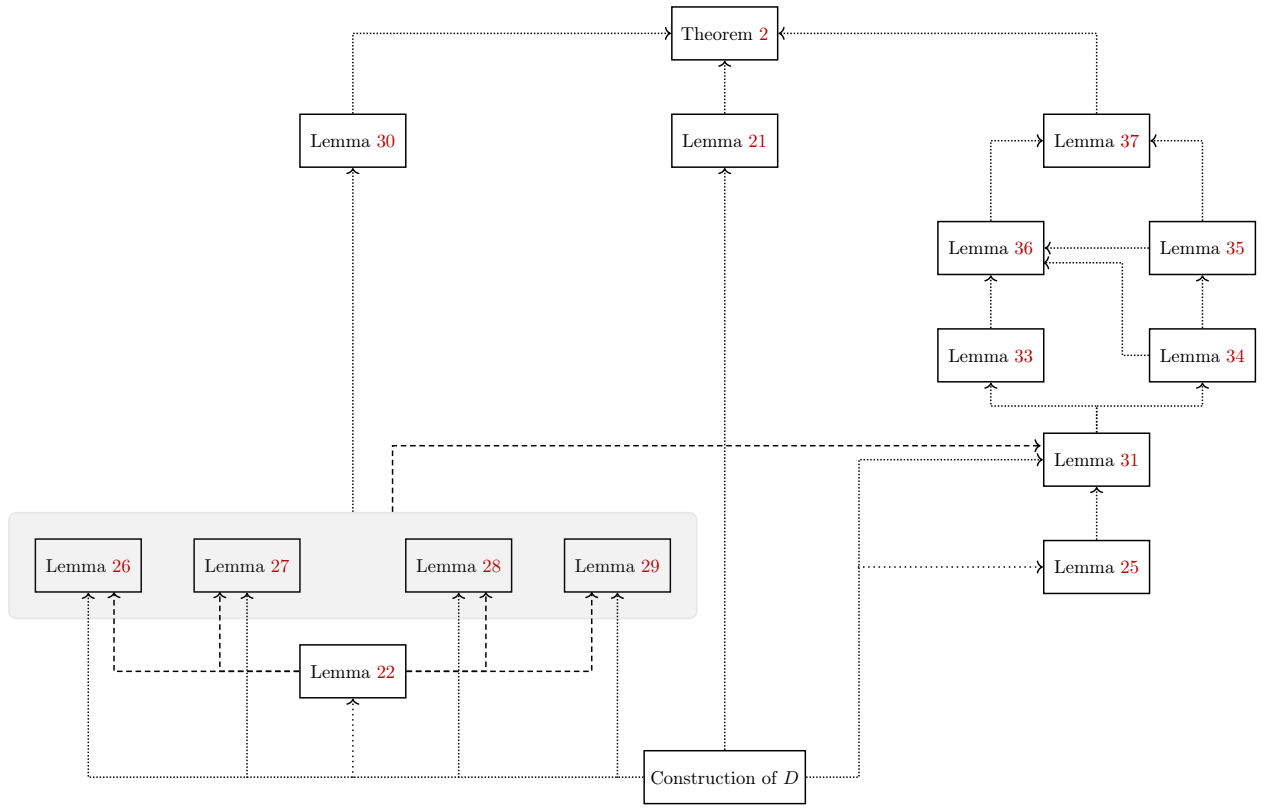

In this section, we shall show that the 3-SAT instance is satisfiable if and only if the constructed graph has a geodetic set with a certain cardinality. For the entirety of this section, $F$ shall denote a 3-SAT instance with variables $x_1,x_2,\ldots, x_n$ and clauses $C_1,C_2,\ldots,C_m$. The set $D$ will denote the set of constructed intervals, $\mathcal{T}$ will denote the tracks. For each $i\in [n]$, $\VarGadget_i$ shall denote the constructed variable gadget and for each $j\in [m]$, $\ClauseGadget_j$ shall denote the set of constructed gadget. The reader may use Figure~\ref{fig:roadmap-interval} to navigate the proof of Theorem~\ref{thm:interval-hard}. 

We begin by showing that the number of vertices in the constructed graph is polynomial in the number of variables and clauses, which implies that the construction procedure takes polynomial time.

\begin{lemma}\label{lem:uper-bound-interval}
There are $\ntracks$ tracks in $\mathcal{T}$ and $\nbSimplicial$ point intervals in $D$. The total number of intervals in $D$ is $O((n+m)^2)$.
\end{lemma}

\begin{proof}
    Recall that the construction of the start gadget consists of creating two point intervals and two tracks. The construction of one implication gadget consists of creating one point interval and two new tracks. The construction of one AND gadget consists of creating two point intervals (one of which is due to the creation of the insert gadget) and four new tracks. The construction of one covering gadget consists of creating two point intervals and five new tracks. Recall that each variable gadget consists of two implication gadgets, and one clause gadget contains three implication gadgets. Thus, the total number of implication gadgets created is $2n+3m$. Each clause gadget contains six AND gadgets. Hence, the total number of AND gadgets created is $6m$. The total number of covering gadgets created is $m$. So, the total number of tracks created is $\ntracks$. The end gadget contains one distinct point interval for each track. Hence, the total number of point intervals is $\nbSimplicial$. Thus, the total number of intervals in $D$ is $O((n+m)^2)$.  
\end{proof}

In the next few lemmas, we prove that if an induced path of an interval graph has a certain structure, then it is also a shortest path between its end-vertices.

\begin{lemma}\label{lem:shortest}
Let $\Interval{u}$ and $\Interval{v}$ be two intervals of $D$ such that $\min(\Interval{u}) < \min(\Interval{v})$. Let $P=\Interval{u_0} \Interval{u_1} \dots \Interval{u_k} \Interval{v}$ be an induced path such that $\Interval{u_0}=\Interval{u}$ and $\Interval{u_{i+1}}$ is the rightmost neighbour of $\Interval{u_i}$ for $0\leq i\leq k-1$, and $\Interval{u_{k-1}} \notin N(\Interval{v})$, while $\Interval{u_{k}} \in N(\Interval{v})$. Then, $P$ is a shortest path.
\end{lemma}

\begin{proof} 
Let $P = \Interval{u_0} \Interval{u_1} \dots \Interval{u_k} \Interval{v}$ and $P'$ be a shortest path from $\Interval{u}$ to $\Interval{v}$ which starts by the longest common subpath with $P$. Hence $P'= \Interval{u_0} \Interval{u_1} \dots \Interval{u_i} \Interval{z} \dots \Interval{v}$ with $\Interval{z} \notin P$. If $i=k$ then $P$ and $P'$ have the same length. Otherwise, replace $\Interval{z}$ in $P'$ by $\Interval{u_{i+1}}$ to obtain a path $P''$. It is indeed a path as $\max(\Interval{z}) < \max(\Interval{u_{i+1}})$. Moreover $P'$ and $P''$ have the same length, a contradiction with the definition of $P'$.
\end{proof}

\begin{definition}
    Let $P=\Interval{u_1}\ldots \Interval{u_k}$ be a path such that $\Interval{u_{i+1}}$ is the rightmost neighbour of $\Interval{u_i}$ for all $i\in 1\leq i\leq k-2$ and $\Interval{u_{k-2}}$ is not adjacent to $\Interval{u_{k}}$. The path $P$ is a shortest path between $\Interval{u_1}$ and $\Interval{u_k}$ by Lemma~\ref{lem:shortest}. We say that such a path is a ``\emph{good shortest path}''.
\end{definition}

\begin{definition}
    Let $\Interval{u},\Interval{v},\Interval{w}$ be three intervals of $D$ such that $\max(\Interval{u}) < \min(\Interval{w})$ and $\Interval{v} \in I(\Interval{u},\Interval{w})$. A shortest path $P=\Interval{u_0} \Interval{u_1} \ldots \Interval{u_k} \Interval{v} \Interval{u_{k+1}} \ldots \Interval{w}$ between $\Interval{u_0}=\Interval{u}$ and $\Interval{w}$ is ``\emph{semi-good}'' if for each $0\leq i\leq k-1$, $\Interval{u_{i+1}}$ is the rightmost interval of $\Interval{u_i}$. 
\end{definition}

\begin{lemma}\label{lem:semi-good}
    Let $\Interval{u},\Interval{v},\Interval{w}$ be two intervals of $D$ such that $\max(\Interval{u}) < \min(\Interval{w})$ and $\Interval{v} \in I(\Interval{u},\Interval{w})$. Then, there exists a semi-good shortest path between $\Interval{u}$ and $\Interval{w}$ containing $\Interval{v}$.
\end{lemma}
\begin{proof}
Let $P=\Interval{u_0} \ldots \Interval{u_k} \Interval{v} \Interval{u_{k+1}} \ldots \Interval{w}$ be a shortest path between $\Interval{u_0} = \Interval{u}$ and $\Interval{w}$ that contains $\Interval{v}$.  Now consider the path $Q = \Interval{u'_0} \ldots \Interval{u'_k}$ such that $\Interval{u'_0}=\Interval{u}$ and for $0\leq i \leq k-1$, $\Interval{u'_{i+1}}$ is the rightmost neighbour of $\Interval{u'_i}$. Observe that $\Interval{u'_k}$ intersects $\Interval{v}$ and $\Interval{u'_k}$ does not  intersect $\Interval{u_{k+1}}$ (else, it contradicts the fact that $P$ is a shortest path). Therefore, $\Interval{u'_0} \ldots \Interval{u'_k} \Interval{v} \Interval{u_{k+1}} \ldots \Interval{w}$ is a semi-good shortest path containing $v$.
\end{proof}

In Subsection~\ref{sec:sat-opt}, we shall show that if the 3-SAT instance is satisfiable, then the constructed interval graph has a geodetic set of certain cardinality. In Section~\ref{sec:opt-sat}, we shall show that if the constructed interval graph has a geodetic set of certain cardinality, then the 3-SAT instance is satisfiable.

\subsubsection{Satisfiability implies optimality} \label{sec:sat-opt}

Recall that each track $T$ is associated with a set of intervals called its roots, denoted by $\Root{T}$. Also from our construction, it is clear that an interval is the root of at most one track. Hence we define the following notation: for an interval $\Interval{z}$, let $T(\Interval{z})$ denote the track $T$ such that $\Interval{z} \in \Root{T(\Interval{z})}$.

\begin{lemma}\label{lem:track-covered-rev}
If $T$ is a track in $\mathcal{T}$, then $T \subseteq I(\startInterval, \Tail{T})$.
\end{lemma}

\begin{proof} 
Let $T = \set{\Interval{a_1},\dots,\Interval{a_j}}$ be the set of intervals of $T$ such that $max(\Interval{a_i})=min(\Interval{a_{i+1}})$ for $i\in [j-1]$. Let $T(\startInterval) = \set{\Interval{b_1},\dots,\Interval{b_k}}$ such that $max(\Interval{b_i})=min(\Interval{b_{i+1}})$ for $i\in [k-1]$. By Lemma~\ref{lem:shortest}, the path induced by the set $\{\startInterval, \Interval{b_1},\dots,\Interval{b_k},\Interval{z},\Tail{T}\}$ is a good shortest path from $\startInterval$ to $\Tail{T}$ where $z$ is the rightmost neighbour of $\Interval{b_k}$. Hence the distance between $\startInterval$ and $\Tail{T}$ is $k+3$. 

Observe that both $\Interval{b_k}$ and $\Interval{a_j}$ belong to the end gadget and therefore they are adjacent with $\max(\Interval{b_k}) < \max(\Interval{a_j})$. From the construction it implies that $\min(\Interval{b_{k}}) < \min(\Interval{a_{j}})$. Since $\max(\Interval{b_{k-1}}) = \min(\Interval{b_{k}})$, it follows that $\Interval{b_{k-1}}$ is not adjacent with $\Interval{a_j}$. Hence the path induced by the set $\{\startInterval, \Interval{b_1},\dots,\Interval{b_k},\Interval{a_j},\Tail{T}\}$ has length $k+3$ and therefore is a shortest path. 

Now, using an inductive argument on $i$ (from $i=j$ to $i=1$), it follows that $\Interval{a_i}$ and $\Interval{b_{k-j+i}}$ are neighbours. Hence, the path induced by the set $\{\startInterval, \Interval{b_1},\dots,\Interval{b_{k-j+1}},\Interval{a_1},\dots,\Interval{a_j},\Tail{T}\}$ induces a shortest path $P'$ from $\startInterval$ to $\Tail{T}$ such that $V(T)\subseteq V(P')$. Hence, $T \subseteq I(\startInterval, \Tail{T})$.  
\end{proof}

\begin{lemma}\label{lem:imply-cover-rev}
Consider an implication gadget $\Implication{p}{q}$. We have that $\Interval{q} \in I(\Interval{p}, \Interval{s_q})$ and $\Interval{r_q} \in I(\startInterval, \Interval{s_q})$.
\end{lemma}

\begin{proof}
Consider the track $T(\Interval{p})$. Observe from the construction of $\Implication{\Interval{p}}{\Interval{q}}$ that there exists a subset  $ \set{\Interval{v_1}, \Interval{v_2}, \ldots, \Interval{v_{k+2}}} \subseteq T(\Interval{p})$ such that $\Interval{p}$ and $\Interval{v_1}$ are adjacent, $\Interval{v_{i+1}}$ is the rightmost neighbour of $\Interval{v_{i}}$ for $1 \leq i \leq k+1$ and $\Interval{v_k}$ is the leftmost neighbour of $\Interval{q}$. By Lemma~\ref{lem:shortest}, $\set{\Interval{p},\Interval{v_1},\Interval{v_2}, \ldots, \Interval{v_k}, \Interval{v_{k+1}}, \Interval{v_{k+2}}, \Interval{s_q}}$ induces a good shortest path between $\Interval{p}$ and $\Interval{s_q}$. Hence, the path induced by $\set{\Interval{p},\Interval{v_1},\Interval{v_2}, \ldots, \Interval{v_k}, \Interval{q}, \Interval{r_q}, \Interval{s_q}}$ is also a shortest path between $\Interval{p}$ and $\Interval{s_q}$.

Now consider the track $T(\startInterval)$. Observe from the construction of 
$\Implication{\Interval{p}}{\Interval{q}}$ that there exists a subset  $ \set{\Interval{w_1}, \Interval{w_2}, \ldots, \Interval{w_{t}}} \subseteq T(\startInterval)$ such that $\Interval{w_1}$ is the rightmost neighbour of $\startInterval$, $\Interval{w_{i+1}}$ is the leftmost neighbour of $\Interval{v_{i}}$ for $1 \leq i \leq t-1$, $\Interval{s_q}\cap \Interval{w}_{t-1}=\emptyset$, and $\Interval{s_q}\subset \Interval{w}_t$. By Lemma~\ref{lem:shortest}, the set $\set{\startInterval,\Interval{w_1},\Interval{w_2}, \ldots, \Interval{w_{t-1}}, \Interval{w_t}, \Interval{s_q}}$ is a shortest path between $\startInterval$ and $\Interval{s_q}$ of length $t+1$. Also observe from the construction that $\Interval{w}_{t-1}$ is adjacent to $\Interval{r}_q$. Hence, the path induced by the set $\set{\startInterval,\Interval{w_1},\Interval{w_2}, \ldots, \Interval{w_{t-1}}, \Interval{r_q}, \Interval{s_q}}$ is a shortest path between $\startInterval$ and $\Interval{s_q}$.
\end{proof}

\begin{lemma}\label{lem:cover-gadget-1-rev}
Consider the cover gadget $\CovGadget{i}$ and let $\Interval{z} \in \set{\CGa{i},\CGb{i},\CGc{i}}$. Then, $\Interval{z}\in I(\CGd{i}, \Tail{T(\Interval{z})})$ and $\CGCov{i}\in I(\Interval{z},\CGf{i})$.
\end{lemma}

\begin{proof}
Let $T$ denote the track $T(\Interval{z})$ and let $T=\set{\Interval{u_1},\Interval{u_2},\ldots,\Interval{u_{k}}}$ such that $\Interval{u_{i+1}}$ is the rightmost neighbour of $\Interval{u_{i}}$ for all $1\leq j\leq k-1$. Observe from the construction that the path induced by the set $\set{\CGd{i},\Interval{z},\Interval{u_1},\Interval{u_2},\ldots,\Interval{u_{k}},\Tail{T}}$ is a shortest path between $\CGd{i}$ and $\Tail{T}$. This proves the first part of the proposition. For the second part, consider the path $P$ induced by the set $\set{\Interval{z},\Interval{u_1},\Interval{u_2},\CGCov{i},\CGf{i}}$. Clearly, $P$ is a shortest path between $\Interval{z}$ and $\CGf{i}$.
\end{proof}

\begin{lemma}\label{lem:AND-cover-1-rev}
Consider an AND gadget $\AND{p}{q}$ and the insert gadget $\INS{p}{q}$. Let $T_1 = T(\alphaAND{p}{q})$, $T_2 = T(\gammaAND{p}{q})$ and $T_3=T(\deltaAND{p}{q})$. Then

\begin{enumerate}
\renewcommand{\theenumi}{(\alph{enumi})}
\renewcommand{\labelenumi}{\theenumi}

\item\label{it:and-cover-1} $\gammaAND{p}{q} \in I(\InsertSigma{p}{q}, \Tail{T_2})$,

\item\label{it:and-cover-2} $\alphaAND{p}{q} \in I(\Interval{p}, \Tail{T_1})$, $\deltaAND{p}{q} \in I(\Interval{q}, \Tail{T_3})$, and

\item\label{it:and-cover-3} $\alphaAND{p}{q} \in I(\betaAND{p}{q} , \gammaAND{p}{q})$ and $\deltaAND{p}{q} \in I(\gammaAND{p}{q},\Tail{T_3})$.
\end{enumerate}
\end{lemma}

\begin{proof}
First, we prove (a). Consider the track $T(\InsertSigma{p}{q})$. Observe from the construction that $T(\InsertSigma{p}{q})$ can be partitioned into two sets  $ P_a=\set{\Interval{u_1},\Interval{u_2},\ldots,\Interval{u_j}}$ and $P'_a = \set{\Interval{u_{j+1}},\Interval{u_{j+2}},\ldots,\Interval{u_k}}$ where $\Interval{u_1}$ is the rightmost neighbour of $\InsertSigma{p}{q}$, $\Interval{u_{i+1}}$ is the rightmost neighbour of $\Interval{u_i}$ for all $1\leq i\leq k-1$, and $\Interval{u_j}$ is the interval of $T(\InsertSigma{p}{q})$ with minimal index which intersects $\Interval{\gammaAND{p}{q}}$. Let $\Interval{z_a}$ be the rightmost neighbour of $\Interval{u_k}$. 
Since $\Interval{u_k}\in \EndGadget$, either $\Interval{z_a} = \Interval{u_{T_2}}$ or $ \max(\Interval{z_a})\geq \max(\Tail{T_2})$. Therefore $\Tail{T_2}$ is adjacent to $\Interval{z_a}$. Hence, by Lemma~\ref{lem:shortest}, the path induced by $Q_a=\set{\InsertSigma{p}{q}} \cup P_a \cup P'_a \cup \set{\Interval{z_a},\Tail{T_2}}$ is a good shortest path between $\InsertSigma{p}{q}$ and $\Tail{T_2}$. Observe from the construction that $|T_2| = |P'|$. Therefore, the set $Q'_a=\set{\InsertSigma{p}{q}} \cup P_a \cup \set{\gammaAND{p}{q}} \cup T_2 \cup \{ \Tail{T_2}\}$ has the same cardinality as $Q_a$. Moreover,  $Q'_a$ induces a path. This implies $Q'_a$ induces a shortest path between $\InsertSigma{p}{q}$ and $\Tail{T_2}$.

Second, we prove (b). Consider the track $T(\Interval{p})$. Observe from the construction that $T(\Interval{p})$ can be partitioned into two sets $ P_b=\set{\Interval{v_1},\Interval{v_2},\ldots,\Interval{v_{j'}}}$ and $P'_b = \set{\Interval{v_{j'+1}},\Interval{v_{j'+2}},\ldots,\Interval{v_{k'}}}$ where $\Interval{v_1}$ is the rightmost neighbour of $\Interval{p}$, $\Interval{v_{i+1}}$ is the rightmost neighbour of $\Interval{v_i}$ for all $1\leq i\leq k'-1$, and $\Interval{v_{j'}}$ is the interval of $T(\Interval{p})$ with minimal index which intersects $\Interval{\alphaAND{p}{q}}$. Let $\Interval{z_b}$ be the rightmost neighbour of $\Interval{v_{k'}}$. 
Since $\Interval{v_{}k'}\in \EndGadget$, either $\Interval{z_b} = \Interval{u_{T_1}}$ or $ \max(\Interval{z_b})\geq \max(\Tail{T_1})$. Therefore $\Tail{T_1}$ is adjacent to $\Interval{z_b}$. 
Hence, by Lemma~\ref{lem:shortest}, the path induced by $Q_b=\set{\Interval{p}} \cup P_b \cup P'_b \cup \set{\Interval{z_b},\Tail{T_1}}$ is a good shortest path between $\Interval{p}$ and $\Tail{T_1}$. Observe from the construction that $|T_1| = |P'_b|$. Therefore, the set $Q'_b=\set{\Interval{p}} \cup P_b \cup \set{\alphaAND{p}{q}} \cup T_1 \cup \{ \Tail{T_1}\}$ has the same cardinality as $Q_b$. Moreover,  $Q'_b$ induces a path. This implies $Q'_b$ induces a shortest path between $\Interval{p}{q}$ and $\Tail{T_1}$. Using similar arguments, we can show that $\deltaAND{p}{q} \in I(\Interval{q}, \Tail{T_3})$.

Finally, we prove $(c)$. Observe from the construction that the distance between $\betaAND{p}{q}$ and $\deltaAND{p}{q}$ is two and therefore $\alphaAND{p}{q} \in I(\betaAND{p}{q} , \gammaAND{p}{q})$. Let $T_3 = \set{\Interval{w_1},\Interval{w_2},\ldots,\Interval{w_{k''}}}$ such that $\Interval{w_1}$ is the rightmost neighbour of $\deltaAND{p}{q}$, $\Interval{w_{i+1}}$ is the rightmost neighbour of $\Interval{w_i}$ for all $1\leq i\leq k''-1$. Hence, by Lemma~\ref{lem:shortest}, the path induced by $\set{\deltaAND{p}{q}, \Interval{w_1},\Interval{w_2},\ldots,\Interval{w_{k''}}, \Tail{T_3}}$ is a good shortest path. Now observe from the construction that the distance between $\Interval{w_1}$ and $\gammaAND{p}{q}$ is two and no neighbour of $\gammaAND{p}{q}$ is adjacent to $\Interval{w_i}$ for $2\leq i \leq k''$. This implies that the path induced by the set $\set{ \gammaAND{p}{q}, \deltaAND{p}{q}, \Interval{w_1},\Interval{w_2},\ldots,\Interval{w_{k''}}, \Tail{T_3}}$ is a shortest path between $\gammaAND{p}{q}$ and $\Tail{T_3}$. This completes the proof.
\end{proof}

\begin{lemma}\label{lem:sat-opt}
    If the 3-SAT instance $F$ is satisfiable then $\mathcal{I}(D)$ has a geodetic set of cardinality $\bound$.
\end{lemma}
\begin{proof}
    We shall show that if $F$ is satisfiable, then $D$ has a geodetic set of cardinality $\bound = n_{point} + \boundWithoutSimplicial$. Let $S_p$ denote the set of point intervals in $D$. Note that the point intervals are the only simplicial vertices in $D$. Hence, they must all belong to any geodetic set of $D$. Let $\phi \colon \{ x_1,x_2,\ldots,x_n \} \rightarrow \{1,0\}$ be a satisfying assignment of $F$. Now, define the following sets. Let $S_1=\{ \Interval{x_i} \colon \phi(x_i)=1 \} \cup \{ \Interval{\overbar{x_i}} \colon \phi(x_i)=0 \}  $. Let $S_2=\emptyset$. Now, for each clause $C_i=(\ell_i^1, \ell_i^2, \ell_i^3)$, do the following. 

\begin{enumerate}
\item If $\phi(\ell_i^1) = 1$, then put $S_2=S_2 \cup \set{\CGa{i},\gammaAND{\CGap{i}}{\Interval{\overbar{\ell_i^1}}}}$ and if $\phi(\ell_i^1) = 0$ then put $S_2=S_2 \cup \set{\CGap{i},\gammaAND{\CGa{i}}{\Interval{\ell_i^1}}}$.

\item If $\phi(\ell_i^2) = 1$, then put $S_2=S_2 \cup \set{\CGb{i},\gammaAND{\CGbp{i}}{\Interval{\overbar{\ell_i^2}}}}$ and if $\phi(\ell_i^2) = 0$ then put $S_2=S_2 \cup \set{\CGbp{i},\gammaAND{\CGb{i}}{\Interval{\ell_i^2}}}$.

\item If $\phi(\ell_i^3) = 1$, then put $S_2=S_2 \cup \set{\CGc{i},\gammaAND{\CGcp{i}}{\Interval{\overbar{\ell_i^3}}}}$ and if $\phi(\ell_i^3) = 0$ then put $S_2=S_2 \cup \set{\CGcp{i},\gammaAND{\CGc{i}}{\Interval{\ell_i^3}}}$.
\end{enumerate}


Let $S=S_1\cup S_2 \cup S_p$. We shall show that $S$ is a geodetic set of $D$. Due to Lemma~\ref{lem:uper-bound-interval}, we have that $|S_1 \cup S_2 \cup S_p| = \bound$.

As $S_p \subseteq S$, observe that $\startInterval\in S$ and for each track $T\in \mathcal{T}$, the point interval $\Tail{T}\in S$. Hence, due to Lemma~\ref{lem:track-covered-rev}, we have that each track $T\in \mathcal{T}$ is covered by $S$. 

Consider any variable gadget $\VarGadget_i$ corresponding to the variable $x_i$. Recall that from the construction, $\VarGadget_i = \Implication{\trueVertex}{x_i} \cup \Implication{x_i}{\overbar{x_i}}$. Due to Lemma~\ref{lem:imply-cover-rev}, we have that $\Interval{x_i} \in I(\trueInterval,\Interval{s_{x_i}})$ and $\Interval{r_{x_i}} \in I(\startInterval,\Interval{s_{x_i}})$. Hence, $\Implication{\trueVertex}{x_i} \subseteq I(S)$. Now, consider the implication gadget $\Implication{x_i}{\overbar{x_i}}$. Due to Lemma~\ref{lem:imply-cover-rev}, we have that $\overline{x_i} \in I(\startInterval,\Interval{s_{\overline{x_i}}})$. Since $\phi$ is a satisfying assignment, then either $\Interval{\overline{x_i}} \in S$ or $\Interval{x_i} \in S$. In the latter case, due to Lemma~\ref{lem:imply-cover-rev}, $\Interval{r_{\overline{x_i}}} \in I(\Interval{x_i}, \Interval{s_{\overline{x_i}}})$. Hence, $\Implication{x_i}{\overbar{x_i}} \subseteq S$, and therefore $\VarGadget_i \subseteq I(S)$.


Now, consider any clause $C_i = (\ell_i^1,\ell_i^2,\ell_i^3 )$ and recall the construction of $\ClauseGadget_i$. Since $\phi$ is a satisfying assignment, observe that, there exists at least one interval $\Interval{z} \in \set{\CGa{i},\CGb{i},\CGc{i}} \cap S$. Now due to Lemma~\ref{lem:cover-gadget-1-rev}, $\Interval{z}\in I(\CGd{i}, \Tail{T(\Interval{z})})$ and $\CGCov{i}\in I(\Interval{z},\CGf{i})$. Hence $\CovGadget{i} \subseteq I(S)$. Now, consider the implication gadget $\Implication{\CGa{i}}{\CGap{i}}$. From our definition of $S_2$, it follows that $\set{\CGa{i},\CGap{i}}\cap S \neq \emptyset$. Now due to Lemma~\ref{lem:imply-cover-rev}, $\Interval{\CGap{i}} \in I(S)$ and $\Interval{r_{\CGap{i}}} \in I(S)$. Hence, $\Implication{\CGa{i}}{\CGap{i}} \subseteq I(S)$. Repeating the above arguments for $\Implication{\CGb{i}}{\CGbp{i}}$ and $\Implication{\CGc{i}}{\CGcp{i}}$, we infer that $$ \Implication{\CGa{i}}{\CGap{i}} \cup \Implication{\CGb{i}}{\CGbp{i}} \cup \Implication{\CGc{i}}{\CGcp{i}} \subseteq I(S)$$ 

Now, consider the AND gadgets $\AND{\CGa{i}}{\ell_i^1}$ and $\AND{\CGap{i}}{\overline{\ell_i^1}}$. From our definition of $S_2$, it follows that either $\set{\CGa{i},\gammaAND{\CGap{i}}{\Interval{\overbar{\ell_i^1}}}} \subseteq S$ or $\set{\CGap{i},\gammaAND{\CGa{i}}{\Interval{\ell_i^1}}} \subseteq S$.
First consider the case when $\set{\CGa{i},\gammaAND{\CGap{i}}{\Interval{\overbar{\ell_i^1}}}} \subseteq S$. This means $\phi(\ell_i^1)=1$ and therefore $\Interval{\ell_i^1} \in S$. Now invoking Lemma~\ref{lem:AND-cover-1-rev}\ref{it:and-cover-1} and \ref{it:and-cover-2} (with $\Interval{p}=\Interval{\CGa{i}}$ and $\Interval{q}=\Interval{\ell_i^1}$), we have that $\AND{\CGa{i}}{\Interval{\ell_i^1}} \subseteq I(S)$. Also,  invoking Lemma~\ref{lem:AND-cover-1-rev}\ref{it:and-cover-1} and \ref{it:and-cover-3} ($\Interval{p}=\CGap{i}$ and $\Interval{q} = \Interval{\overline{\ell_i^1}}$) we have that $\AND{\CGap{i}}{\Interval{\overline{\ell_i^1}}} \subseteq I(S)$. The above arguments imply that when $\set{\CGa{i},\gammaAND{\CGap{i}}{\Interval{\overbar{\ell_i^1}}}} \subseteq S$, we have $$\AND{\CGa{i}}{\ell_i^1} \cup \AND{\CGap{i}}{\overline{\ell_i^1}} \subseteq I(S)$$

Similarly when  $\set{\CGap{i},\gammaAND{\CGa{i}}{\Interval{\ell_i^1}}} \subseteq S$ using similar arguments we have $\left(\AND{\CGa{i}}{\ell_i^1} \cup \AND{\CGap{i}}{\overline{\ell_i^1}}\right) \subseteq I(S)$. Now arguing similarly for $\left(\AND{\CGb{i}}{\ell_i^2} \cup \AND{\CGbp{i}}{\overline{\ell_i^2}} \right)$ and $\left(\AND{\CGc{i}}{\ell_i^3} \cup \AND{\CGcp{i}}{\overline{\ell_i^3}}\right)$, we have that  $\ClauseGadget_i\subseteq I(S)$. This completes the proof.
\end{proof}

\subsubsection{Optimality implies satisfiability}\label{sec:opt-sat}

Now, we shall show that if the geodetic number of $D$ is at most $\bound$, then $F$ is satisfiable. 

\begin{lemma}\label{lem:good-geodetic-set}
Let $S$ be a geodetic set of $D$. There exists a geodetic set $S^*$ of $D$ with $|S^*|\leq |S|$ such that for any track $T\in \mathcal{T}$ we have $S\cap T = \emptyset$.
\end{lemma} 

 \begin{proof}

Let $U$ consists of all intervals contained in some track. In other words, $U=\displaystyle\bigcup\limits_{T\in \mathcal{T}} T$. Let $S_p$ denote the set of point intervals in $D$. Observe that $S_p \subseteq S$. Let $\Interval{u} \in S \cap U$ be the interval with $min(S\cap U)=min(\Interval{u})$ such that $\Interval{u}$ belongs to a track with a root $\Interval{y}$. Let $A_1$ and $A_2$ be the sets of all $\gammaAND{p}{q}$ intervals and $\Interval{r}_q$ intervals in $D$, respectively. Let $A_3 = \set{\Interval{z}\in D \colon  \Interval{z}\in \set{\CGa{i},\CGb{i},\CGc{i}}, 1\leq i\leq m}$  and $A = A_1 \cup A_2 \cup A_3$ and $\overline{A}=D \setminus (A \cup U \cup S_p)$. Due to Lemma~\ref{lem:track-covered-rev}, \ref{lem:imply-cover-rev}, \ref{lem:cover-gadget-1-rev}, and \ref{lem:AND-cover-1-rev}, observe that, $A\cup U \subseteq I(S_p)$. Therefore, we will be done by proving the following claim.

 \begin{claim}
 Let $\Interval{v}$ be any interval in $S$. We have $I(\Interval{u},\Interval{v}) \cap \overline{A} \subseteq I(\Interval{y},\Interval{v}) \cap \overline{A}$.
\end{claim}

\medskip \noindent To prove the claim, define $S_1=\set{ \Interval{w}\in S \colon min(\Interval{w}) < min(\Interval{u}) }$ and $S_2=\set{\Interval{w}\in S \colon min(\Interval{u}) < min(\Interval{w})}$. First, assume that $\Interval{v} \in S_1$ and $\Interval{z}$ be an interval in $I(\Interval{u},\Interval{v})\cap \overline{A}$. In this case, $\Interval{v}$ must be a root of some track $T$ (by definition of $\Interval{u}$). Now, we have the following cases.
 \begin{enumerate}
     \item Assume $\Interval{z}=\Interval{q}$ for some implication gadget $\Implication{\Interval{p}}{\Interval{q}}$. Then by Lemma~\ref{lem:semi-good}, there exists a semi-good shortest path $Q$ between $\Interval{v}$ and $\Interval{u}$ containing $\Interval{q}$. From the construction of $\Implication{p}{q}$, this is only possible if $\Interval{v} = \Interval{p}$. By Lemma~\ref{lem:imply-cover-rev}, $\Interval{q}\in I(\Interval{p},\Interval{s_q})$. 
    
     \item Assume $\Interval{z} = \CGCov{i}$ for some $1\leq i \leq m$. Using similar arguments as above and Lemma~\ref{lem:cover-gadget-1-rev}, we can show that $\Interval{v} \in \set{\CGa{i},\CGb{i},\CGc{i}}$ and therefore $\Interval{z}\in I(\Interval{v}, \CGf{i})$.
    
    \item Assume $\Interval{z} \in \set{\alphaAND{p}{q} ,\deltaAND{p}{q}}$ for some AND gadget $\AND{p}{q}$. Using arguments as in Case $1$ and Lemma~\ref{lem:AND-cover-1-rev}, we can show that $\Interval{v} \in \set{\Interval{p},\Interval{q}}$ and therefore $\Interval{z}\in I(\Interval{v} \cup S_p)$. 
\end{enumerate}

The above cases imply that when $\Interval{v} < \Interval{u}$, then $I(\Interval{v},\Interval{u}) \subset I(\Interval{v},\Interval{u'})$, where $\Interval{u'}\in S \setminus \set{\Interval{u}}$.  Now assume that $\Interval{v} \in S_2$ and $\Interval{z} \in I(\Interval{u},\Interval{v})\cap \overline{A}$. Observe that there exists a shortest path from $\Interval{u}$ to $\Interval{v}$ that contains $\Interval{z}$. Consider now the good shortest path between $\Interval{y}$ and $\Interval{u}$ concatenated with the shortest path between $\Interval{u}$ and $\Interval{v}$ covering $\Interval{z}$. This is a shortest path between $\Interval{y}$ and $\Interval{w}$ covering $\Interval{z}$. This completes the proof of the claim.

\medskip \noindent The above arguments imply that $S^*=(S\setminus \set{\Interval{u}})\cup \set{\Interval{y}}$ is also a geodetic set of $D$. Arguing similarly for all intervals in $S\cap U$, we have the lemma.
\end{proof}

A \emph{good geodetic set} of $D$ is a geodetic set of minimum cardinality which does not contain any interval belonging to a track. By Lemma~\ref{lem:good-geodetic-set}, a good geodetic set of $D$ always exists. Now we shall prove some further properties of good geodetic sets of $D$.

\begin{lemma}\label{lem:implication-select}
Let $S^*$ be a good geodetic set of $D$ and $\Implication{p}{q}$ be an implication gadget where $\Interval{p}$ is the only root of $T(\Interval{p})$. Then, either $\Interval{p} \in S^*$ or $\Interval{q}\in S^*$.
\end{lemma}

\begin{proof}
 Suppose $\Interval{q}\notin S^*$. Then, Lemma~\ref{lem:semi-good} implies that there exists a semi-good shortest path $P$ containing $\Interval{q}$ whose end-vertices lie in $S^*$. Then, there must exist two intervals $\Interval{u_1},\Interval{v_1}$ in $P$ such that both $\Interval{u_1},\Interval{v_1}$ intersect $\Interval{q}$ and $\max{(\Interval{u_1})}<\min{(\Interval{v_1})}$. From the construction of $\Implication{p}{q}$, it follows that $\Interval{u_1}\in T(\Interval{p})$. Since $S^*$ is a good geodetic set, our construction implies that $\Interval{p}\in S^*$. 
\end{proof}

\begin{lemma}\label{lem:AND-select}
Let $S^*$ be a good geodetic set of $D$ and $\AND{p}{q}$ be an AND gadget where $\Interval{p}$ is the only root of $T(\Interval{p})$ and $\Interval{q}$ is the only root of $T(\Interval{q})$. Then, either $\set{\Interval{p},\Interval{q}} \subseteq S^*$, or $S^*$ contains at least one interval among $\set{\alphaAND{p}{q},\gammaAND{p}{q},\deltaAND{p}{q}}$.
\end{lemma}

\begin{proof}
Suppose that $S^*\cap \set{\alphaAND{p}{q},\gammaAND{p}{q},\deltaAND{p}{q}}=\emptyset$. Then, Lemma~\ref{lem:semi-good} implies that there exists a shortest path $P$ containing $\alphaAND{p}{q}$ whose end-vertices lie in $S^*$. Then, there must exist two intervals $\Interval{u_1},\Interval{v_1}$ in $P$ such that both $\Interval{u_1},\Interval{v_1}$ intersect $\alphaAND{p}{q}$ and $\max{(\Interval{u_1})}<\min{(\Interval{v_1})}$. From the construction of $\AND{p}{q}$ it follows that $\Interval{u_1}\in T(\Interval{p})$. Since $S^*$ is a good geodetic set, our construction implies that $\Interval{p}\in S^*$. 

Now, Lemma~\ref{lem:semi-good} implies that there exists a semi-good shortest path $Q$ containing $\deltaAND{p}{q}$, whose end-vertices lies in $S^*$. Then, there must exist two intervals $\Interval{u_2},\Interval{v_2}$ in $Q$ such that both $\Interval{u_2},\Interval{v_2}$ intersect $\deltaAND{p}{q}$ and $\max{\Interval{u_2}}<\min{\Interval{v_2}}$. From the construction of $\AND{p}{q}$, it follows that either $\Interval{u_2}\in T(\Interval{q})$ or $\Interval{u_1}=\gammaAND{p}{q}$. 

Consider the case when $\Interval{u_1}=\gammaAND{p}{q}$. Since $\gammaAND{p}{q} \notin S^*$, there must exist an interval $\Interval{w_2}$ in $Q$ such that $\gammaAND{p}{q}$ intersects $\Interval{w_2}$ and the distance between $\Interval{w_2}$ and $\Interval{v_2}$ is exactly three. But again from the construction of $\AND{p}{q}$, it follow that no such $\Interval{w_2}$ exists, leading to a contradiction. Hence, $\Interval{u_2}\in T(\Interval{q})$. Since $S^*$ is a good geodetic set, our construction implies that $\Interval{q}\in S^*$. 
\end{proof}

\begin{lemma}\label{lem:clause-select}
Let $S^*$ be a good geodetic set of $D$ and let $C_i=(\ell_i^1,\ell_i^2,\ell^3_i)$ be a clause. Then we have 

\begin{enumerate}
\renewcommand{\theenumi}{(\alph{enumi})}
\renewcommand{\labelenumi}{\theenumi}
\item $\abs{ S^* \cap \set{\CGa{i},\CGap{i}, \alphaAND{\CGa{i}}{\ell^1_i},\gammaAND{\CGa{i}}{\ell^1_i},\deltaAND{\CGa{i}}{\ell^1_i}, \alphaAND{\CGap{i}}{\overbar{\ell^1_i}},\gammaAND{\CGap{i}}{\overbar{\ell^1_i}},\deltaAND{\CGap{i}}{\overbar{\ell^1_i}}}}  \geq 2$,

\item $\abs{ S^* \cap \set{\CGb{i},\CGbp{i}, \alphaAND{\CGb{i}}{\ell^2_i},\gammaAND{\CGb{i}}{\ell^2_i},\deltaAND{\CGb{i}}{\ell^2_i}, \alphaAND{\CGbp{i}}{\overbar{\ell^2_i}},\gammaAND{\CGbp{i}}{\overbar{\ell^2_i}},\deltaAND{\CGbp{i}}{\overbar{\ell^2_i}}}}  \geq 2$, and

\item $\abs{ S^* \cap \set{\CGc{i},\CGcp{i}, \alphaAND{\CGc{i}}{\ell^3_i},\gammaAND{\CGc{i}}{\ell^3_i},\deltaAND{\CGc{i}}{\ell^3_i}, \alphaAND{\CGcp{i}}{\overbar{\ell^3_i}},\gammaAND{\CGcp{i}}{\overbar{\ell^3_i}},\deltaAND{\CGcp{i}}{\overbar{\ell^3_i}}}}  \geq 2$.
\end{enumerate}  
\end{lemma}

\begin{proof}
 Recall that $\ClauseGadget{i}$ contains the gadgets $\AND{\CGa{i}}{\ell_i^1}$ and $\AND{\CGap{i}}{\overbar{\ell_i^1}}$ (along with some other gadegets). By Lemma~\ref{lem:AND-select}, if $\CGa{i} \notin S^*$ then we need at least one vertex among  $\alphaAND{\CGa{i}}{\ell^1_i},\gammaAND{\CGa{i}}{\ell^1_i},\deltaAND{\CGa{i}}{\ell^1_i}$.  The same holds when $\CGa{i} \notin S^*$, we need at least one among $\alphaAND{\CGap{i}}{\overbar{\ell^1_i}},\gammaAND{\CGap{i}}{\overbar{\ell^1_i}},\deltaAND{\CGap{i}}{\overbar{\ell^1_i}}$. This implies that (a) holds. Using analogous arguments we can show that $(b)$ and $(c)$ holds.
\end{proof}

\begin{lemma}\label{lem:clause-select-1}
Let $S^*$ be a good geodetic set of $D$ and $C_i=(\ell_i^1,\ell_i^2,\ell^3_i)$ be a clause. If none of $\Interval{\ell^i_1},\Interval{\ell^i_2},\Interval{\ell^i_3}$ is in $S^*$ then $|S^* \cap \ClauseGadget_i|\geq 7$.
\end{lemma}

\begin{proof}
    Assume none of $\Interval{\ell^i_1},\Interval{\ell^i_2},\Interval{\ell^i_3}$ is in $S^*$. Since $\Interval{\ell^1_i}\notin S^*$, due to Lemma~\ref{lem:AND-select}, we have that at least one among $\alphaAND{\CGa{i}}{\ell^1_i}, \gammaAND{\CGa{i}}{\ell^1_i}, \deltaAND{\CGa{i}}{\ell^1_i}$ lies in $S^*$. If $\CGap{i}\notin S^*$, then $\CGa{i}\in S^*$ and one more interval from  $\alphaAND{\CGap{i}}{\overbar{\ell^1_i}}, \gammaAND{\CGap{i}}{\overbar{\ell^1_i}},  \deltaAND{\CGap{i}}{\overbar{\ell^1_i}}$ lies in $S^*$ (Lemma~\ref{lem:implication-select} and~\ref{lem:AND-select}). Therefore, $$\abs{ S^* \cap \set{\CGa{i},\CGap{i}, \alphaAND{\CGa{i}}{\ell^1_i},\gammaAND{\CGa{i}}{\ell^1_i},\deltaAND{\CGa{i}}{\ell^1_i}, \alphaAND{\CGap{i}}{\overbar{\ell^1_i}},\gammaAND{\CGap{i}}{\overbar{\ell^1_i}},\deltaAND{\CGap{i}}{\overbar{\ell^1_i}}}}  \geq 3$$

  Now, using Lemma~\ref{lem:clause-select}, we also have $$\abs{ S^* \cap \set{\CGb{i},\CGbp{i}, \alphaAND{\CGb{i}}{\ell^2_i},\gammaAND{\CGb{i}}{\ell^2_i},\deltaAND{\CGb{i}}{\ell^2_i}, \alphaAND{\CGbp{i}}{\overbar{\ell^2_i}},\gammaAND{\CGbp{i}}{\overbar{\ell^2_i}},\deltaAND{\CGbp{i}}{\overbar{\ell^2_i}}}}  \geq 2$$ and $$\abs{ S^* \cap \set{\CGc{i},\CGcp{i}, \alphaAND{\CGc{i}}{\ell^3_i},\gammaAND{\CGc{i}}{\ell^3_i},\deltaAND{\CGc{i}}{\ell^3_i}, \alphaAND{\CGcp{i}}{\overbar{\ell^3_i}},\gammaAND{\CGcp{i}}{\overbar{\ell^3_i}},\deltaAND{\CGcp{i}}{\overbar{\ell^3_i}}}}  \geq 2$$

  The above arguments imply that when $\CGap{i}\notin S^*$, $|S^* \cap \ClauseGadget_i|\geq 7$. Arguing similarly as above, we can show that if at least one of $\CGbp{i},\CGcp{i}$ does not belong to $S^*$, then also $|S^* \cap \ClauseGadget_i|\geq 7$. Now consider the case, when $\set{\CGap{i},\CGbp{i},\CGcp{i}}\subset S^*$. Moreover, $S^*$ contains one interval from each of $\AND{\CGa{i}}{\ell^1_i}$, $\AND{\CGb{i}}{\ell^2_i}$ and $\AND{\CGc{i}}{\ell^3_i}$. Now, we will be done by showing that at least one of $\CGa{i},\CGb{i},\CGc{i}, \CGCov{i}$ must be in $S^*$. 

  Suppose that $\CGCov{i}\notin S^*$ and let $ \CGCov{i} \in I(\Interval{u},\Interval{v})$, where $\Interval{u},\Interval{v}\in S^*$ and $min(\Interval{u}) < min(\Interval{v})$.  Let $P$ be a shortest path between $\Interval{u}$ and $\Interval{v}$ such that $P$ contains $\CGCov{i}$. Then, $P$ must contain two distinct intervals $\Interval{w_1}$ and $\Interval{w_2}$ such that both $\Interval{w_1},\Interval{w_2}$ intersects $\CGCov{i}$ and $\max(\Interval{w_1}) < \min(\Interval{w_2})$. From construction of $\CovGadget{i}$, it follows that $\Interval{w_1}$ lies in $T(\Interval{z})$, where $\Interval{z}\in \set{\CGa{i},\CGb{i},\CGc{i}}$. Now, since $S^*$ is a good geodetic set, we infer that $S^*$ contains at least one of $\CGa{i},\CGb{i},\CGc{i}$. Hence $S^*$ contains at least seven intervals from $\ClauseGadget_i$.
\end{proof}

\begin{lemma}\label{lem:opt-sat}
If there is a geodetic set of $D$ with cardinality $\bound$, then $F$ is satisfiable.
\end{lemma}

\begin{proof}
Let $S$ be a geodetic set of $D$ with cardinality $\bound$. Due to Lemma~\ref{lem:good-geodetic-set}, there exists a good geodetic set $S^*$ of $D$ with $|S^*|\leq |S| \leq \bound$. Recall that a variable gadget $\VarGadget_i = \Implication{\trueVertex}{x_i} \cup \Implication{x_i}{\overbar{x_i}}$. Due to Lemma~\ref{lem:implication-select}, we know that at least one among $\set{\Interval{x_i},\Interval{\overbar{x_i}}}$ lies in $S^*$. Let $S_1 = S^* \cap \parentesis{\bigcup\limits_{1\leq i\leq n} \VarGadget_i}$, and $S_2 = S^* \cap \parentesis{\bigcup\limits_{1\leq i\leq m} \ClauseGadget_i}$. Let $S_p$ denote the set of point intervals in $D$. Note that $S_p \subseteq S^*$. We have $|S_1|\geq n$,  $|S_2|\geq 6m$ by Lemma~\ref{lem:clause-select}, and $|S_p|=\nbSimplicial$. 
Therefore, $|S_1|=n$ as $\abs{S^*} \leq \bound$. This means that for each $1\leq i\leq n$, exactly one of $\Interval{x_i},\Interval{\overbar{x_i}}$ lies in $S^*$. Based on these, we define the following truth assignment $\phi\colon \set{x_1,\ldots,x_n}\rightarrow \set{1,0}$ of $F$. Define $\phi(x_i)=1$ if $\Interval{x_i}\in S^*$ and $\phi(x_i)=0$, otherwise. Using Lemma~\ref{lem:clause-select} we can infer that for each $1\leq i\leq m$, we have that $|S^*\cap \ClauseGadget_i| = 6$. Due to Lemma~\ref{lem:clause-select-1}, at least one of the intervals $\Interval{\ell^i_1},\Interval{\ell^i_2},\Interval{\ell^i_3}$ lies in $S^*$. Thus, for at least one literal $\ell_i^j$, we have that $\phi(\ell_i^j)=1$, as needed.
\end{proof}

\subsubsection{Completion of Proof of Theorem~\ref{thm:interval-hard}}\label{sec:complete}

Lemma~\ref{lem:uper-bound-interval},~\ref{lem:sat-opt}, and~\ref{lem:opt-sat} completes the proof of Theorem~\ref{thm:interval-hard}.

\section{Conclusion}\label{sec:conclusion}


We proved that \textsc{Minimum Geodetic Set} is FPT on chordal graphs when parameterized by the clique number/treewidth, and that \textsc{Minimum Geodetic Set} is NP-hard on interval graphs.

An interesting question is whether there are FPT algorithms for \textsc{Minimum Geodetic Set} on interval or chordal graphs, when parameterized by the geodetic number? Are there constant-factor approximation algorithms for these classes? (These are not true for general graphs, see~\cite{kellerhals2020} and~\cite{caldam2020}.)

Assuming the \emph{Exponential Time Hypothesis}, our reduction implies that there cannot be a $2^{o(\sqrt{n})}$ time algorithm for \textsc{Minimum Geodetic Set} on interval graphs of order $n$. Are there subexponential time algorithms for \textsc{Minimum Geodetic Set} on interval graphs or chordal graphs, matching this lower bound? (This is the case for many graph problems on geometric intersection graphs, see~\cite{stringETH}.)

We have seen that for every $k$, \textsc{Minimum Geodetic Set} is solvable in time $f(k)n$ for $k$-trees (which are chordal graphs with clique number $k+1$), but such a running time is unlikely to be possible for \emph{partial} $k$-trees (i.e. graphs of treewidth $k$), since \textsc{Minimum Geodetic Set} is known to be NP-complete for graphs of treewidth~14~\cite{DBLP:conf/iwpec/Tale25}. 
However, it is unknown whether \textsc{Minimum Geodetic Set} is solvable in polynomial time on partial $2$-trees (also known as graphs of treewidth at most~2, series-parallel graphs, and $K_4$-minor-free graphs): an algorithm is known only for the special case of outerplanar graphs~\cite{mezzini2018}.

Finally, we think that studying the computational complexities of related problems like \textsc{Isometric Path Cover}~\cite{chakraborty2022}, \textsc{Strong Geodetic Set}~\cite{DBLP:journals/rairo/LimaSSU24}, \textsc{Geodetic Hull}~\cite{hernando2005steiner} on interval graphs is another interesting direction of research.

\medskip \noindent \textbf{Conflict of Interest:} The authors declare that there is no conflict of interest. 
\bibliographystyle{alpha}
\bibliography{references}

\end{document}